\documentclass[acmsmall]{acmart}

\usepackage{mathrsfs}
\usepackage{bbm}
\usepackage{wrapfig}
\usepackage{lipsum}
\usepackage{footmisc}

\usepackage{diagbox}

\usepackage{algorithm}
\usepackage{algorithmicx}
\usepackage[noend]{algpseudocode}

\usepackage{bm}
\usepackage{graphicx}
\usepackage{subfigure}
\usepackage{multirow} 
\usepackage{makecell}

\usepackage[normalem]{ulem}
\usepackage{enumitem}

\theoremstyle{theorem}
\newtheorem{theorem}{Theorem}

\theoremstyle{definition}
\newtheorem{definition}{Definition}
\theoremstyle{problem}

\usepackage{color}
\usepackage{bbding}
\usepackage{array}
\newcolumntype{C}[1]{>{\footnotesize}p{#1}l}
\newcolumntype{T}[1]{>{\footnotesize}p{#1}l}

\usepackage{tikz}
\usetikzlibrary{decorations.pathreplacing,calc}

\usepackage{pifont}

\AtBeginDocument{%
	}
\AtBeginDocument{%
  }

\setcopyright{acmlicensed}
\copyrightyear{2024}
\acmYear{2024}
\acmDOI{XXXXXXX.XXXXXXX}

\acmJournal{PACMMOD}
\acmVolume{37}
\acmNumber{4}
\acmArticle{111}
\acmMonth{6}

\begin{document}

\title{Play like a Vertex: A Stackelberg Game Approach for Streaming Graph Partitioning}


\newcommand{\sharednote}{\thanks{Contributed equally to this work.}\textsuperscript{,}}
\author{Zezhong Ding}
\authornote{School of Data Science.}
\email{zezhongding@mail.ustc.edu.cn}
\orcid{0000-0002-6286-8679}
\author{Yongan Xiang}
\authornote{School of Computer Science and Technology.}
\email{ya_xiang@mail.ustc.edu.cn}
\orcid{0009-0001-0572-5479}
\author{Shangyou Wang}
\authornotemark[1]
\email{wash_you@mail.ustc.edu.cn}
\orcid{0009-0001-5738-985X}
\author{Xike Xie}
\authornotemark[1]
\authornote{School of Biomedical Engineering.}
\email{xkxie@ustc.edu.cn}
\orcid{0000-0001-5290-5408}

\author{S. Kevin Zhou}
\authornotemark[3]
\email{s.kevin.zhou@gmail.com}
\orcid{0000-0002-6881-4444}

\affiliation{%
  \institution{Data Darkness Lab, MIRACLE Center, Suzhou Institute for Advanced Research, University of Science and Technology of China}
  \city{Suzhou}
  \state{Jiangsu}
  \country{China}
}

\renewcommand{\shortauthors}{Zezhong Ding et al.}

\begin{abstract}
In the realm of distributed systems tasked with managing and processing large-scale graph-structured data, optimizing graph partitioning stands as a pivotal challenge. The primary goal is to minimize communication overhead and runtime cost. However, alongside the computational complexity associated with optimal graph partitioning, a critical factor to consider is memory overhead. Real-world graphs often reach colossal sizes, making it impractical and economically unviable to load the entire graph into memory for partitioning. This is also a fundamental premise in distributed graph processing, where accommodating a graph with non-distributed systems is unattainable. Currently, existing streaming partitioning algorithms exhibit a skew-oblivious nature, yielding satisfactory partitioning results exclusively for specific graph types.
In this paper, we propose a novel streaming partitioning algorithm, the \emph{\textbf{S}kewness-aware \textbf{V}ertex-cut \textbf{P}artitioner} (\emph{S5P}), designed to leverage the skewness characteristics of real graphs for achieving high-quality partitioning. S5P offers high partitioning quality by segregating the graph's edge set into two subsets, {\it head} and {\it tail} sets.
Following processing by a skewness-aware clustering algorithm, these two subsets subsequently undergo a Stackelberg graph game.
Our extensive evaluations conducted on substantial real-world and synthetic graphs demonstrate that, in all instances, the partitioning quality of S5P surpasses that of existing streaming partitioning algorithms, operating within the same load balance constraints.
For example, S5P can bring up to a 51\% improvement in partitioning quality compared to the top partitioner among the baselines. 
Lastly, we showcase that the implementation of S5P results in up to an $81$\% reduction in communication cost and a $130\%$ increase in runtime efficiency for distributed graph processing tasks on PowerGraph.
\end{abstract}



\received{15 October 2023}
\received[revised]{20 January 2024}
\received[accepted]{23 February 2024}

\maketitle

\section{Introduction}
\label{sec:introduction}
Graph partitioning is a crucial technique for managing large-scale graph analytics in a distributed manner. In this process, an input graph is partitioned across a cluster of machines. This enables the handling of graphs that are too large for a single machine while accelerating computation through parallelization. This trend is evident in the rise of distributed graph systems, e.g., Blogel\cite{blogel}, PowerGraph\cite{powergraph}, GraphX\cite{graphx}, PowerLyra\cite{powerlyra}, TopoX\cite{TopoX}, EASE\cite{MerkelMFJ23}, ScaleG\cite{ScaleG}, HCPD\cite{hcpd},  Pregel\cite{pregel}, AGP\cite{FanXYYZ23},  G-thinker\cite{Gthinker}, and GridGraph\cite{ZhuHC15}.
Recent advances in graph neural networks further underscore the significance of graph partitioning, e.g., BNS-GCN\cite{WanLLKL22}, AliGraph\cite{ZhuZYLZALZ19}, Betty\cite{yang2023betty}, and DistGNN\cite{distgcn}.

Using the GAS (\emph{Gather-Apply-Scatter}) model\footnote{
	Other programming models in distributed graph processing, such as BSP~\cite{pregel} and asynchronous model~\cite{graphlab}, consist of computation, communication, and synchronization phases, which
	can be viewed as different orders of API calls of GAS.
} as an example in the context of distributed graph processing, vertices in a distributed system iteratively gather information from their neighbors, apply computation to the data, and then scatter the results back to their neighbors.
At each iteration, a vertex collects messages sent by its replicas maintained in other machines (or partitions) and organizes them for synchronization.
The primary objective of graph partitioning is to minimize the number of vertices (or edges) that are cut and replicated among different partitions. This minimization helps significantly reduce the communication and synchronization overhead in distributed graph tasks, as each cut results in message transfers between different partitions. 

\begin{figure}[t]
	\centering
		{\includegraphics[width= 0.26\columnwidth]{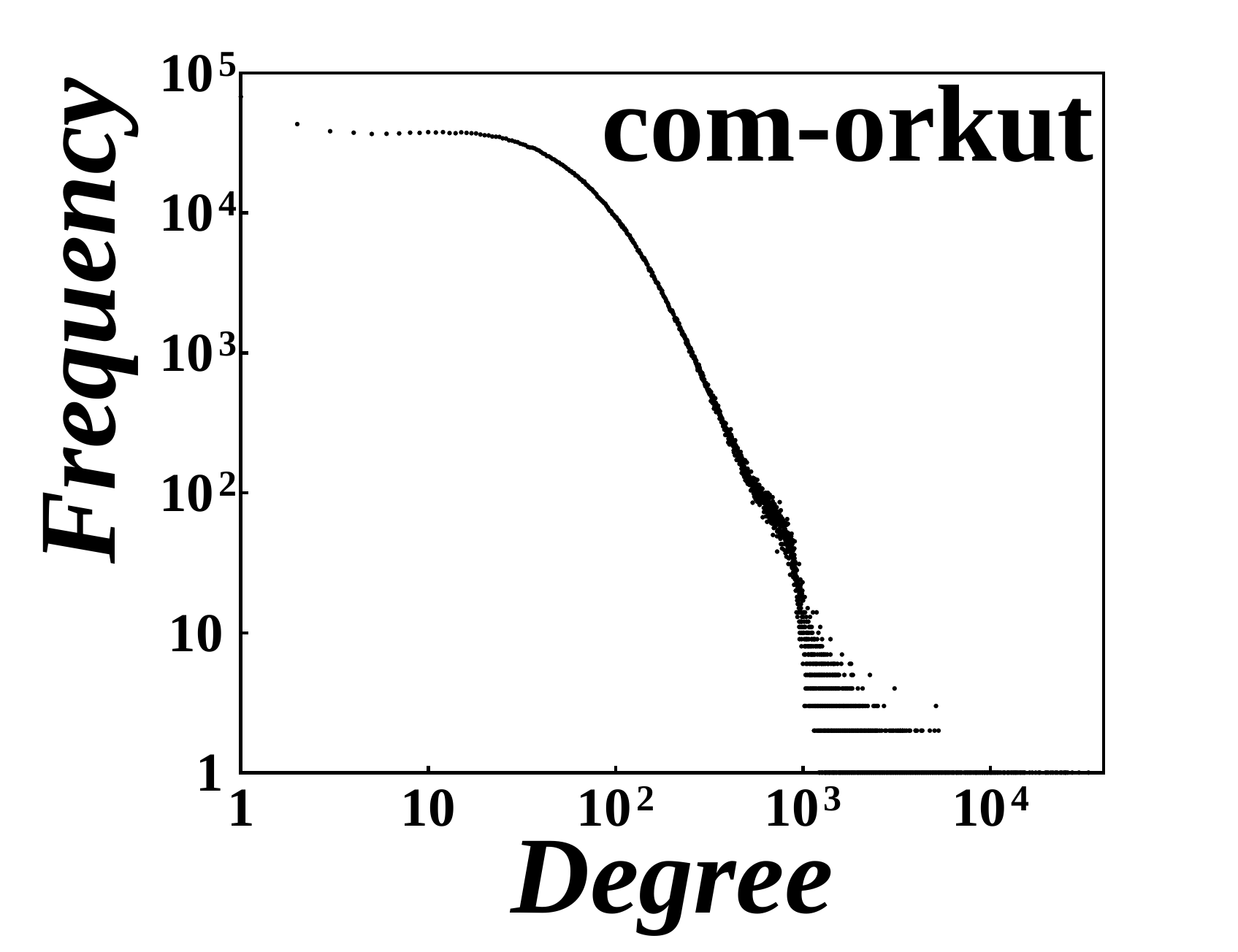}}
		{\includegraphics[width= 0.22\columnwidth]{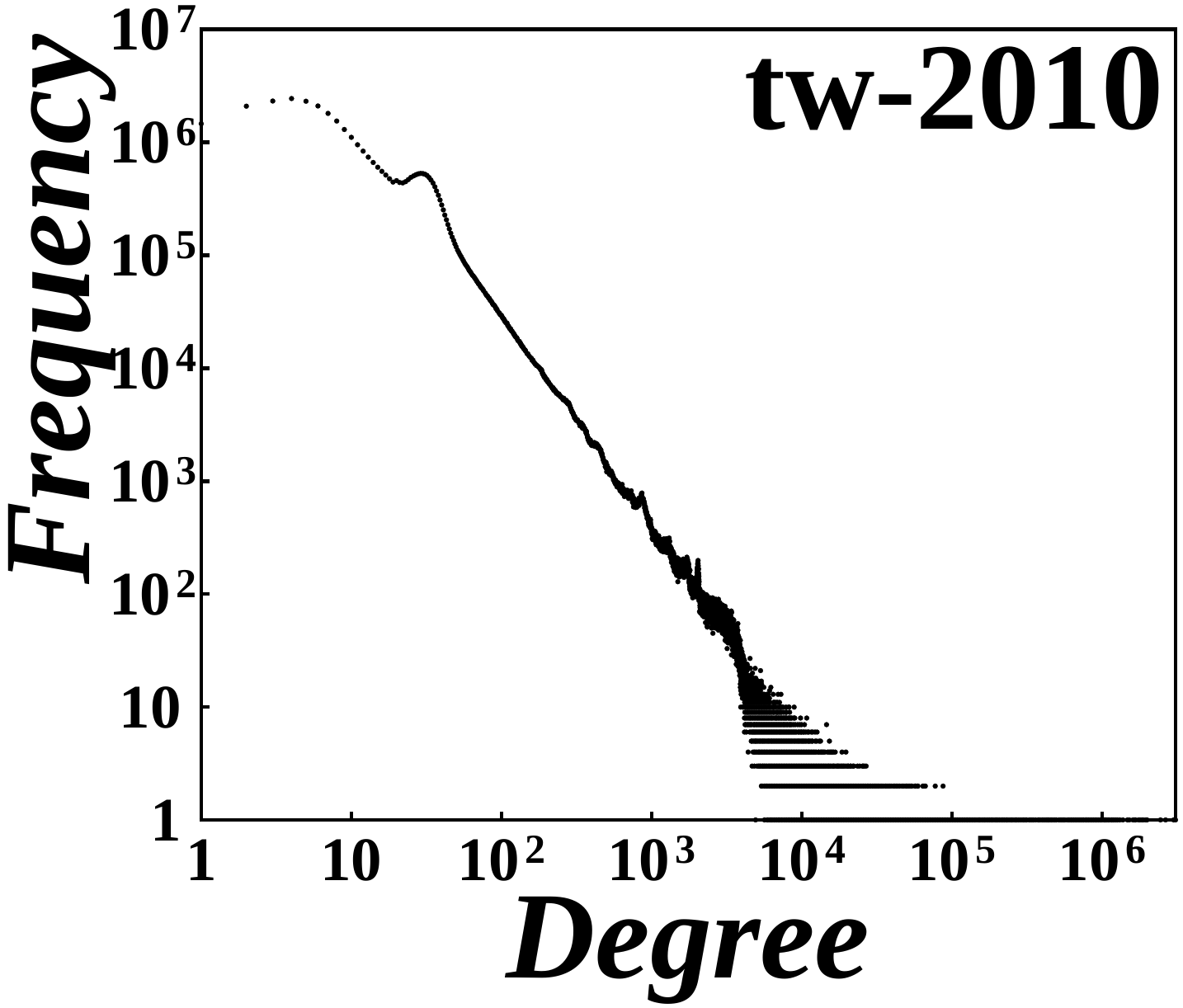}}
		{\includegraphics[width= 0.22\columnwidth]{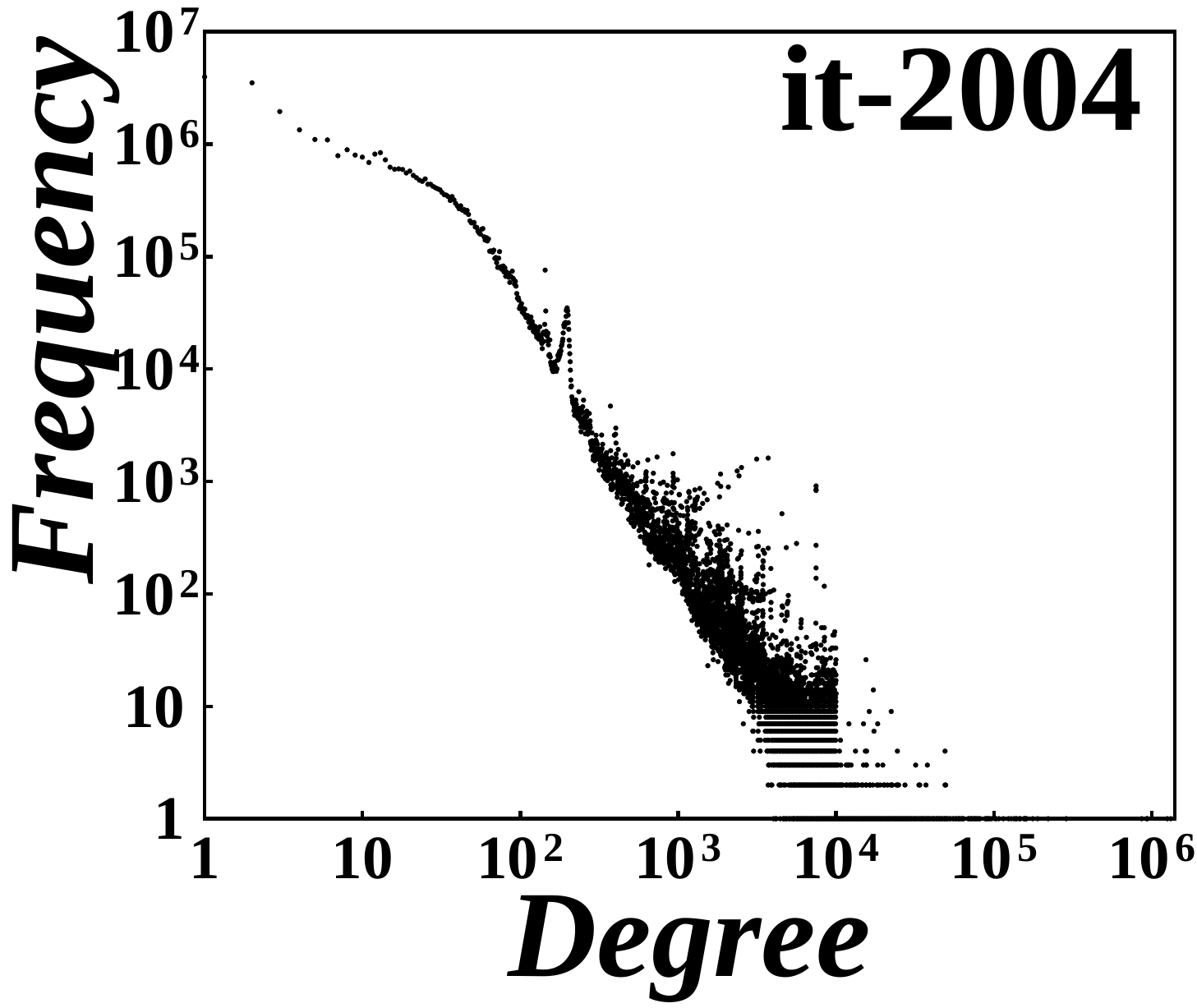}}
		{\includegraphics[width= 0.22\columnwidth]{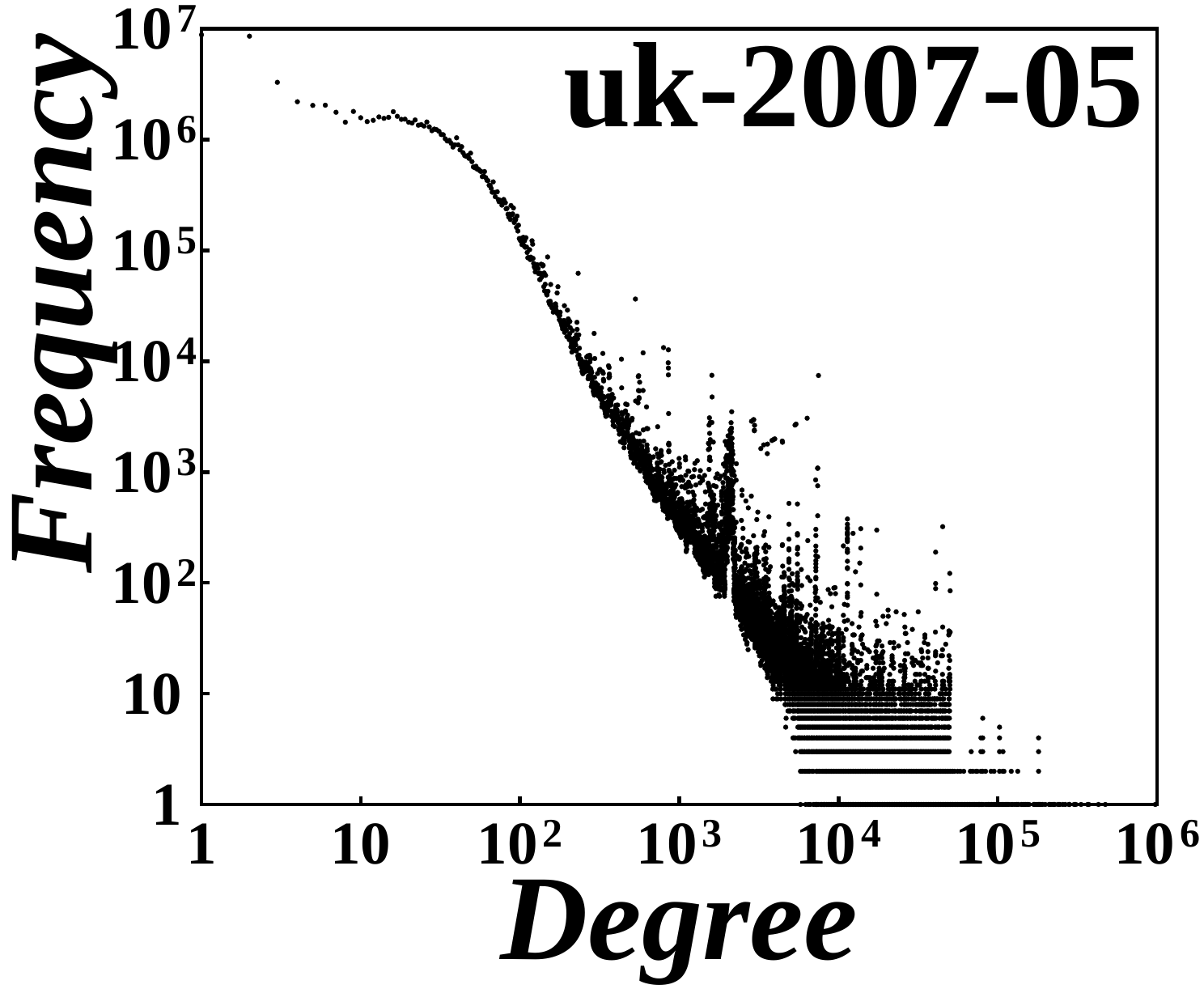}}
		\caption{ Distributions about Graph Skewness }
		\label{fig:skewdegree}
\end{figure}

The evolution of graph partitioning algorithms has unfolded across two pivotal dimensions: the transition from {\it edge-cut} to {\it vertex-cut} partitioning strategies, and the shift from {\it offline} to {\it streaming} methodologies.

 {Basically, edge-cut partitioning involves cutting edges to distribute vertices across partitions, minimizing cross-partition edges, while vertex-cut partitioning achieves the same by cutting vertices into multiple replicas.}
{Traditional edge-cut algorithms, effective in non-skewed graph scenarios, confront challenges in dealing with real-world power-law graphs \cite{powergraph, bgep, IAA, albert2000error}, where vertex degrees follow a long-tail distribution, yielding a large number of low-degree vertices in the tail and a small number of very high-degree vertices in the head \cite{tailgnn}, as depicted in Figure~\ref{fig:skewdegree}.
The skewness challenge drove the adoption of vertex-cut techniques, which excel in addressing skewed distributions by selectively replicating high-degree vertices \cite{ne}.}

Simultaneously, scaling partitioning algorithms to massive graphs prompts a shift from offline methods that preload the entire graph, to streaming methods that access and process edges in a sequence.
However, early streaming partitioning algorithms, e.g., \cite{dbh} \cite{hdrf}, despite the advantage of scalability, tend to be agnostic to skewness, treating all parts of a graph uniformly. This leads to inferior partitioning quality on skewed graphs.
Recent research trends as seen in \cite{2ps-l} and \cite{clugp} share a common improvement theme, which unconsciously addresses the underutilization of skewness through graph clustering and improves the performance.
However, this implicit addressing of skewness is only marginal and not substantial,
leaving room for further performance enhancement over skewed graphs.
In fact, a low-degree vertex in the tail has a marginal effect on the partitioning quality, whereas a high-degree vertex in the head exerts a substantial influence (cf. Figure~\ref{fig:skewdegree}). Therefore, it is crucial to strike a balance between the two sets and their mutual impact on the partitioning quality.

In this paper, we study a skewness-aware vertex-cut streaming partitioning approach, which employs a Stackelberg game for addressing the mutual effect of high- and low-degree vertices in partitioning quality.
Considering the partitioning process as a game of allocation, where players represent graph entities, e.g., vertices, edges, or subgraphs, to a predetermined set of partitions, we posit that the Stackelberg game effectively encapsulates the skewness characteristics.
In general, the Stackelberg game embodies the hierarchical interaction among players taking different roles, akin to the hierarchical structure seen in an ant colony, with its queen and worker ants.
Following a similar line of thought,
we can also assign different roles to graph entities, e.g., vertices, fostering interactions through game-theoretic dynamics.
High-degree vertices assume leadership roles much like queen ants, while low-degree vertices take on follower roles reminiscent of worker ants. Real-world graphs often exhibit a significant preponderance of low-degree vertices, analogous to the higher number of worker ants in an ant colony than queen ants. Modeling the dynamic interplay between the two groups through the lens of a Stackelberg game effectively captures the asymmetric influence of high-degree leader vertices on their low-degree follower counterparts. 

Still, in the context of streaming partitioning, efficiently extracting skewness-related information for each vertex to address graph skewness presents technical challenges that demand both low memory consumption and high-speed processing.

To address the challenge, our techniques are three-fold. First, we investigate a skewness-aware clustering algorithm, seamlessly integrated with the Stackelberg game to constitute a {\it clustering-refinement} framework for graph partitioning, aligning with the cutting-edge framework of stream partitioners. Second, we use the count-min sketch for summarizing graph information to construct compact strategy sets for different players with high time and space efficiency yet without compromising the quality.
Third, we exploit parallelization for accelerating the game process, enabling the concurrent operations among graph entities participated in the partitioning process, akin to a colony of ants working towards a common goal.

Our contributions can be listed as follows.

\begin{itemize}
\item We have spearheaded the examination of the profound connection between skewness and the graph partitioning quality, based on which we propose the first Stackelberg game-based approach to tackle the challenges of streaming graph partitioning.
\item
Taking graph skewness into consideration, we first propose skewness-aware graph clustering. Oriented at vertex-cut partitioning, this approach brings in over 16$\times$ space efficiency and over 8$\times$ time efficiency (Section~\ref{sec:ca}).
Furthermore, it offers excellent versatility and can be seamlessly integrated into other partitioning algorithms, conforming to the clustering-refinement framework.

\item
  We are the first to introduce the Stackelberg game from the perspective of graph skewness into the solution of vertex-cut partitioning problem. We offer relevant theoretical analysis, specifically in the context of graph partitioning.
  Additionally, we have incorporated sketching and parallelization techniques to streamline and optimize the time and space complexities associated with the game process.
\item Extensive experiments are conducted on $11$ real and $6$ synthetic graphs, varying in scales and types. The experimental results show that S5P outperforms state-of-the-art solutions in terms of quality, especially for social and web graphs.
      Moreover, the deployment of partitioning algorithms to real distributed graph systems, e.g., Powergraph, showcases an enhancement in executing graph tasks, e.g., PageRank, with up to an $81$\% reduction in communication cost and $130\%$ increase in runtime efficiency.

  \end{itemize}

The subsequent sections of this paper are structured as follows.
In Section \ref{sec:preliminaries}, we commence by formalizing and delving into the vertex-cut partitioning problem, and present preliminaries of the two-stage Stackelberg game  and graph skewness.
Section \ref{sec:relatedwork} engages in a discussion of related works.
Section \ref{sec:method} is dedicated to introducing the S5P framework, encompassing the innovative streaming skewness-aware graph clustering algorithm, the two-stage Stackelberg-based partitioning algorithm, and the optimization strategy.
In Section \ref{sec:theory}, we present a comprehensive theoretical analysis.
Section \ref{sec:experiment} evaluates the performance of S5P against {8} robust baselines under {11} real-world graphs and {6 synthetic graphs}.
Section~\ref{sec:conclusion} concludes the paper.

\section{preliminaries}
\label{sec:preliminaries}
\subsection{Vertex-Cut Streaming Partitioning}
\subsubsection*{Vertex-Cut Graph Partitioning}

Give a directed or undirected graph $G=(V, E)$, where $V$ is the set of vertices and $E \subseteq V \times V$ is the set of edges, and given a number of $k$ partitions $P=\{p_i\}_{i \leq k}$, the vertex-cut partitioning algorithm exclusively assigns each edge $e_i \in G$ to a partition $p_i \in P$, satisfying $\cup_{i \leq k}p_i=E$ and $p_i \cap p_j = \emptyset (i \neq j)$.

\subsubsection*{Vertex-Cut Streaming Partitioning} The stream edge model $S_G = \{e_1, e_2, ..., e_{|E|}\}$ assumes the edges of an input graph $G=(V, E)$ arrive in a sequence.
Then, the vertex-cut (edge) streaming partitioning performs single or multiple passes over the graph stream and makes partitioning decisions.

\subsubsection*{Quality of Vertex-cut Partitioning}
The primary goal of the partitioning algorithm is to optimize the efficiency of upper-level distributed graph systems.
So, the \emph{replication factor} (\emph{RF}) serves as the standard metric for evaluating partitioning quality, as shown in Equation~(\ref{eqn:rf}).

\begin{equation}
\label{eqn:rf}
\small
  \text{RF} = \underbrace{\frac{\sum_{v\in V}{|P(v)|}}{|V|}}_{\textit{Vertex Replication Form}} = \underbrace{\frac{\sum_{d_{i}}{g(d_{i})f(d_{i})}}{|V|}}_{\textit{ Degree Distribution Form}}
\end{equation}
, where $P(v)$ represents the set of partitions holding vertex $v$, and $|P(v)|$ signifies the number of partitions holding $v$, a.k.a., the number of replications of $v$.

We can also rewrite the equation of replication factor from the {\it vertex replication form} into the {\it degree distribution form}, as shown in Equation~(\ref{eqn:rf}), where $g(d_{i}) = \sum_{d(v) = d_{i}}|P(v)|$ represents the average replication of all the vertices with degree $d_i$, and $f(d_{i})$ represents the frequency of degree $d_{i}$.
 Thus, RF can be viewed as the result of integrating $g(d_i)$ across various degrees.
So, the development of partitioning strategies that take into account the skewness in distributions holds the promise of effectively reducing RFs.

With Equation~(\ref{eqn:rf}), the vertex-cut partitioning problem can be formalized as an optimization problem~\cite{powergraph}, which is known to be NP-hard~\cite{nphard, FeigeHL08}, as follows.
\begin{equation}
	minimize \ \text{RF} \ \ s.t. \ \frac{k \max_{i=1,...k}|P_i|}{|E|} \leq \tau
	\label{eq:youhua}
\end{equation}
Here, $\tau$ is the threshold of imbalance.

\subsection{Stackelberg Game}
The \emph{Stackelberg game} can be categorized into single-stage and multi-stage versions. In a single-stage scenario, leaders and followers make decisions simultaneously and adopt a unique optimization function. In contrast, the multi-stage Stackelberg game unfolds over multiple stages. The leaders and followers are divided into different stages to make decisions.
In this paper, we adopt the simplest variant of the multi-stage Stackelberg game, known as the two-stage Stackelberg game, to maximally reduce the time overhead incurred in-between multiple subsequent stages.
In a two-stage Stackelberg game, all leaders make the initial move (Stage 1), and then all followers choose an action after being informed of the leaders' choice (Stage 2). 
The cost function of the leaders can be defined as $\Omega(\theta, \lambda)$, where $\theta$ is one of the leader's strategies and $\lambda$ is one of the follower's strategies. Also, the cost function of the follower can be defined as $\Phi(\theta, \lambda)$.
In this paper, we consider a two-stage Stackelberg game.

\subsubsection*{Disincentive Strategy Modeling} Determining the optimal disincentive strategy $\langle \theta^*, \lambda^* \rangle$ is modeled as a two-stage Stackelberg game. Each of the players tries to minimize its own cost by determining an optimal parameter in the disincentive strategy, which can be manipulated by itself (called its optimal strategy hereafter, for simplicity), satisfying:
\begin{equation}
	\begin{aligned}
	\textbf{Stage 1 [Leaders' Side]: } \theta^* = argmin_\theta\Omega(\theta, \lambda) \\
	\textbf{Stage 2 [followers' Side]: } \lambda^* = argmin_\lambda\Phi(\theta, \lambda)
	\end{aligned}
\end{equation}

In the above game, the objective is to find an optimal disincentive strategy $\langle \theta^*, \lambda^* \rangle$, by which each participant can minimize its own cost. Meanwhile, the optimal solution must satisfy the \emph{Stackelberg Equilibrium} (\textit{SE}), so that no one is willing to adopt other strategies. The \textit{SE} is defined as follows:

\subsubsection*{Stackelberg Equilibrium, SE} An optimal disincentive strategy $\langle \theta^*, \lambda^* \rangle$ constitutes a  \textit{SE}  \textit{iff} the following set of inequalities is satisfied:
\begin{equation}
	\begin{aligned}
		\Omega(\theta^*, \lambda^*) \leq \Omega(\theta, \lambda^*) \\
		\Phi(\theta^*, \lambda^*) \leq \Phi(\theta^*, \lambda)
	\end{aligned}
\end{equation}

\subsection{Graph Skewness}
{
\label{subsec:skew}

The significance of graph skewness has long been recognized in many fields, referring to the measurement of uneven distributions of graph entities (e.g., vertices or edges) or statistics (e.g., structural characteristics), each attributed to varying definitions.
\subsubsection*{Regression-based Graph Skewness~\cite{newman2005power}}

Given a specific vertex degree $d$, the number of vertices follows a power-law distribution, $f(d) \propto d^{-\rho}$. Most natural graphs typically have a power-law constant $2 < \rho < 3$ \cite{choromanski2013scale, durrett2007random}. A smaller $\rho(G)$ indicates a more skewed graph. Since regression-based skewness incurs low precision due to the inability of logarithmic transformation to handle zero values, two other metrics are more often used~\cite{MerkelMFJ23}.

\subsubsection*{Pearson's Graph Skewness~\cite{pearson}}  
Pearson's first skewness of a graph $G$ is
$\rho_1(G) = \frac{mean(d) - mode(d)}{\sigma}$, with $d$ being the degree of vertices in $G$ and $\sigma$ being the standard deviation of the degree.
Pearson's second skewness of $G$ can be obtained by replacing the calculation of mode with that of the median,
$\rho_2(G) = \frac{3(mean(d) - median(d))}{\sigma}$.
Both
$\rho_1(G)$ and $\rho_2(G)$
are based on degree perspectives, with bigger values indicating more skewed graphs.
\subsubsection*{Planarization Graph Skewness~\cite{cimikowski1992graph}}  
The skewness of a graph $G$ can also be defined from the topology perspective, as the minimum number of edges, whose removal results in a planar graph. Since such determination is NP-hard, 
an alternative indicator
$\rho_3(G)=|E|-(3|V|-6)$
 is proposed. It implies that
the graph is more skewed as the increase of $\rho_3(G)$.
}

\section{related work}
\label{sec:relatedwork}

\subsubsection*{Edge- vs. Vertex-cut Partitioning}

There are two major graph partitioning strategies, namely {\it edge-} and {\it vertex-cut} partitioning, based on whether edges or vertices are cut.
The vertex-cut partitioning strategy evenly distributes graph edges across distributed nodes while minimizing vertex replications.
Empirical and theoretical evidence, as shown in \cite{powergraph, bgep, IAA, albert2000error}, consistently highlights the superior effectiveness of the vertex-cut partitioning over the edge-cut counterpart, especially when dealing with large graphs.
This superiority is particularly pronounced in real-world graphs characterized by a skewed distribution~\cite{tailgnn}, where vertex degrees follow a power-law pattern, yielding a large number of low-degree vertices and a small number of very high-degree vertices, as depicted in Figure~\ref{fig:skewdegree}.

\subsubsection*{Stream vs. Offline Vertex-cut Partitioning}
Despite advances in the field of vertex-cut partitioning over the past decade, the challenge of effective vertex-cut partitioning remains unsolved.
Existing approaches to address the vertex-cut partitioning problem can be broadly classified into two categories: \emph{streaming} methods \cite{clugp,hdrf,2ps-l,adwise,cusp,dbh,powergraph, gcnsplit, adwise, TaimouriS19, hcpd} and \emph{offline} methods \cite{metis,ne,dne, hep}. The offline algorithms necessitate the preloading of the partial or entire graph into the main memory before partitioning, while \emph{streaming} algorithms process the graph incrementally as a continuous stream of edges. This allows streaming methods to operate with big graphs under lower memory constraints.

\subsubsection*{Stream Partitioning }
Streaming partitioning is deemed practical for managing large-scale graphs. Methods like DBH~\cite{dbh} and Grid~\cite{grid}, which solely rely on hash-based partitioning, often exhibit poor partitioning quality.
Greedy~\cite{powergraph}, ADWISE~\cite{adwise}, and HDRF~\cite{hdrf} improve the partitioning quality by designing different scoring functions. However, their algorithmic time complexity is dependent on the number of partitions, thus limiting the system scalability.
To improve the performance, cutting-edge streaming partitioning algorithms follow a common clustering-refinement framework~\cite{2ps-l}~\cite{clugp}.
In particular, 2PS-L~\cite{2ps-l} adopts a two-stage graph traversal strategy, whereas the first phase refers to vertex clustering, i.e., Holl~\cite{holl}, and the second phase utilizes existing heuristic methods, i.e., HDRF, for graph partitioning. 
Concurrently, CLUGP~\cite{clugp} improves the performance by revising clustering algorithm, i.e., Holl~\cite{holl}, specifically for optimizing the replication factor, and accelerates the refinement process by employing a post-clustering game.
However, CLUGP and S5P differ significantly, because 1) game processes differ theoretically, leading to divergent downstream techniques. CLUGP is with static game theories, while S5P is with sequential game theories; 2) targets differ. CLUGP targets web graphs but struggles with skewed graphs. In contrast, S5P recognizes the skewness in general graphs.

\subsubsection*{Other Game-based Partitioning}
There exist edge-cut partitioning algorithms that employ a Nash game.
MDSGP~\cite{mdsgp} takes $O(r(min(2^mkm,kmG_{max})))$ time, where $m$ is the size of the strategy space, which is combinatorially high, and $r$ is the number of windows. The game function of RMGP~\cite{rmgp} is
derived based on graphical distances and semantic similarities, which incurs significant time cost ($O(|V|^3)$) and space cost ($O(|V|^2)$).
What is more, RMGP's optimization relies on heavy offline graph coloring, which takes $O(|V|^2)$ time and $O(|E|)$ space. Its game process takes $O(|V|+|E|)$ time for each iteration, and the number of iterations is not bounded.
Besides, the game has been studied in other research problems, such as \textit{vertex separator}, which aims to find a set of vertices that split a graph
into equally-sized components. CVSP~\cite{c+cv} proposes a Stackelbeg game-based method to solve this problem.
There exist significant distinctions between CVSP and S5P. First, their problem settings differ. S5P is on stream partitioning, requiring access to edges in a sequence, a prerequisite not met by CVSP. This dissimilarity gives rise to markedly different downstream techniques and strategies. Second, the application of game theories diverges entirely. In CVSP, the leader extracts a subgraph, and the follower identifies the maximum connected component, whereas, in S5P, leaders and followers are represented by head and tail clusters. Theoretically, in CVSP, each game iteration requires $O(|E|)$ time and $O(|V|^2)$, and the number of iterations is unbounded. In contrast, S5P's total time and space complexity is $O(|E|)$ and $O(|V|)$ (cf. Section~\ref{sec:theory}). Empirically, S5P shows an efficiency advantage of up to 2 orders of magnitude over CVSP (Table~\ref{tab:gameresult}).
So, it is technically challenging to apply game theories to graph partitioning.

\subsubsection*{Streaming Graph Clustering}
Streaming graph clustering is used to efficiently and effectively analyze large-scale dynamic graphs \cite{graphsummary, tcm, GSS}.
In the field of graph partitioning, representative graph clustering algorithms include \emph{Holl}~\cite{holl}, \emph{CLUGP-Clustering}\cite{clugp}, and \emph{2PS-L-Clustering}~\cite{2ps-l}. These algorithms can be generalized into an \emph{allocation-migration} framework for graph clustering.
For the migration step, the Holl algorithm calculates cluster volume based on local vertex degrees. The 2PS-L-Clustering precomputes the global degrees for the alternative of local degrees to improve the clustering performance.
Differently, the CLUGP-Clustering algorithm adopts local degrees to save the precomputation cost, while enhancing the clustering performance by introducing a \emph{splitting} operation in the migration step.
The characteristics of various streaming clustering algorithms are shown in Table \ref{tab:clusteringNEW}.
\begin{table}[ht] 
	\centering
	\caption{Streaming Graph Clustering Algorithms}
	\begin{tabular}{c|ccc}
		\toprule
		Algorithm &  Allocation & Migration & Skewness-aware \\
		\midrule
		Holl  &  $\textbf{\Checkmark}$ & local & $\text{\sffamily \textbf{X}}$ \\
		CLUGP-Clustering  & $\textbf{\Checkmark}$ & local & $\text{\sffamily \textbf{X}}$ \\
		2PS-L-Clustering   & $\textbf{\Checkmark}$ & global & $\text{\sffamily \textbf{X}}$ \\
		S5P-Clustering   & $\textbf{\Checkmark}$ & local/global & $\textbf{\Checkmark}$ \\
		\bottomrule
	\end{tabular}
	\label{tab:clusteringNEW}
\end{table}

\section{Skewness-aware Vertex-cut Partitioner}
\label{sec:method}
 \subsection{Approach Overview}
\begin{figure}[ht!]
	\centering
	\includegraphics[width=\columnwidth]{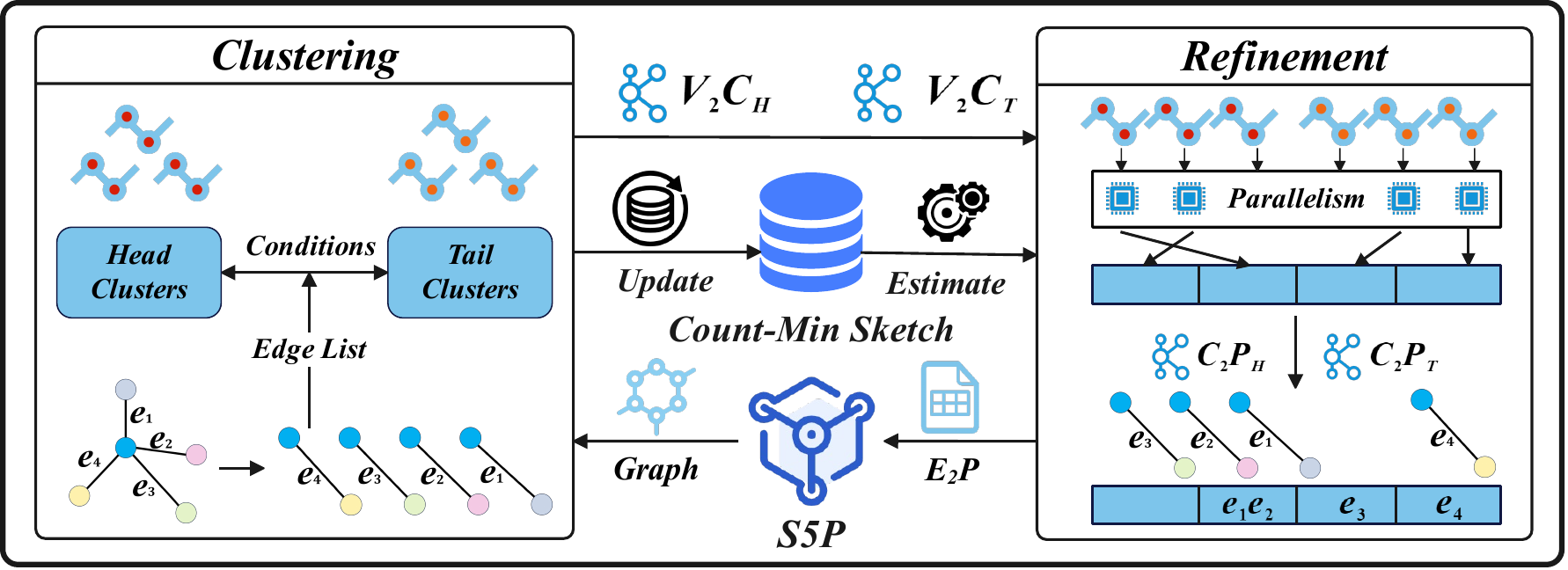}
	\caption{Skewness-aware Vertex-cut Partitioner Framework}
	\label{fig:framework}
\end{figure}

The architecture of our proposed skewness-aware vertex-cut partitioner
follows {\it clustering-refinement} framework \cite{2ps-l,clugp}, where streaming graph edges are processed in a biphasic manner, as shown in Figure \ref{fig:framework}.
First, we study skewness-aware graph clustering to generate fine-grained head and tail clusters (Section 3.2).
Second, we delve into a Stackelberg game to allocate head and tail clusters to a predefined set of partitions (Section 3.3).
This allocation minimizes the number of vertices crossing partitions while ensuring an evenly distributed storage load.
At last, we present a series of optimization techniques to enhance the time and space efficiency in meeting the requirements of stream partitioning (Section 3.4).

\subsection{Skewness-aware Streaming Graph Clustering}

\begin{definition}[\bf Head and Tail Vertices/Edges/Clusters]
 Given a vertex $v \in V$ of a graph $G=(V, E)$, if the degree of $v$ ($d(v)$) is higher than a predefined threshold\footnote{\label{ft:para}
In our implementation, we adopt consistent parameter settings for all datasets, with $\xi$ and $\kappa$ following~\cite{holl,clugp,2ps-l}.
In particular, we set $\xi = \beta\frac{2|E|}{|V|}$, which is the product of coefficient $\beta$ and the average degree of the graph. Here, $\beta$ equals $1$, which is a constant (as discussed in Section~\ref{sec:ps}) and $|E|$ and $|V|$ are fixed for a given dataset.
We set $\kappa = \frac{2|E|}{k}$, where $k$ is the number of partitions, and $|E|$ is fixed for a given dataset.
} $\xi$, i.e., $d(v)>\xi$, then $v$ is a head vertex; otherwise, it is a tail vertex. Accordingly, an edge $e(v_i, v_j)$ is a head edge, if $v_i$ and $v_j$ are both head vertices; otherwise, $e$ is a tail edge.

	\label{defGHCM}
\end{definition}

It is worth highlighting that tail vertices exclusively appear within tail edges, whereas some head vertices may be associated with two distinct edges: one as a head edge and another as a tail edge. This occurs because a head vertex can be present in both head and tail edges, according to Definition \ref{defGHCM}.  
We showcase an example in Figure~\ref{fig:vis1}, where head vertices and edges are in blue, and tail vertices and edges are in orange.
The number of an edge represents the order of its arrival in a stream.

\begin{figure}[ht!]
	
	\centering
	\includegraphics[width=0.6\columnwidth]{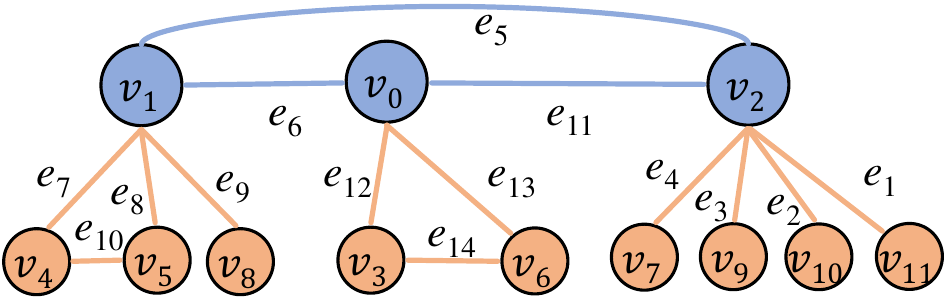}
	\caption{A Toy Graph with $12$ Vertices and $14$ Edges}
	\label{fig:vis1}
\end{figure}

Based on the definition of head/tail vertices/edges, we can define head/tail clusters.
For a cluster produced by head edges/vertices, the cluster is a head cluster named $C_H$; otherwise, it is a tail cluster named $C_T$. Next, we formalize the problem of streaming graph clustering as follows.

\begin{definition}[\textbf{Streaming Graph Clustering}]
Consider a graph of streaming edges, denoted as $G_S = \{e_1, e_2, \cdots, e_{|E|}\}$, along with a maximum cluster volume represented by $\kappa$\footref{ft:para}
the objective of streaming graph clustering is to assign each vertex $v$ to head or tail clusters, such that the edge-cutting is minimized.
The clustering output can be presented as a table that maps a vertex to multiple specific (head/tail) clusters, i.e., $\{\langle v_i, c_j\rangle\}$.

\end{definition}

\begin{algorithm}[h!]
	\caption{Skewness-aware Graph Clustering}
	\label{alg:kagc}
	\begin{flushleft}
		\hspace*{\algorithmicindent} \textbf{Input} Streaming Graph Edge Set $E_S$, Cluster constrains $\kappa$\\
		\hspace*{\algorithmicindent} \textbf{Output} $V_2C_H$, $V_2C_T$
	\end{flushleft}
	\begin{algorithmic}[1]
		\For{$e(u, v) \in E_S$}
		\If{$e \in E_H$}
		\If{$V_2C_H[u] \ or \ V_2C_H[v] \ is \ NULL$}	
		\state Assign a new ID
		\EndIf
		\State Update $vol(\cdot)$ by $d(u)$ and $d(v)$
		\If{ $vol(V_2C_H[u])\ and \ vol(V_2C_H[v]) \ < \ \kappa$}
		\State $i \leftarrow \mathop{\arg\min}_{z \in \{u, v\}} (vol(V_2C_H[z]) - d(z))$
		\State $j \leftarrow z \in \{u, v\} : z \neq i$
		\If{ $vol(V_2C_H[j]) + d(i) < \ \kappa$}
		\State $vol(V_2C_H[j]) \leftarrow vol(V_2C_H[j]) + d(i)$
		\State $vol(V_2C_H[i]) \leftarrow vol(V_2C_H[i]) - d(i)$
		\State $V_2C_H[i] \leftarrow V_2C_H[j]$
		\EndIf
		\EndIf
		\EndIf
		\If{$e \in E_T$}
		\If{$V_2C_T[u] \ or \ V_2C_T[v] \ is \ NULL$}	
		\state Assign a new ID
		\EndIf
		\State Update $vol(\cdot)$ by $1$
		\State Update $ld(\cdot)$ by $1$
		\If{ $vol(V_2C_T[u])\ and \ vol(V_2C_T[v]) \ < \ \kappa$}
		\State $i \leftarrow \mathop{\arg\min}_{z \in \{u, v\}} (vol(V_2C_T[z]))$
		\State $j \leftarrow z \in \{u, v\} : z \neq i$
		\State $vol(V_2C_T[j]) \leftarrow vol(V_2C_T[j]) + ld(i)$
		\State $vol(V_2C_T[i]) \leftarrow vol(V_2C_T[i]) - ld(i)$
		\State $V_2C_T[i] \leftarrow V_2C_T[j]$
		\EndIf
		\EndIf
		\EndFor
	\end{algorithmic}
\end{algorithm}

Clustering is widely acknowledged as a method for data summarization. In the context of graphs, clustering functions as a method for graph summarization\cite{graphsummary}.  {Graph clustering can be classified into two categories: {\it vertex-} and {\it edge-centric clustering}. Vertex-centric clustering offers the advantage of being more space- and time-efficient ($O(|V|)$ vs. $O(|E|)$), since the number of vertices is significantly smaller than the number of edges~\cite{clugp}. Conversely, edge-centric clustering preserves more connection information but at the expense of high space cost, i.e., $O(|E|)$.} To strike a balance between the two, we represent edge-centric clusters in the form of vertex clusters, employing two tables, denoted as $V_2C_H$ (Vertex-to-Head-Cluster) and $V_2C_T$ (Vertex-to-Tail-Cluster).

Algorithm \ref{alg:kagc} processes graph edges in a one-pass manner. For each upcoming edge, we first determine its type (head or tail) and subsequently execute an allocation operation in the corresponding Vertex-to-Cluster table ($V_2C$).
This involves assigning a unique ID to vertices without cluster labels. In our experiments, we employed natural numbers as IDs, commencing from $0$ and incrementing by $1$ for each new allocation.
Following this, we update the corresponding cluster size, denoted as \emph{vol} in Algorithm~\ref{alg:kagc}.
In the case of a head edge, a ``splitting'' operation is performed independently for both clusters associated with the two vertices of the edge, provided their sizes exceed the specified cluster capacity $\kappa$. Regarding the ``Update $vol(\cdot)$'' operation, for tail clusters, when a vertex is added to a cluster , we increase the volume of the cluster by $1$. This process is equivalent to updating based on local degree ($ld$) information.
For head clusters, when a vertex joins a cluster, we increase the volume of the cluster by the degree of the vertex, reflecting an update that incorporates global degree information.

For head edges, we execute a global degree-aware operation (lines $5-11$). In contrast, for tail edges, we employ a local degree-aware operation (lines $13-21$). This distinction arises from the fact that the vertices in head edges possess higher degrees, making it necessary to use a global degree approach to accurately represent their degrees. Conversely, the nodes in tail edges have relatively smaller degrees, and employing local degree measures is adequate for reflecting their true degrees.

 {Taking the toy graph in Figure~\ref{fig:vis1} as an example ($k$ = 3), we can compute $\kappa = \frac{2 \times 14}{3} \approx 9.3$. When considering the initial head edge $e_5(v_1, v_2)$, which arrives first, both vertices $v_1$ and $v_2$ are assigned to newly created head clusters, $c_1$ and $c_2$ (cf. line 3). Following this, we proceed to update the volumes of $c_1$ and $c_2$ to $5$ and $6$, respectively (cf line 4). After that, due to the combined volume of $c_1$ and $d(v_2)$ exceeding $\kappa$ (11 > 9.3), we abstain from executing migration operations (lines 8-11). Similarly, when the subsequent edge $e_6(v_0, v_1)$ arrives, $v_0$ is assigned to a new head cluster and migrated to $c_1$.
For tail edges, we apply a comparable operation, with the primary distinction being the substitution of the global degree with a real-time updated local degree.

Once the clustering process is completed, we can obtain the clustering results, as shown in Figure~\ref{fig:vis2}. Here, $c_1$ and $c_2$ are head clusters, while $c_4$, $c_5$, and $c_6$ are tail clusters.
From the outcome presented, we can observe that skewness-aware clustering effectively captures information from both head and tail nodes, resulting in the formulation of different cluster types, including head and tail clusters, as well as leaders and followers.
This underscores a capacity that eludes the clustering methods described in Section~\ref{sec:relatedwork}..
}
\begin{figure}[ht!]
	
	\centering
	\includegraphics[width=0.6\columnwidth]{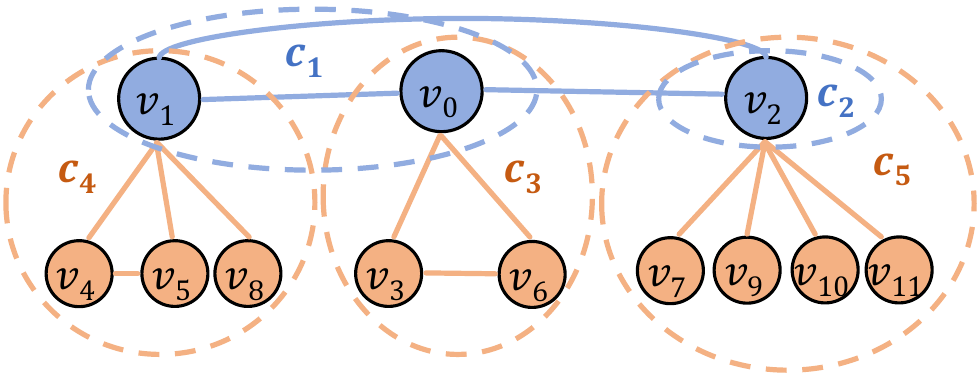}
	\caption{\small {Skewness-aware Streaming Graph Clustering ($k$=3)} }
	\label{fig:vis2}
\end{figure}

\subsection{Stackelberg Game-based Partitioning}

Using the clusters generated by Algorithm~\ref{alg:kagc} as input, we can transform the graph partitioning problem into cluster-to-partition assignment process. In the process, each cluster is allocated to one of the predefined partitions. Once clusters are assigned, the edges within each cluster can be mapped to their respective partitions.

In the cluster-to-partition assignment process, every (head/tail) cluster takes on the role of a player. These players engage in both competitive and collaborative actions concurrently, all geared toward minimizing the replication factor during the partition assignment process. To portray this dynamic, each player is equipped with a distinct optimization objective encapsulated within their cost function. In the context of a Stackelberg game involving both head and tail clusters, these cost functions stand as the measure of their objectives.
The process continues until the attainment of the \emph{Nash equilibrium} within the Stackelberg game, signifying that all clusters have fine-tuned their cost functions, and no player can improve the cost by unilaterally changing their strategy.

In the sequel, we investigate the player modeling, cost functions, Stackelberg game Nash equilibrium, as well as the best response dynamics and postprocessing.

\subsubsection*{Clusters as Players}
As aforementioned, each head or tail cluster is taken as a player of a two-stage Stackelberg game. The decision not to directly operate at the level of individual graph edges is driven by two primary factors.
{First, the time complexity of using edges is $O(|E|^2)$ and that of using clusters is only $O(|C|^2)$ (cf. Section~\ref{sec:timeana}), and $|E|$ is much greater than $|C|$ (cf.  Section~\ref{sec:timeana}), so conducting the game at the granularity of edges would be expensive and using clusters as players can avoid computational bottlenecks, i.e., from $O(|E|^2)$ to $O(|C|^2)$.}
Second, clustering offers a means to summarize the graphical structural information and connectivity patterns through cluster sizes and intersections\cite{graphsummary}. With this streamlined representation, modeling interactions between head and tail clusters via a Stackelberg game becomes a practical undertaking. The distinct leader-follower roles align with the varying influence of head and tail edges. The cluster-level representatives are generated through our skewness-aware clustering method. It can preserve skewness properties because the graph clustering is regarded as a summarization of graph information\cite{graphsummary}.

\subsubsection*{Cost Function}
In a \emph{Stackelberg game}, cost functions are vital mathematical constructs, which quantify the objectives and modeling the actions of all players, i.e., head and tail clusters, by assigning numerical values to their strategies based on utility or cost.
These functions are instrumental in modeling the decision-making process within this hierarchical framework, where head clusters' initial actions influence subsequent tail clusters' responses.
Essentially, cost functions guide players in optimizing their strategies, considering both their objectives and the expected reactions of others, thus shaping the equilibrium dynamics of a Stackelberg game.
Hereby, we define a new social welfare function as follows.

\begin{equation}
	\label{eq:software}
	S(\Lambda) =  \underbrace{\delta\frac{\sum_{i=1}^k{|p_i|^2}}{k}}_{load \ balance \ part}  + \underbrace{\frac{\sum_{i=1}^k{\Theta(p_i, V)}}{k}}_{communication \ part}
\end{equation}
where $\delta$ is the normalization factor and $\Theta(p_i, V - p_i)$ is the total number of edges in $|E|$ that span across partition $p_i$ and all other partitions.
Thus, $\Theta(p_i, V) = \Theta(p_i, V- p_i) + |p_i|$. Subsequently, we establish the unique cost function for every player (cluster) using the social welfare metric $S(\Lambda)$ as the basis.

For each cluster, $c_i \in C_H$ or $C_T$, the cost associated with the selection of partition $p_i$ for the cluster is expressed as follows.
\begin{equation}
	S_{c_i}(p_i) = \frac{\delta}{k} |c_i| \cdot |p_i| + \frac{F(c_i) + |c_i|}{k} \text{~~, where}
\label{eq:clustercost}
\end{equation}
\begin{equation}
\label{eqn:theta}
F(c_i) = \sum_{c_j \in C_H \cup C_L}\Theta(c_i, c_j)\mathbb{I}(i, j)
\end{equation}
Here, $\Theta(c_i, c_j)$ is the total number of edges in $|E|$ that span across cluster $c_i$ and $c_j$. $\mathbb{I}(i, j)$ serves as an indicator function that yields $1$ if $P(c_i) \neq P(c_j)$, and $0$ otherwise. According to Theorem \ref{theorem:social}, it holds that the social welfare $S(\Lambda)$ corresponds to the summation of all individual cost incurred by clusters.

\subsubsection*{Stackelberg Game Nash Equilibrium}
Given established cost functions, the \emph{Stackelberg game Nash equilibrium} occurs when all head and tail clusters simultaneously achieve their optimization objectives while accounting for the strategic choices made by the other clusters, in relevance to graph partitioning.

A strategy decision profile, denoted as $\Lambda^* =  \Lambda^*_H + \Phi^*_T$, constitutes a Nash equilibrium, where $\Lambda^*_H$ and $\Lambda^*_T$ represent the set of strategies employed by head and tail clusters, respectively.
In this equilibrium state, when both head and tail clusters achieve their respective local optimization targets, there exists no cluster with a motivation to unilaterally deviate from its strategy to obtain a lower cost. 
{
Leaders in the Stackelberg game have the advantage of the first move \cite{stackelberggameeq}, making the equilibrium of game in favor of optimization towards leaders, which is consistent with the strategies in skewness-aware graph partitioning.
For example, as seen in Figure~ \ref{fig:degree}, S5P has a significant advantage over CLUGP in terms of average RF for vertices with high degrees, which is based on simultaneous-move game theories. We empirically examine the performance of Stackelberg game in comparison to other game variants in Table~\ref{tab:gameresult}.

}
\begin{algorithm}[h]
	\caption{Stackelberg Game}
	\label{alg:psgg}
	\small
	\begin{flushleft}
		\hspace*{\algorithmicindent} \textbf{Input} $V_2C_H$, $V_2C_T$ \\
		\hspace*{\algorithmicindent} \textbf{Output} $C_2P_H$, $C_2P_T$
	\end{flushleft}
	\begin{algorithmic}[1]
		\State \Call{Initialization}{$V_2C_H$, $V_2C_T$}
		\Repeat
		\State\Call {LeaderDesisionProcess}{}
		\State\Call {FollowerDesisionProcess}{}
		\Until{No clusters change their current strategies}
		
		\Procedure{Leader/FollowerDesisionProcess}{}
		\For{each high/tail cluster $c_i \in C_H/C_T$}
		\State Find $C_i$'s best response $\theta$ which is based on Eq.(\ref{eq:clustercost})
		\If{$\theta \neq \theta_{c_i}$}
		\State$\theta_{c_i} \leftarrow \theta$
		\EndIf
		\EndFor
		\EndProcedure
	\end{algorithmic}
\end{algorithm}
\begin{algorithm}[h]
	\caption{Postprocessing}
	\label{alg:ppe}
	\small
	\begin{flushleft}
		\hspace*{\algorithmicindent} \textbf{Input} $C_2P_H$, $C_2P_T$, $L = \frac{\tau|E|}{k}$ \\
		\hspace*{\algorithmicindent} \textbf{Output} Edge Partition Results
		
	\end{flushleft}
	\begin{algorithmic}[1]
		\For{$e(u, v) \in E$}
		\If{$e \in E_H$}
		\State $Clu_u = V_2C_H(u), Clu_v = V_2C_H(v)$
		\EndIf
		\If{$e \in E_L$}
		\State $Clu_u = V_2C_T(u), Clu_v = V_2C_T(v)$
		\EndIf
		\State $P_u = C_2P(Clu_u), P_v = C_2P(Clu_v)$
		\If {$Load(P_u) > L$ and $Load(P_v) > L$}
		\State Place $e$ in a partition with available space.
		\ElsIf{$Load(P_u) > Load(P_v)$} 
		\state \ \ Place $e$ into $P_v$
		\Else
		\state \ \ Place $e$ into $P_u$
		\EndIf
		\State update $Load$
		\EndFor
	\end{algorithmic}
\end{algorithm}
\subsubsection*{Best Response Dynamics}
The Nash equilibrium of a game can be either pure Nash equilibrium or mixed Nash equilibrium.
{{The decision form corresponding to the mixed strategy Nash equilibrium is expressed in probabilities. As a result, the game is with randomness and poses challenges in convergence~\cite{nash_2005}.} Therefore, in this work, we consider pure Nash equilibrium, which offers more intuitive and stable solutions.
A pure Nash equilibrium in a Stackelberg game can be achieved through the application of \emph{Best Response Dynamics}, as discussed in~\cite{cusp}.
The core of the Best Response Dynamics algorithm lies in the design of the {\it response} function, as shown in Algorithm~\ref{alg:psgg} (line 8). The design of the response function needs to satisfy two requirements. First, it should optimize the global software function (Equation~\ref{eq:software}) through local best response (Equation ~\ref{eq:clustercost}), guaranteed by Theorem~\ref{theorem:social}. Second, the computation cost of the response function should not be very high, as ensured by the optimized sketch-based methods discussed in Section \ref{sec:sketch}.
\begin{figure}[ht!]
	\centering
	\includegraphics[width=\columnwidth]{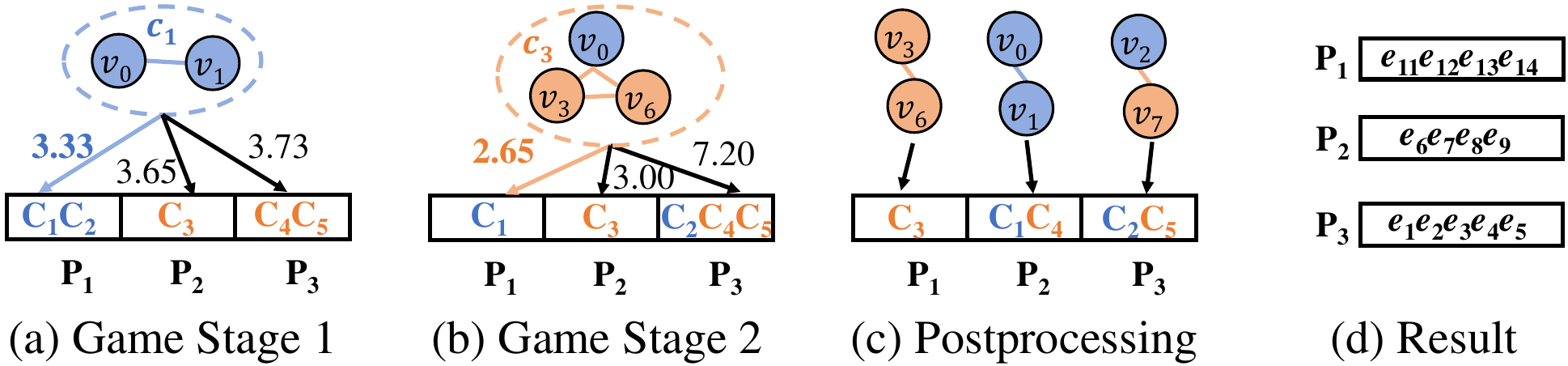}
	\caption{\small {Stackelberg Game-based 
Partitioning ($k$=$3$)}}
	\label{fig:vis3}
\end{figure}

 {{Figures~\ref{fig:vis3} (a) and (b) depict the operation of the Stackelberg game. Assuming the initial partitioning as $c_1$ and $c_2$ in $p_1$, $c_3$ in $p_2$, and $c_4$ and $c_5$ in $p_3$, we first calculate
$\delta \approx 0.65$ based on Equation (\ref{sigmamax}).

In Stage $1$, head clusters, starting with $c_1$, compute their partition cost (cf. Equation~\ref{eq:clustercost}) and migrate to the partition with the minimum cost. In this case, $c_1$ stays in $p_1$ and $c_2$ is mitigated to $p_3$. Then, tail clusters initiate their decision making process in Stage $2$. For example, $c_3$ calculates its cost and moves to the appropriate partition. Stage 2 terminates once all tail clusters have been processed, making the end of this game round.
Another round of the game commences, and continues until Nash equilibrium is reached, where decisions of all clusters remain unchanged (Figure~\ref{fig:vis3}(c)).

}}

\subsubsection*{Postprocessing}
 {Following the Stackelberg game phase, we can get a
 cluster-level balancing because of the load balance part of Equation~(\ref{eq:software}). Then, the transitioning from a broader cluster-level balancing to a more refined edge-level one is expected.}
Leveraging the previously obtained graph information, we can make efficient use of existing clusters during the final allocation of edges. Similarly, in the allocation process outlined in Algorithm \ref{alg:ppe}, we adopt a strategy that takes into account the skewness of the distribution.
For a given edge $e(u, v)$, we query the current partitions to which vertices $u$ and $v$ belong to. If both vertices belong to partitions whose sizes have already exceeded the capacity threshold $L$, we execute the operation described in line 8 of the algorithm.
For head edges, we traverse from the first partition to the last, adding these edges to the first partition that satisfies the condition $|p_i| < L$.
Conversely, for tail edges, we perform the same operation in reverse order. The objective of this operation is to minimize the number of partitions an edge can be assigned to, thereby reducing the replication factor.
If each of them belongs to partitions whose sizes are already less than $L$, add this edge to the partition with a larger size.
Upon completing the postprocessing phase, all edges have been successfully assigned to partitions, and no partition contains more than $\lceil\tau\frac{|E|}{k} \rceil$ edges\footnote{In our implementation, we set $\tau$ to $1.0$ to meet the requirement of load balancing. }. {In Figure~\ref{fig:vis3}(c), for $e_6(v_0,v_1)$, since the head clusters corresponding to $v_0$ and $v_1$ are both in partition $p_2$, the edge is added to $p_2$. Similarly, $e_{14}(v_3,v_6)$ is added to $p_1$, and $e_4(v_2,v_7)$ is added to $p_3$. After processing, we obtain the edge partitioning results, as shown in Figure~\ref{fig:vis3} (d).}

\subsection{Optimization Strategy}
\subsubsection*{Sketch}
\label{sec:sketch}
As a player of the game, a partition necessitates the calculation the cost function for strategy tuning, which requires the calculation of the number of inter-cluster edges (i.e., $\Theta(c_i, c_j)$ in Equation~\ref{eqn:theta}).
The most straightforward approach would be to design a red-black tree data structure for mapping incident vertices of an edge to its count, which is incrementally updated during the edge traversal.
However, such a design would incur significant cost in terms of both query and insertion operations. To address this issue, we adopt the \emph{Count-Min Sketch}~\cite{cmsketch} (\emph{CMS}) method to achieve the same functionality. It uses multiple hash functions and maintains a two-dimensional array of $w$ columns and $r$ rows.
One of the key features of CMS is its efficient utilization of space in summarizing large-scale data, making it particularly suitable for applications where memory usage is of utmost concern. As a trade-off, it provides approximate results with guaranteed errors stemming from hash collisions.
The parameters $w$ and $d$ are chosen by setting $w = \lceil\frac{\hat{e}}{\epsilon}\rceil$ and $d = \lceil ln\frac{1}{\nu}\rceil$, where the error in answering a query is within an additive factor of $\epsilon$ with probability $1 - \nu$ and $\hat{e}$ is Euler's number\footnote{{In our implementation, we set $\epsilon$ to 0.1 and $\nu$ to 0.01 and we can get $w = 27$ and $d=4.6$.}}.
In our specific problem context, we set the elements to be hashed as the concatenation of two cluster ID strings. The concatenated string can be posted to a CMS during the edge traversal, and the corresponding count can be approximately retrieved for the calculation of Equation~\ref{eqn:theta}. {If we do not use a CMS, we would need to allocate space of size $|C|^2$ to store the intersection sizes of different clusters. By using CMS, we only need a two-dimensional array of size $w \times d$ for insertion and querying. The accuracy of the queries is guaranteed by the principles of CMS~\cite{cmsketch}. For instance, with our configuration, we only need $27 \times 4.6$ units of space to ensure that the probability of the query result differing from the true result by no more than 10\% is greater than 99\%. We can reduce $\epsilon$ and $\nu$ to make the probability of errors in the CMS smaller; however, this will result in an increase in the space of the two-dimensional array (i.e., larger values of w and d), leading to higher memory usage.}

\subsubsection*{Parallelism}
\label{sec:para}
The primary advantage of the game theory in graph partitioning is its ability to enable parallel acceleration~\cite{clugp, cusp}.
In S5P, parallelization is facilitated through clustering, which preserves the graph's locality, ensuring that clusters with closely aligned cluster IDs are often adjacent in the graph structure. This feature makes our Stackelberg game method well-suited for the parallelization implementation.
On the other hand, given that clustering algorithms effectively extract graph information and reduce the scale of game players, the number of the game players is still significant.
Besides, we employ a batch-wise parallel approach to accelerate the gaming process. 
The parallelism between different batches is managed through a thread pool, with each batch running independently.
To enhance the convergence speed and the quality of the game, we can also impose restrictions on the number of rounds in the game.
 {There also exists potential enhancement of partitioning efficiency through more sophisticated parallelization techniques.
We assert that a significant contribution of the Stackelberg game lies in its facilitation of parallelization. It is essential to highlight that the current implementation's efficiency suffices for enhancing downstream distributed graph processing performance. The partitioning runtime overhead constitutes only a small fraction of the overall distributed graph processing. For example, in Figure~\ref{fig:powergraph} (c), the runtime cost of partitioning on S5P is merely 19.7\% of that of the overall processing. Further optimization efforts in parallelization are orthogonal to the scope of this paper.}

\section{theoretical analysis}
\label{sec:theory}
\subsection{Time Cost Analysis}
\label{sec:timeana}
We analyze each phase of S5P separately. Phase 1, speciﬁed in Algorithm \ref{alg:kagc}, performs a ﬁxed number of passes through the edge set. In each pass, a constant number of operations is performed on each edge. Hence, the time complexity of the ﬁrst phase is in $O(|E|)$. Phase 2, the time cost for each round of Stackeberg graph game is $O(|C|^2 + k|C|)$ and the round complexity can be bound by $O(k|V|)$. With parallelization, the average time complexity can be approximated as $O(|C_{avg}|^{2} \frac{|C|}{batchsize \times t})$, where $|C_{avg}|$ is the average number of clusters in each batch. In summary, the second phase of S5P has a time complexity of $O(|E|)$, as $|E| \gg |V|\gg |C|$.
Phase 3, edges are partitioned with a constant number of operations resulting in $O(|E|)$ time complexity. The total time complexity of S5P is, hence, in $O(|E|)$, i.e., linear in the number of edges.

\subsection{Memory Cost Analysis}
In our analysis of the S5P algorithm's data structures, we employ various arrays as illustrated in Algorithm \ref{alg:kagc}. These arrays are utilized to store information such as vertex local and global degrees, cluster volumes, vertex-to-cluster mappings, and cluster details. It's worth noting that each of these arrays has a space complexity of $O(|V|)$. In Algorithm \ref{alg:psgg}, we introduce additional arrays for mapping clusters to partitions and maintaining cluster-related information. Similar to the previous case, these arrays also exhibit a space complexity of $O(|V|)$. In Algorithm \ref{alg:ppe}, we employ a vertex-to-partition replication matrix, which, despite its seemingly larger size, boasts a space complexity of $O(k|V|)$. Notably, this space complexity remains unaffected by the number of edges present in the graph.

\subsection{Metric Analysis }
\begin{theorem}
	The relative load balance $\tau$ can be bounded by $\frac{kL}{|E|}$.  If we want to ensure that the upper limit for relative load balance is $t$, we can set $maxLoad$ to $\frac{t|E|}{k}$.
\end{theorem}

\subsubsection*{Bounds on Replication Factor}
To the best of our knowledge, most of the streaming vertex-cut algorithms do not have bounds. {In the line of research on stream partitioning, RF bounds are commonly employed to quantify the gap with optimal solutions~\cite{hdrf, dbh}.} The Grid \cite{grid} and HDRF \cite{hdrf} have a well-defined bound, and the DBH algorithm \cite{dbh} only bounds the average RF value. The typical approach for bounding vertex-cut partitioning algorithms is to limit the maximum number of partitions each edge can be assigned to. However, such constraints often lead to improved Load Balance. Striking a trade-off between RF and load balance is a challenging task. Based on this, we modified S5P to create \emph{S5P w/ bound} (\emph{S5P-B}), which is a bounded version of the streaming vertex-cut algorithm with an RF bound.

S5P-B has two key modifications compared to S5P. Firstly, it entirely relies on global degrees in place of all local degrees in its processing. Secondly,
We have removed the cluster constraints $\kappa$ Algorithm \ref{alg:kagc} and $maxLoad$ in Algorithm \ref{alg:ppe} to better formalize the RF bound of our skewness-aware design.

\begin{theorem}

	Algorithm S5P-B achieves a replication factor,
when applied to partition a graph with $|E|$ edges on $k$ partitions, that can be bounded by:
\begin{align}
	\label{eq:rfbound}
    \text{RF} \leq \underbrace{\chi_H \cdot k}_{head \ part} + \underbrace{\sum_{i = 1}^{\chi_T|V|}\frac{ d_m((\frac{k - 1}{d_m})^{1-\rho} + \frac{i-1}{|V|})^{-1}}{\chi_T\cdot |V|} + 1}_{tail \ part}
\end{align}
	where $\chi_T$ (or $\chi_H$) represents the fraction of low-degree (or high-degree) vertices and $d_m$ is the global minimum vertex degree.
\end{theorem}
\begin{proof}

	{
	The RF bound comprises two components: the head part and the tail part. The head part accounts for the fraction $\chi_H$ of high-degree vertices in the graph.
In the worst case, head vertices are replicated across all partitions, resulting in the head part of RF as $\chi_H\cdot k$.
Here, $\chi_T$ (or $\chi_H$) can be expressed as $\frac{\sum_{j=1}^{\xi}j^{-\rho}}{|V|}$ (or $1- \frac{\sum_{j=1}^{\xi}j^{-\rho}}{|V|}$), because the vertex degrees in the graph follow a power-law distribution (Section~\ref{subsec:skew}).

The tail part considers the contribution to the RF from tail vertices.
Firstly, only \emph{allocation} in Algorithm \ref{alg:kagc} (lines 3 and 13) replicates a tail vertex connecting head vertices and produces at most 1 replica. Secondly, the replication of a tail vertex connecting to tail vertices can be bound by its degree. Rescaling based on the combination of the tail vertex degree bound from \cite{cohen, hdrf} and the range of $\rho$ in real-world graphs in \cite{choromanski2013scale, durrett2007random}, the degree of a tail vertex connecting tail vertices can thus be bounded by $d_m((\frac{k - 1}{d_m})^{1-\rho} + \frac{i-1}{|V|})^{-1}$. So, the tail part RF can be bounded by $\sum_{i = 1}^{\chi_T|V|}\frac{1 + d_m((\frac{k - 1}{d_m})^{1-\rho} + \frac{i-1}{|V|})^{-1}}{\chi_T|V|} = \sum_{i = 1}^{\chi_T|V|}\frac{ d_m((\frac{k - 1}{d_m})^{1-\rho} + \frac{i-1}{|V|})^{-1}}{\chi_T|V|} + 1$.}
\end{proof}
{Given that RF of an optimal solution is known to be no less than 1, Equation~(\ref{eq:rfbound}) provides an upper bound for the approximation ratio. 
}

\begin{theorem}
	\label{thm:dandiao}
	{The bound of \text{RF} will decrease as the decrease of $\rho$.}
	\end{theorem}
	\begin{proof}
	{For the head part, as $\rho$ decreases, the head part decreases due to the decrease of 1- $\frac{\sum_{j=1}^{\xi}j^{-\rho}}{|V|}$.
The tail part can be regarded as the average of $\Gamma(\rho)= d_m((\frac{k - 1}{d_m})^{1-\rho} + \frac{i-1}{|V|})^{-1}$. Since $\Gamma(\rho)$ decreases monotonically with respect to $\rho$, the tail part also deceases monotonically. Therefore, the theorem is proven.}
	\end{proof}

	S5P-B can be regarded as the approximation of S5P. Firstly, for tail edges, the local degree can be a good estimate for the global degree \cite{hdrf}. Secondly, adding load balance constraints during the clustering process and the postpartitioning process can better serve downstream tasks. Overall, S5P is expected to achieve a replication factor similar to S5P-B.
	{Since the $\rho$ serves as an indicator of graph skewness,
Theorem~\ref{thm:dandiao} theoretically guarantees a tighter RF bound of S5P on a more skewed graph.}

\subsection{Stackelberg Game Analysis}

\begin{theorem}
\label{theorem:social}
The stackelberg game social welfare is the sum of all individual cost of clusters.
\end{theorem}
\begin{proof}
\begin{align}
		&~~~~\sum_{c_i \in C_H \cup C_L}\frac{\delta}{k} |c_i| \cdot |p_i| + \sum_{c_i \in C_H \cup C_L}\frac{F(c_i)}{k} + \sum_{c_i \in C_H \cup C_L}\frac{|c_i|}{k}  \\
		&= \frac{\delta}{k}\sum_{i=1}^{|C|}|c_i|\cdot|C_2P(i)| + \frac{1}{k}\sum_{i=1, j=1}^{|C|} \Theta(c_i, c_j)\mathbb{I}(i, j) +\frac{1}{k}\sum_{i=1}^{k}\sum_{c_j \in p_i} |c_j| \nonumber \\
		&= \frac{\delta}{k}\sum_{i=1}^{k}\sum_{c_j \in p_i}|c_j|\cdot|p_i| + \frac{1}{k}\sum_{i=1}^{k}\sum_{c_j \in p_i} \Theta(c_j, V - C_2P(c_j)) + \sum_{i=1}^{k}\frac{|p_i|}{k} \nonumber\\
		&= \delta \frac{\sum_{i=1}^k{|p_i|^2}}{k} + \frac{1}{k}\sum_{i=1}^{k} (\Theta(p_i, V - p_i) + |p_i|) \nonumber\\ 
		&=\delta \frac{\sum_{i=1}^k{|p_i|^2}}{k}  + \frac{\sum_{i=1}^k{\Theta(p_i, V)}}{k}
	\label{eq:socialproof}
\end{align}
\end{proof}
\subsubsection*{Normalization}
We demonstrate the process of selecting the normalization factor $\delta$. The numerical values of the two components of the social welfare function exhibit significant disparities. As a result, we present a strategy for establishing the value of $\delta$. We can assume the two components in Equation~(\ref{eq:software}) are of importance to calculate $\delta = \frac{\sum_{i=1}^k{\Theta(P_i, V)}}{\sum_{i=1}^k{|P_i|^2}}$ \cite{clugp}. Then, we have:
\begin{equation}
	\label{eq:nomal}
	\frac{1}{\sum_{c_i \in C_H \cup C_L}|c_i|} \leq \delta  \leq \frac{k\sum_{c_i \in C_H \cup C_L}(F(c_i) + |c_i|)}{(\sum_{c_i \in C_H \cup C_L}|c_i|)^2}
\end{equation}
\begin{proof}
	We can employ the method of extreme assumptions to demonstrate the above formula. When all clusters are assigned to the same partition, we can get:
	\begin{equation}
	\nonumber	\delta_{min} = \frac{\sum_{c_i \in C_H \cup C_L}|c_i|}{(\sum_{c_i \in C_H \cup C_L}|c_i|)^2} = \frac{1}{\sum_{c_i \in C_H \cup C_L}|c_i|}
	\end{equation}
	
	If all clusters are evenly distributed to all partitions, we have:
	\begin{equation}
		\label{sigmamax}
		\delta_{max} = \frac{\sum_{c_i \in C_H \cup C_L}(F(c_i) + |c_i|)}{k(\frac{\sum_{c_i \in C_H \cup C_L}|c_i|}{k})^2} = \frac{k\sum_{c_i \in C_H \cup C_L}(F(c_i) + |c_i|)}{(\sum_{c_i \in C_H \cup C_L}|c_i|)^2}
	\end{equation}
	
	Setting $\delta$ to a value within the range of Equation~(\ref{eq:nomal}) can help ensure that the two components are treated as equally as possible. In the experiment, we set $\delta$ to the maximum value within its range.	
\end{proof}

\begin{theorem}
	\label{theorem:poabound}
	The \text{PoA} of the stackelberg game is bounded by k + 1
	
\end{theorem}
\subsubsection*{Price of Anarchy}
PoA (\emph{Price of Anarchy}) is the worse ratio of social welfare $S(\Lambda)$ when the strategy profile is a Nash Equilibrium over the optimum value of $S(\Lambda)$. It is an important metric to measure the quality of Nash Equilibriums. To minimize cost social welfare objective, it is formally defined as:
\begin{equation}
	\text{PoA} = \frac{max_{\Lambda^* \in \text{PNE}}S(\Lambda^*)}{\text{OPT}}
\end{equation}
where PNE is the set of pure Nash equilibrium and OPT is the global minimum value of social welfare in the case of cost minimization problem.
\begin{proof}
	When all clusters are in one partition, the load balance part is maximized. It is easy to see from the mean inequality that the load balance part is minimized when the sizes in each partition are equal.
	If the load balancing and cutting factors get their maximum values simultaneously, the upper bound of $S(\Lambda)$ is:
	\begin{equation}
	\label{eq:upperboundPOA}
	S(\Lambda^*) \leq \frac{\delta}{k} (\sum_{i=1}^{m}|c_i|)^2 + \frac{1}{k}\sum_{c_i \in C }(F(c_i) + |c_i|) \leq  \frac{k+1}{k}\sum_{c_i \in C }(F(c_i) + |c_i|)
	\end{equation}
	Similarly, if the two factors simultaneously get their minimum values, the lower bound of $S(\Lambda)$ can be computed as:
	\begin{equation}
		\label{eq:lowerboundPOA}
		\text{OPT} \geq \delta  (\frac{\sum_{i=1}^{m}|c_i|}{k})^2 + \frac{\sum_{i=1}^{m}|c_i|}{k} \geq \frac{1}{k}\sum_{c_i \in C }(F(c_i) + |c_i|)
	\end{equation}
	PoA must be no larger than the quotient of the upper bound of $S(\Lambda^*)$ (Equation~(\ref{eq:upperboundPOA})) and the lower bound of OPT (Equation~(\ref{eq:lowerboundPOA})). Thus we have $\text{PoA} = \frac{max_{S(\Lambda^*) \in \text{PNE}}S(\Lambda)}{\text{OPT}} = \frac{S(\Lambda^*)}{\text{OPT}} \leq k + 1$.
\end{proof}

\subsubsection*{Round Complexity}

\begin{theorem}
	The rounds \text{(RD)} required until Stackelberg game converges to a Nash equilibrium can be bounded by:
	\begin{equation}
	\label{eq:rntbound}
	\text{RD} \leq 2(\underbrace{\sum_{i = 1}^{\tau}d_m((\frac{k - 1}{d_m})^{1-\rho} +\frac{i-1}{|V|})^{-1}}_{tail \ part}+\underbrace{|V|(1-\sum_{i = 1}^{\xi}i^{-\rho})d_{M}}_{head \ part} + |V|)
	\end{equation}
	\label{nor}
where $\tau = |V| - |V|(d_M-\xi)d_M^{-\rho}$ is the bound of the tail vertices' number and $d_m$ (or $d_M$) is the global minimum (or maximum) degree.
\end{theorem}

\begin{proof}

Based on Equation~(\ref{eq:nomal}), the upper bound of the load balance part is $\frac{1}{k}\sum_{v \in V}d(v) + \frac{|V|}{k}$ .
The number of tail vertex can be calculated as $|V| -|V| \sum_{i=\xi+1}^{d_M}i^{-\rho} < \tau$.
So, the sum of tail vertices' degrees can be bounded by the head part. The number of head vertex can be calculated as $|V|(1-\sum_{i = 1}^{\xi}i^{-\rho})$. So we can get the sum of head vertices' degrees as the head part. The upper bound of communication part is the same as the load balance part. Since $\delta \sum_{i=1}^k{|P_i|^2}  + \sum_{i=1}^k{\Theta(P_i, V)}$ is in the integer domain, it implies that if a cluster changes its current strategy in the game, the reduction of $S(\Lambda)$ is at least $\frac{1}{k}$. So, the therom is proven.
\end{proof}

{It can  theoretically guarantee a tighter RD bound of S5P on a more skewed graph, because both the tail part and the head part monotonically decrease as $\rho$ decreases.}

\section{experiments}
We conducted experiments with S5P on real-world graphs of varying types and sizes to assess its effectiveness and scalability, aiming at answering the following  {six} key questions.
\begin{itemize}
	\item {\bf $Q_1$.} Can S5P outperform other state-of-the-art partitioning methods, including streaming, in-memory, and hybrid vertex-cut partitioners?
	\item {\bf $Q_2$.} How does each component of S5P, such as clustering and stackelberg gaming, affect the partitioner's performance?
	\item {\bf $Q_3$.} What is the effect of proposed optimization techniques, such as sketching and parallelization, particularly concerning time and space efficiency?
	\item {\bf $Q_4$.} What is the influence of S5P's parameters on its overall performance and effectiveness?
	\item {\bf $Q_5$.} How much improvement does the deployment of \emph{S5P} bring to real distributed graph systems, e.g., PowerGraph?
	\item {\bf $Q_6$.}What is the relationship between the skewness of graph and the performance improvement achieved by S5P.
\end{itemize}

\label{sec:experiment}

\subsection{Experimental Setup}
\subsubsection*{Evaluation Platform} We perform all experiments on a server with $2 \times 32$ Intel(R) Xeon(R) Platinum 8358 CPU @ 2.60GHz, 377 GiB of main memory. We use Ubuntu 20.04.6 LTS as our operating system. To test the partitioning quality on a real distributed environment, we use docker to simulate 32 computing nodes equipped with PowerGraph \cite{powergraph}, and allocate one CPU for each computing node.
Each reported value is the average of 10 runs.
\begin{table}[h!]
	\tiny
	\centering
	\caption{ Details of Graphs}
	\label{ret:dataset}
	\begin{tabular}	{c|c|c|c|c|c}
		\hline
		Graphs & $|V|$ & $|E|$ & Size & Type& {($\rho$,  $\rho_1$, $\rho_2$, $\rho_3$)}\\
		\hline
		OK(com-orkut)  & 3.1M & 117M & 1.7GiB & Social &{(2.13, 0.49, 0.61, 108M)}\\
		TW(twitter-2010)  & 42M & 1.5B & 25.0GiB & Social&{(1.43, 0.03, 0.07, 1.3B)} \\
		FR(com-friendster)  & 66M & 1.8B & 31.0GiB  & Social&{(2.56, 0.39, 1.00, 1.6B)} \\
		LJ(com-livejournal) & 4M & 35M & 479MiB & Social&{(2.40, 0.38, 0.79, 227M )}  \\
		\hline
		IT(it-2004)  & 41M & 1.2B & 19.0GiB & Web&{(1.74, 0.06, 0.13, 1B)}  \\
		{UK7(uk-2007-05)} & 106M & 3.7B & 63.0GiB  & Web&{(1.31, 0.10, 0.20, 3.4B)} \\
		IN(in-2004) & 1M & 16M & 231MiB & Web&{(1.36, 0.15, 0.31, 12M)}  \\
		SK(sk-2005) & 51M & 1.9B & 32.0GiB & Web&{(1.11, 0.04, 0.07, 1.8B)} \\
			{UK2(uk-2002)} &  {18M} & {298M} & {4.7GiB} & {Web}&{(2.06, 0.21, 0.38, 243M)} \\
			{AR(arabic-2005)} & {23M} & {639M} & {11.0GiB} & {Web}&{(1.62, 0.10, 0.19, 572M)} \\
			{WB(webbase-2001)}  & {118M} & {1B} & {17.2GiB} & {Web}&{(2.21, 0.11, 0.23, 665M)} \\
		\hline
		{R-MAT-$G_{1}$}  & {1.04M} & {314M} & {3.0GiB} &{Synthetic} & {(0.89, 0.15, 0.44, 102M)}\\
		{R-MAT-$G_{2}$} & {1.04M} &  {629M} & {5.5GiB} &{Synthetic} & {(0.87, 0.17, 0.48, 626M)}\\
		{R-MAT-$G_{3}$}  & {1.04M} &  {1.04B} & {8.6GiB} &{Synthetic} &  {(0.84, 0.19, 0.52, 1B)}\\
		{R-MAT-$G_{4}$}  & {67.1M} &  {671M} & {10.1GiB} &{Synthetic} &{(1.16, 0.048, 0.145, 469M)}\\
		{R-MAT-$G_{5}$} & {67.1M} &  {2.01B} & {30.1GiB} &{Synthetic} & {(1.11, 0.051, 0.152, 1B)}\\
		{R-MAT-$G_{6}$}  & {67.1M} & {3.36B} & {49.8GiB} &{Synthetic}  & {(1.07, 0.053, 0.157, 3B)}\\

		\hline
	\end{tabular}
\end{table}

\subsubsection*{Graph Datasets}
We employ {11} distinct graphs characterized by varying sizes, originating from diverse web repositories, and obtained through independent web crawling efforts by distinct organizations. These are social network graphs (OK~\cite{snap,fr}, TW~ \cite{snap, tw}, FR~\cite{snap, fr}, and LJ~\cite{snap, fr}), as well as web graphs (IT~\cite{wg, llp, craw}, {UK7}~\cite{wg, llp, craw}, IN~\cite{wg, llp, craw}, SK~\cite{wg, llp, craw}, {UK2}~\cite{wg, llp, craw}, {AR}~\cite{wg, llp, craw}, and  {WB}~\cite{wg, llp, craw}).
We select the real-world graphs for comprehensively evaluating the partitioning strategies, which provide: 1) standardized test cases, 2) scalability challenges of big data, and 3) diversity in graph topology and partitioning difficulty.
First, many of the graphs are established benchmarks commonly used to assess partitioning algorithms in academic literature. Using widely adopted datasets enables fair comparison to prior works.
Second, most graphs are massive in scale, comprising billions of edges. Processing networks of this size poses considerable memory and computational challenges. Our experiments show that not all partitioning techniques scale effectively to such large graphs on our test platform.
Third, the graphs exhibit diversity in structure and complexity. This heterogeneity stringently evaluates partitioning methods.
For certain datasets, even the most proficient partitioning algorithms result in relatively elevated replication factors. {We also use R-MAT (TrillionG)~\cite{rmat, chakrabarti2004r} to generate $6$ big graphs with varying amounts of skewness to get deeper insights into the performance of skewness-aware partitioners}.

\subsubsection*{Competitors}
We compare \emph{S5P} to 8 of the most recent and best
\emph{streaming} and \emph{offline} partitioners. From the line of \emph{streaming}
partitioners, we compare S5P to HDRF~\cite{hdrf}, Greedy~\cite{powergraph}, DBH~\cite{dbh}, 2PS-L~\cite{2ps-l},
and CLUGP~\cite{clugp}. From the line of offline partitioners, we compare to NE~\cite{ne}, METIS~\cite{metis}, and HEP~\cite{hep}. {Additionally, we compare S5P with 3 game-based approaches, denoted as RMGP~\cite{rmgp}, MDSGP~\cite{mdsgp}, and CVSP~\cite{c+cv}, which are related to the streaming vertex-cut partitioning.}

\subsubsection*{Implementation\footnote{Our source code and links to all datasets are available online (\url{https://github.com/BearBiscuit05/S5P}).
}} When possible, we utilize the official reference implementations provided by the authors of~\cite{dbh, powergraph,cusp,2ps-l,clugp,ne,dne,metis,hep}. For HDRF, we employ the improved version in~\cite{2ps-l}, since the original version \cite{hdrf} does not perform well on large graphs.
DBH, {RMGP, and MDSGP have} no public reference implementation, so we use the re-implementation of DBH~\cite{2ps-l} {and re-implement RMGP, MDSGP, and CVSP}. We adopt the C++ versions in our experiments and also provide the Java ports in our open-source repository. We configure all algorithm parameters based on recommendations from the respective papers.

\subsection{Performance ($Q_1$)}
\begin{figure}[h!]
	
	\centering
		\subfigure[{FR: Replication Factor}] {\includegraphics[width= 0.32\columnwidth]{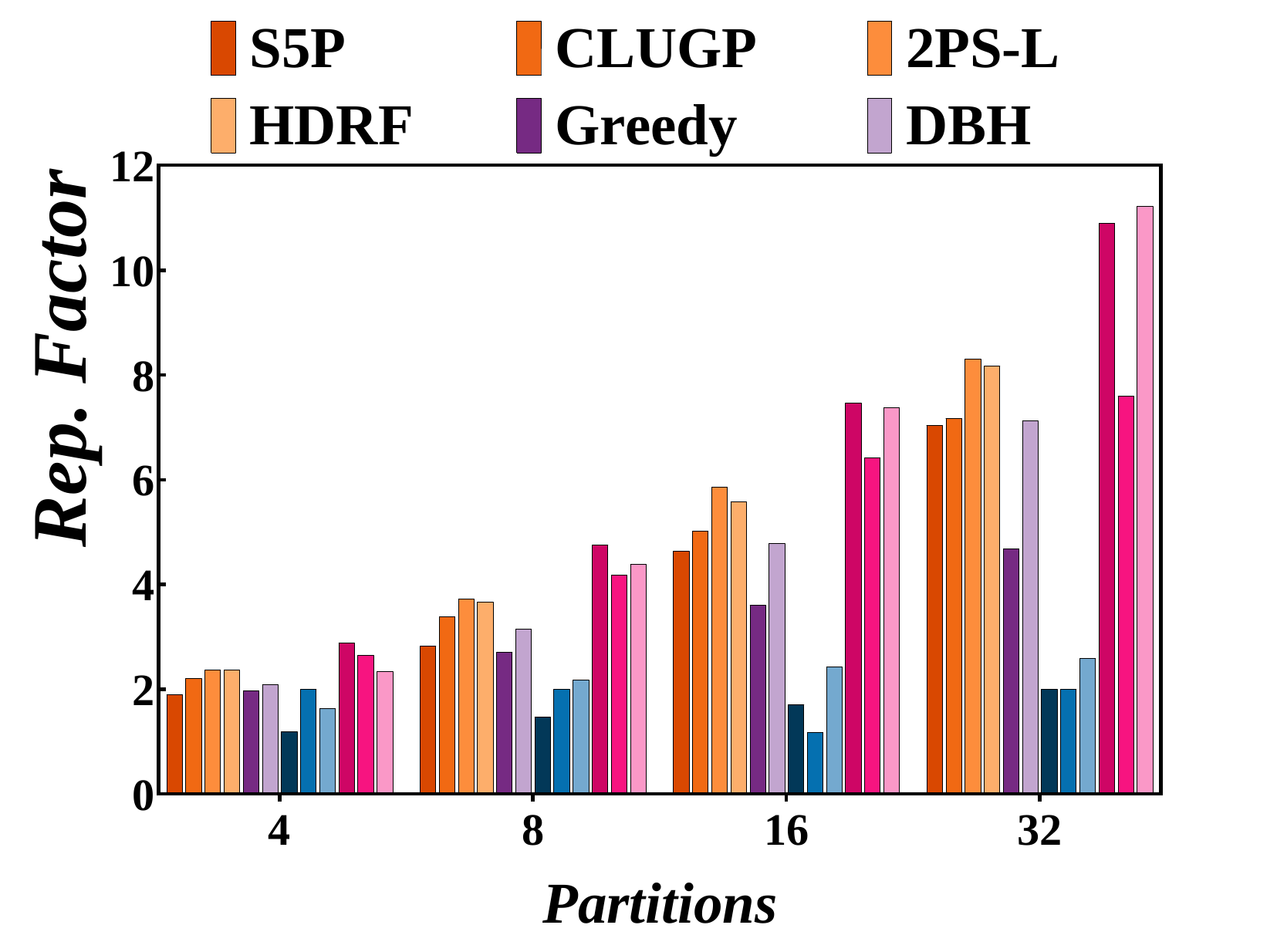}}
		\subfigure[{UK7: Replication Factor}] {\includegraphics[width= 0.32\columnwidth]{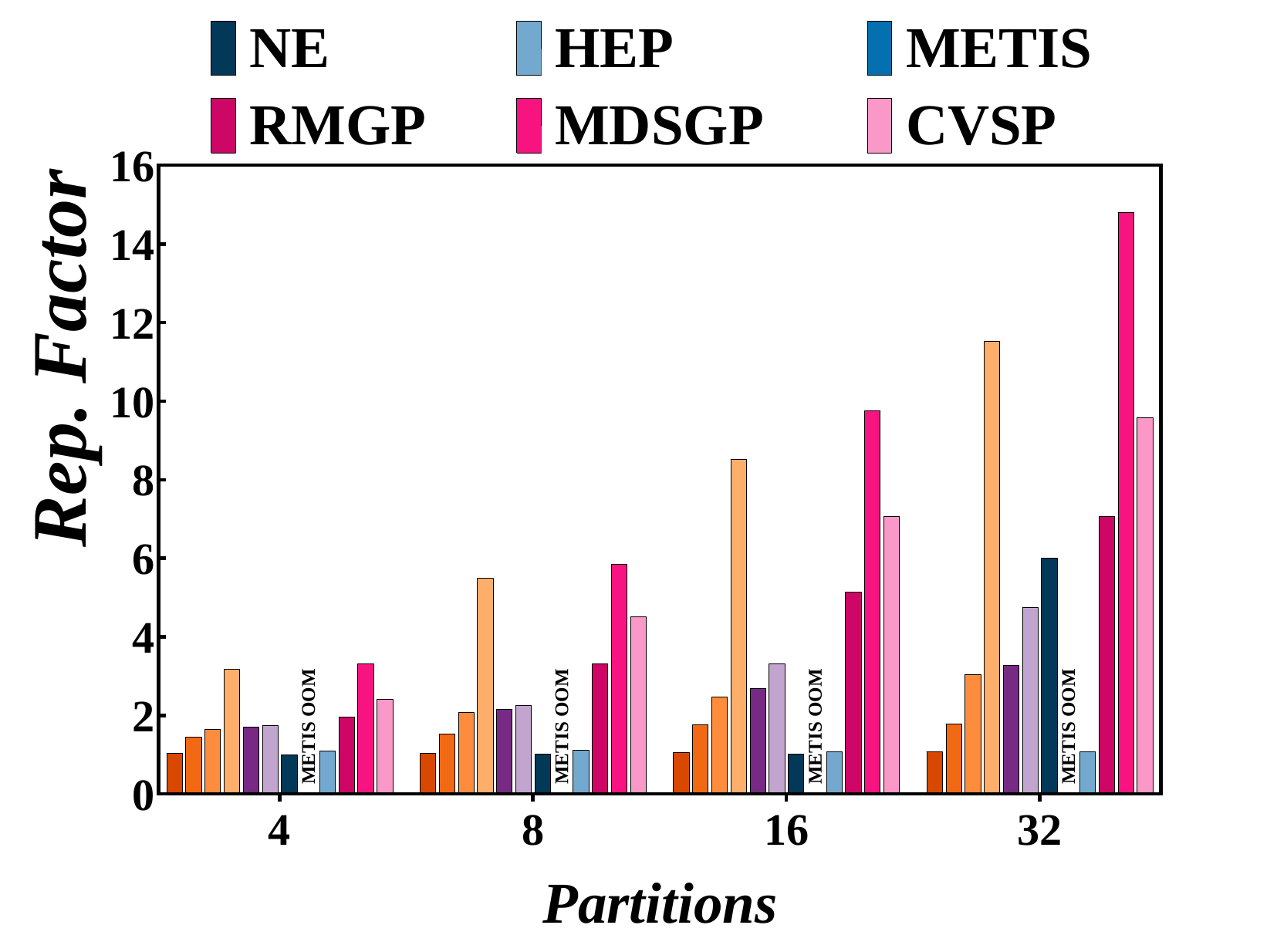}}
		\subfigure[{FR: Run-time}] {\includegraphics[width= 0.32\columnwidth]{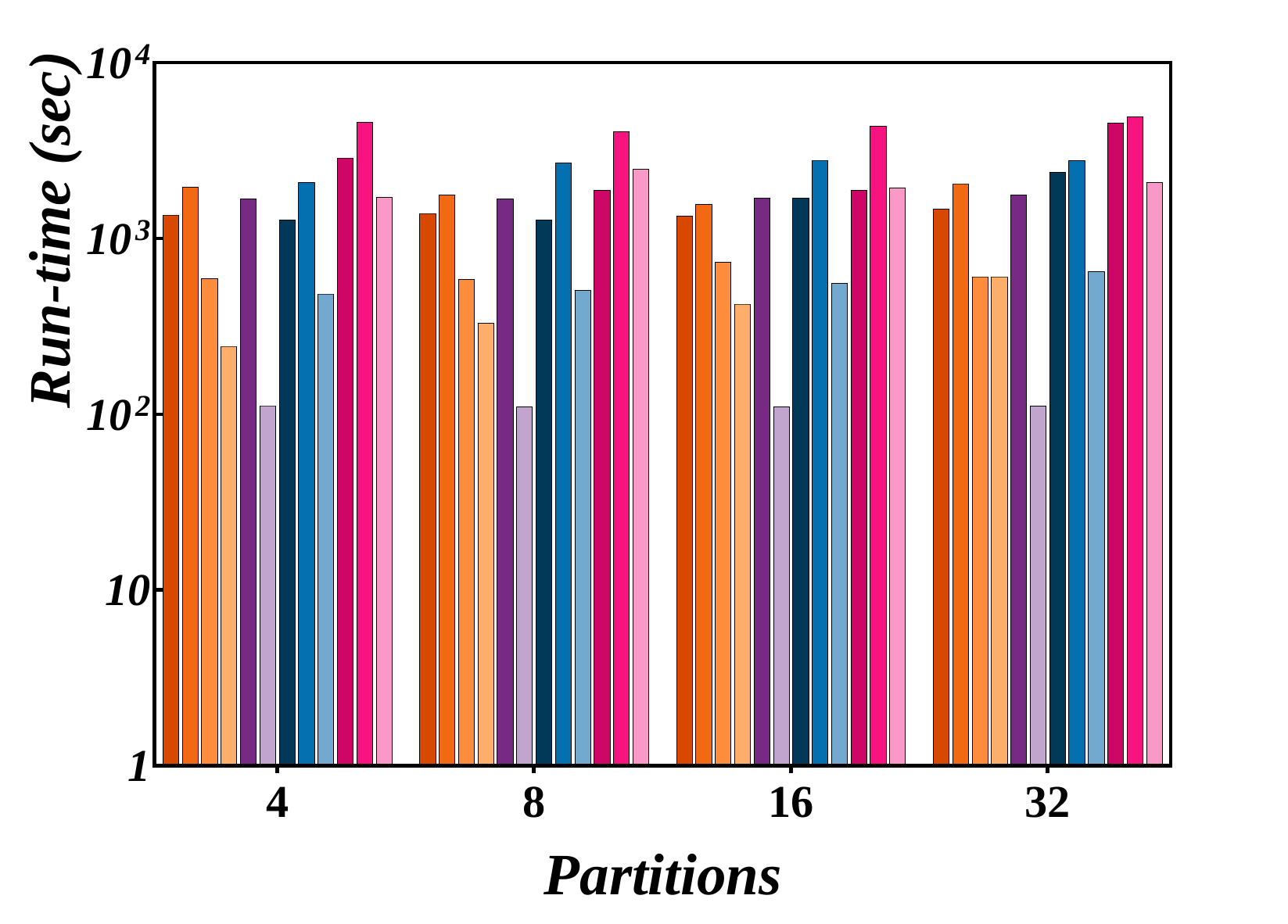}}
		\subfigure[{UK7: Run-time}] {\includegraphics[width= 0.32\columnwidth]{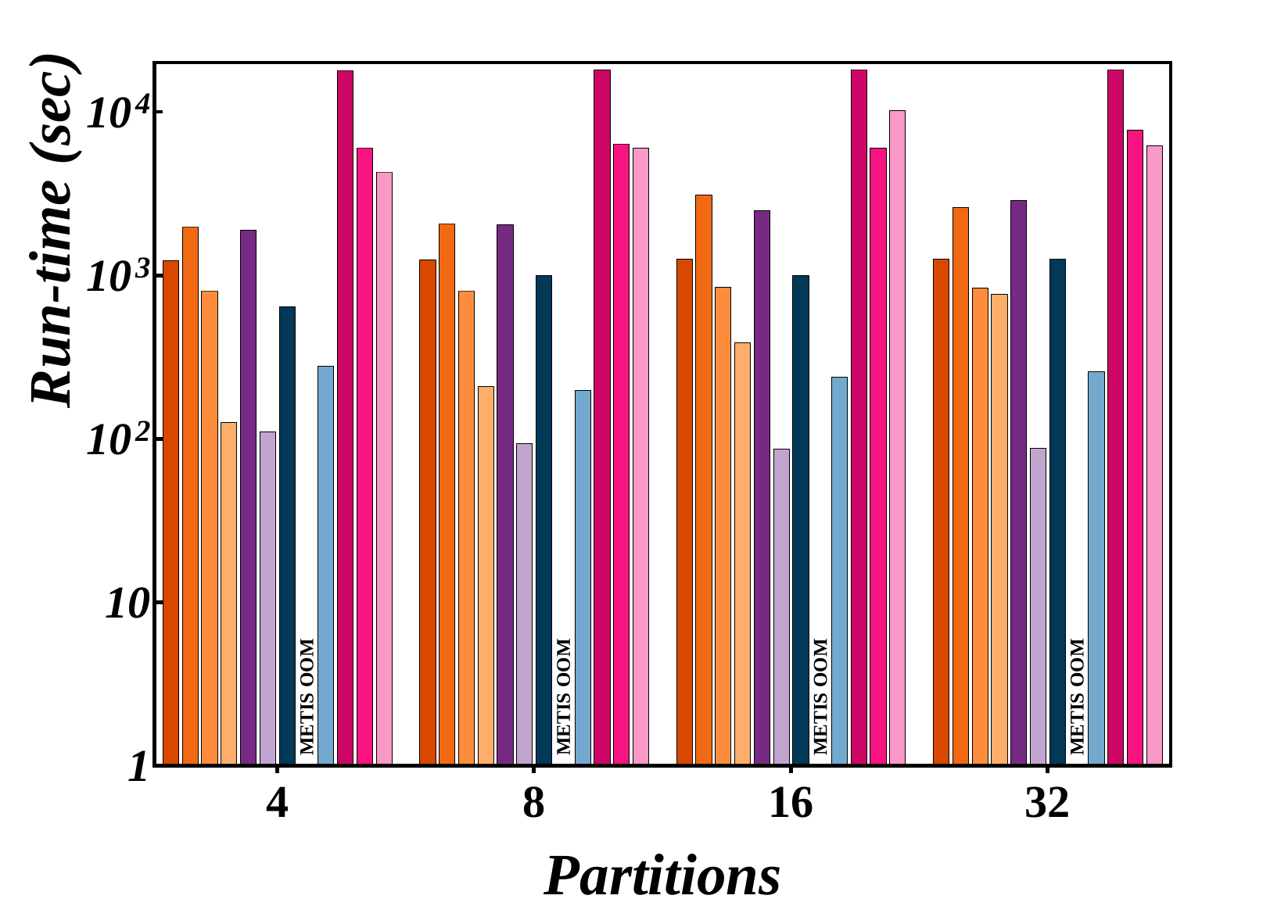}}
		\subfigure[{FR: Memory Overhead}] {\includegraphics[width= 0.32\columnwidth]{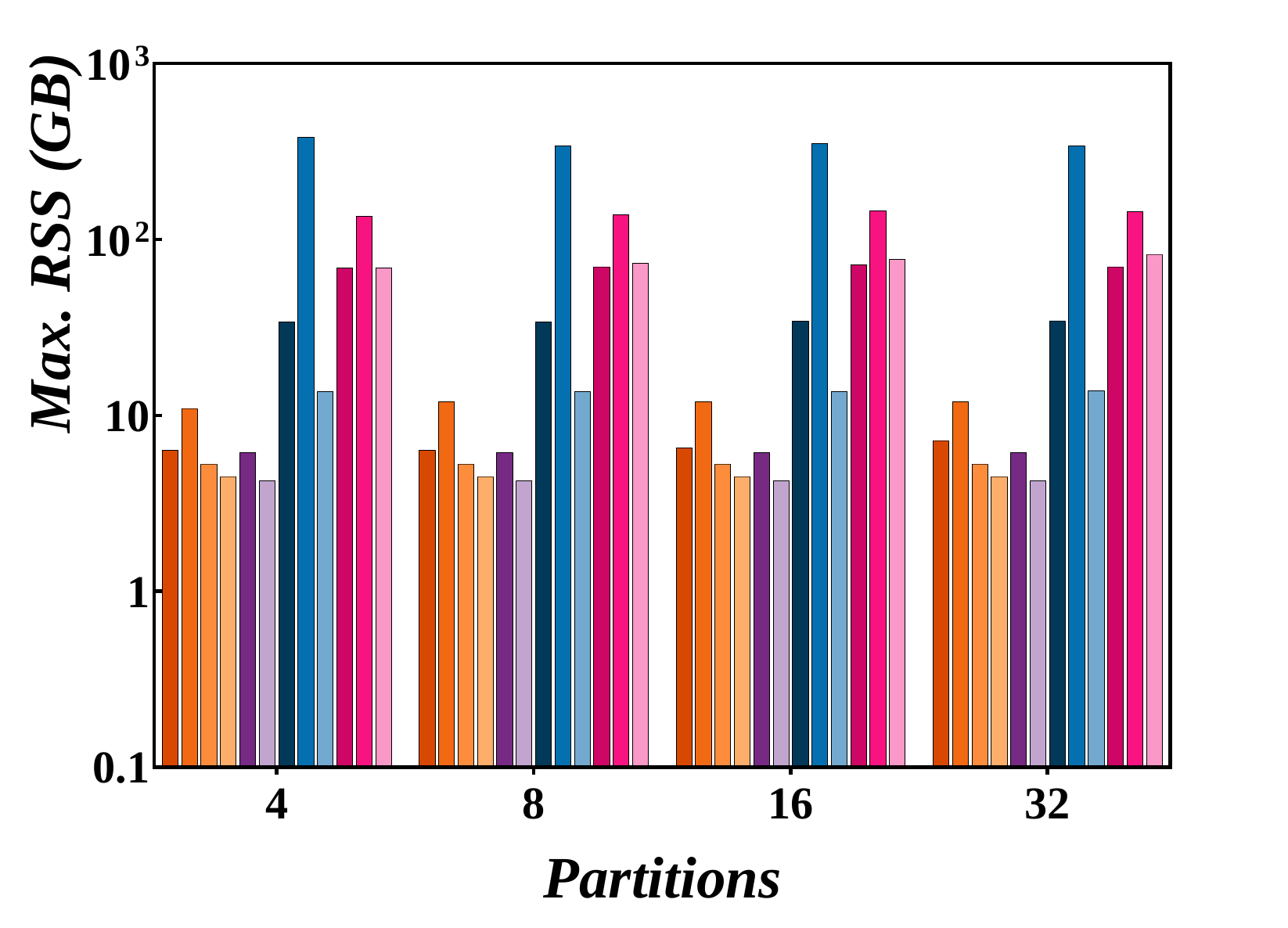}}
		\subfigure[{UK7: Memory Overhead}] {\includegraphics[width= 0.32\columnwidth]{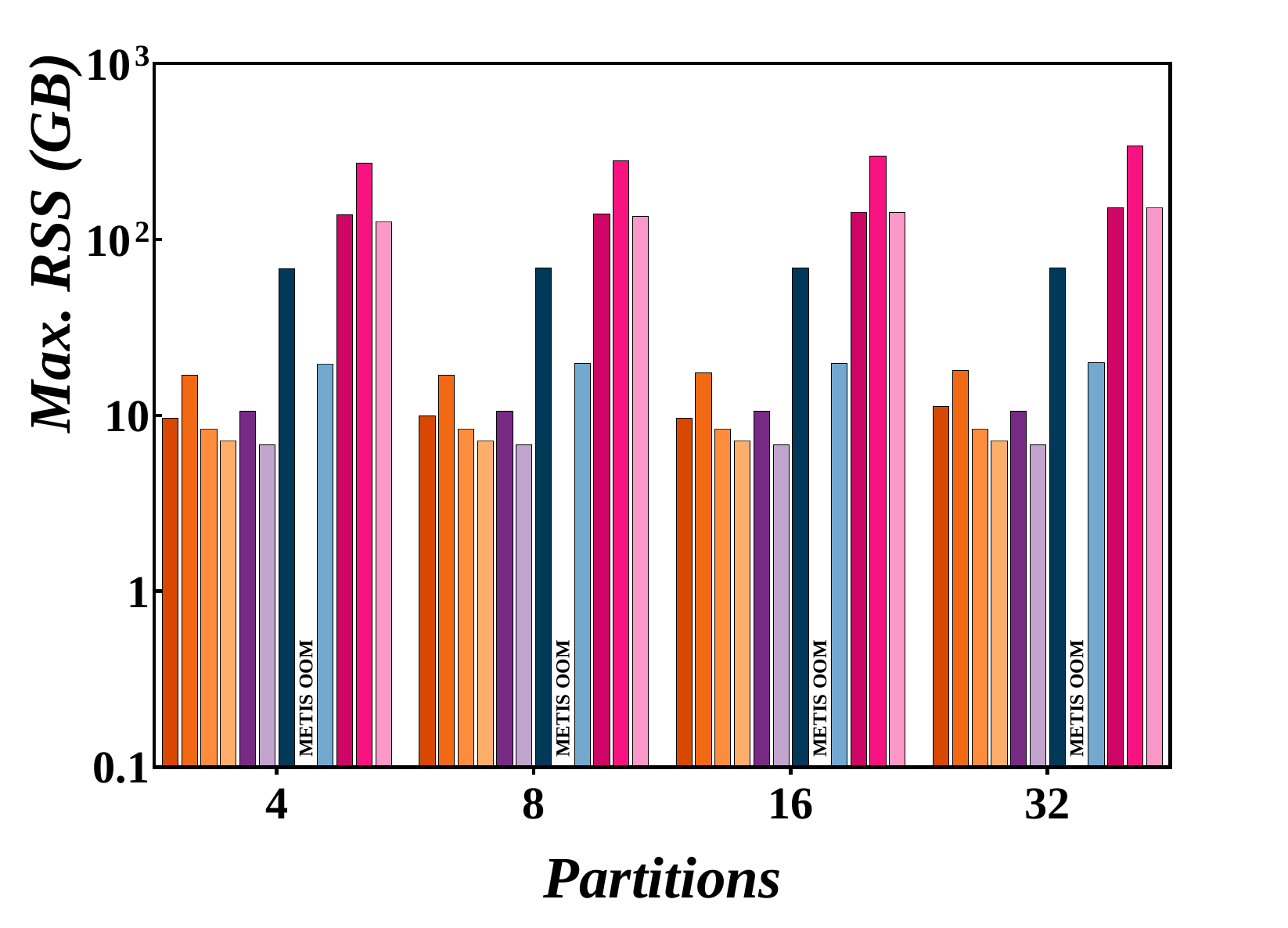}}
		\caption{Performance on Real-world Graphs} 
		\label{fig:performances}
\end{figure}

\begin{table*}[h]
    \tiny
	\centering
	\caption{Replication Factor of Different Graphs}
	\begin{tabular}{|c|c|c|c|c|c|c|c|c|c|c|c|c|}
		\hline
		\multirow{2}{*}{\diagbox{\emph{Graph}}{\emph{Partitioner}}} & \multicolumn{3}{c|}{CLUGP} & \multicolumn{3}{c|}{2PS-L} & \multicolumn{3}{c|}{HDRF} & \multicolumn{3}{c|}{S5P} \\
		\cline{2-13}
		& $k$:64 & $k$:128 & $k$:256 & $k$:64 & $k$:128 & $k$:256 & $k$:64 & $k$:128 & $k$:256 & $k$:64 & $k$:128 & $k$:256 \\
		\hline
		\hline
		OK &14.288  &17.522  &20.636  &15.112 &18.915 &23.200  &17.860  &22.617  & 27.023 & \textbf{11.614} & \textbf{15.391}  & \textbf{19.055} \\
		\hline
		TW & 8.808 & 10.817 & 11.861&10.642 &13.074  &15.577 &9.520  &11.789  &14.408  & \textbf{7.583} & \textbf{9.068} &\textbf{10.526}  \\
		\hline
		FR & 10.311 & 13.432 & 17.011 & 11.241  & 14.359 & 17.457 & 11.324 & 14.757 & 18.122 & \textbf{7.870} & \textbf{11.244} &\textbf{14.995}  \\
		\hline
		LJ &4.913  &5.471  &5.945&5.036  &5.593  &6.045  &6.778  &7.763  &8.545  &\textbf{4.549}  &\textbf{5.112}  & \textbf{ 5.636} \\
		\hline
		IT & 1.908  & 1.973 &2.041 &3.680  &4.110  &4.420  &12.538  &14.500  &16.469  & \textbf{1.273}  & \textbf{1.232} &  \textbf{1.210} \\
		\hline
		{UK7} & 1.754 & 1.876 & 1.839 &3.338  & 3.760 & 4.077 &14.190  & 16.700 & 19.181 & \textbf{1.265} & \textbf{1.213} & \textbf{1.196} \\
		\hline
		IN & 1.415 &1.542 & 1.621 &1.895  &2.241  &2.887  &6.884  &8.028  &8.890  & \textbf{1.229} & \textbf{1.207} & \textbf{1.225} \\
		\hline
		SK & 2.299 & 2.584 & 2.566 & 4.001 & 5.466 &7.029  & 16.561 & 19.413 &  21.766 &  \textbf{1.337} & \textbf{1.310} & \textbf{ 1.293} \\
		\hline
		{UK2} &  {1.561} &  {1.698} &  {1.692} & {2.644} &  {2.752} &  {2.921} &  {9.414} & {10.673} &  {11.791} &   {\textbf{1.371}} & {\textbf{1.227}} &  {\textbf{1.238}} \\
		\hline
		{AR} &  {2.015} &  {1.929} & {2.005}  &  {3.409} &  {3.803} &  {4.119} &  {12.599} &  {14.768} &  {16.762} &   {\textbf{1.131}} &  {\textbf{1.213}} &  {\textbf{1.233}} \\
		\hline		
		 {WB} &  {1.446} &  {1.493} &  {1.485} &  {1.829} &  {1.836} &  {1.822} &  {5.951} &  {6.646} &  {7.283} &   {\textbf{1.296}} &  {\textbf{1.178}} &  {\textbf{1.188}} \\
		\hline		
	\end{tabular}
	\label{tab:rfresult}
\end{table*}

	\begin{table*}[h]
			\centering
			\tiny
			\centering
			\caption{\small {Game-based Methods {\scriptsize (Time/sec | Memory/GB | Hour(s): h | $k$ = 32)}} }
			\begin{tabular}{|c|c|c|c|c|c|c|c|c|c|c|c|c|c|c|c|}
				\hline
				\multirow{2}{*}{\diagbox{ {\emph{G. }}}{ {\emph{Par.}}}} & \multicolumn{3}{c|}{ {RMGP}} & \multicolumn{3}{c|}{ {MDSGP}} & \multicolumn{3}{c|}{ {CVSP}} & \multicolumn{3}{c|}{ {CLUGP}} & \multicolumn{3}{c|}{ {S5P}} \\
				\cline{2-16}
				&  {RF} &  {Time} &  {Mem.} 	&  {RF} &  {Time} &  {Mem.}  &  {RF} &  {Time} &  {Mem.}  &  {RF} &  {Time} &  {Mem.}  &  {RF} &  {Time} &  {Mem.}  \\
				\hline
				\hline
				 {OK}  & {16.7} & {535} & {4.01}
				& {9.9}  & {324}
				& {8.95} 			
				& {17.4}  & {141}  &  {2.25}
				& {10.7}  & {91}  & {1.02}
				& \textbf{ {8.5}} & \textbf{ {60}}  & \textbf{ {0.38}} \\
				\hline
				 {TW} &  {-} &  {>24h}  & {48.70} &  {6.8}& {5189} & {99.08}  & {-}  & {>24h}  & {56.01}
				& {7.6}  & {1333}  & {11.65}
				&\textbf{ {6.0}} & \textbf{ {808}} &\textbf{ {4.64}}  \\
				\hline
				 {FR} &  {10.9} &  {4553} &  {70.20} &  {7.6}  &  {4934} &  {144.96} &  {11.2} &  {2078} &  {80.69}
				& {7.2}  & {3045}  & {14.12}
				& {\textbf{7.0}} & \textbf{ {1466}} &\textbf{ {7.22}}  \\
				\hline
				 {LJ} & {5.4}  & {65}  & {2.08}& {4.5}  & {184}  & {3.83}  & {5.7}  & {32}  & {2.25}
				& {4.2}  & {111}  & {1.11}
				
				&\textbf{ {3.9}}  &\textbf{ {28}}  & \textbf{ {0.48}} \\
				\hline
				 {WB} & {4.2}  & {1871}  & {61.10}& {6.2}  & {6320}  & {119.45}  & {4.8}  & {822}  & {79.46}
				& {1.5}  & {1101}  & {25.11}
				&\textbf{ {1.1}}  &\textbf{ {696}}  & \textbf{ {12.90}} \\
				\hline
				 {$G_6$} & {-}  & {>24h}  & {115.5}& {4.9}  & {11915}  & {231.87}  & {-}  & {>24h}  & {110.8}
				& {4.8}  & {4847}  & {18.01}
				&\textbf{ {4.4}}  &\textbf{ {2620}}  & \textbf{ {8.06}} \\
				\hline
				
			\end{tabular}
			\label{tab:gameresult}
		\end{table*}
	\begin{table*}[h]
		\begin{minipage}{\linewidth}
			\centering
			\tiny
			\caption{ {Optimality (Opt.)}}
			\label{ret:approximate}
			\begin{tabular}{|c|c|c|c|c|c|c|c|c|c|}
				\hline
				\multirow{2}{*}{\diagbox{ {\emph{$G_{\alpha/\beta/\gamma}$(|V|, |E|)} [Opt.]}}{ {\emph{Partitioner}}}} & \multicolumn{2}{c|}{ {CLUGP}} & \multicolumn{2}{c|}{ {2PS-L}} &  \multicolumn{2}{c|}{ {S5P}}
				\\
				\cline{2-7}
				&  {RF} &  {$\alpha$} &  {RF} &  {$\alpha$} &  {RF} &  {$\alpha$}
				\\
				\hline
				\hline
				 {$G_\alpha$(7, 12) {[\textbf{1.43}]} } & {1.86}  & {1.30}  & {2.00}  & {1.41} & {\textbf{1.71}} & {\textbf{1.20}}
				\\
				\hline
				 {$G_\beta$(8, 15) {[\textbf{1.63}]}} & {2.38}  & {1.46}  & {2.38}  & {1.46} & {\textbf{2.12}} & {\textbf{1.30}}
				\\
				\hline
				 {$G_\gamma$(10, 12) {[\textbf{1.30}]}} & {1.90}  & {1.46}  & {2.00}  & {1.54} & {\textbf{1.80}} & {\textbf{1.38}}
				\\
				\hline
			\end{tabular}
		\end{minipage}
	\end{table*}

Figures \ref{fig:performances} (a-f) report the performance under the metrics of replication factor, runtime cost, and space cost, for all $12$ partitioners, on UK7 and FR. Three key observations can be drawn.

(1) Under the same constraints of load balancing, S5P achieves lower replication factors than all other streaming partitioners, like HDRF, 2PS-L, DBH, and CLUGP. On average, CLUGP's replication factor is  {10\%} higher, 2PS-L's is  {72\%} higher, and HDRF's is  {3$\times$} higher, than S5P. Although offline partitioners and hybrid method yields slightly better RF, their time and space overheads are much higher.

(2) We compare the time cost of different methods, including its
scalability w.r.t. the number of partitions, in Figures~\ref{fig:performances} (c) and (d).
It shows that the time cost of HDRF increases significantly as the increase of the number of partitions.
In contrast, results of S5P and others are not sensitive to the number of partitions.
For example, when the number of partitions varies from 4 to 32, the runtime cost of S5P merely increases about  {$2$\%, i.e., from 1230 to 1254 seconds (in UK7)}.
The result also shows that DBH is faster than S5P, but its quality is much lower (cf. Figure~\ref{fig:performances}).

(3) We measure S5P's memory consumption and its scalability w.r.t. the number of partitions
in Figure \ref{fig:performances} (e) and (f). It shows that offline methods consume the largest amount of memory. The widely-recognized benchmark METIS, for instance, utilizes roughly {$50$}$\times$ more memory than S5P. {On UK7, METIS runs out of memory (OOM).}
For the hybrid method HEP, the memory cost is more than $2$ times higher than that of S5P.
}

\begin{figure*}[h!]
	
	\centering
	\begin{minipage}{0.8\textwidth}
		\subfigure[Time S5P- vs. Edge-Clustering]
		{\includegraphics[width= 0.24\columnwidth]{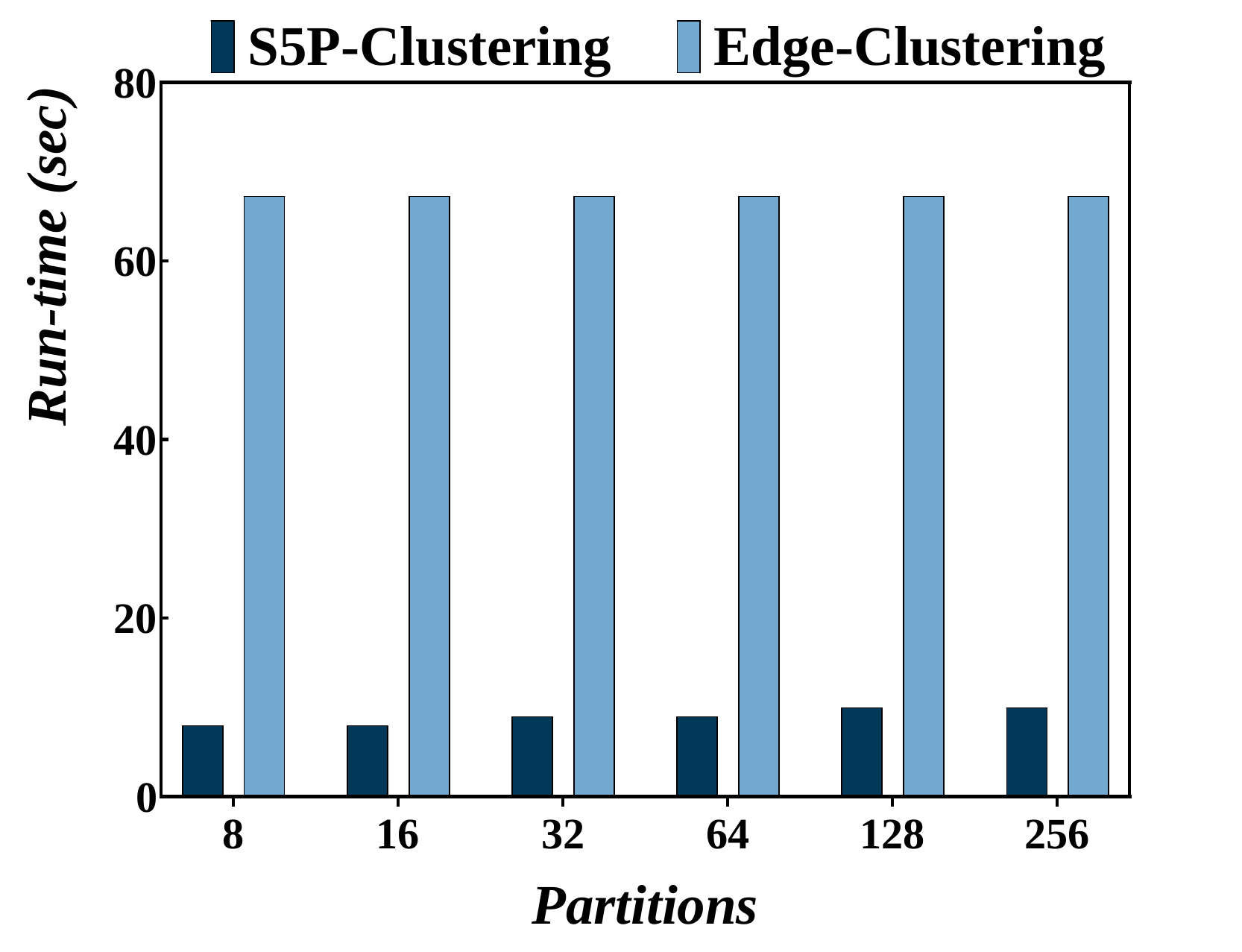}}
		\subfigure[Mem. S5P- vs. Edge-Clustering]
		{\includegraphics[width= 0.24\columnwidth]{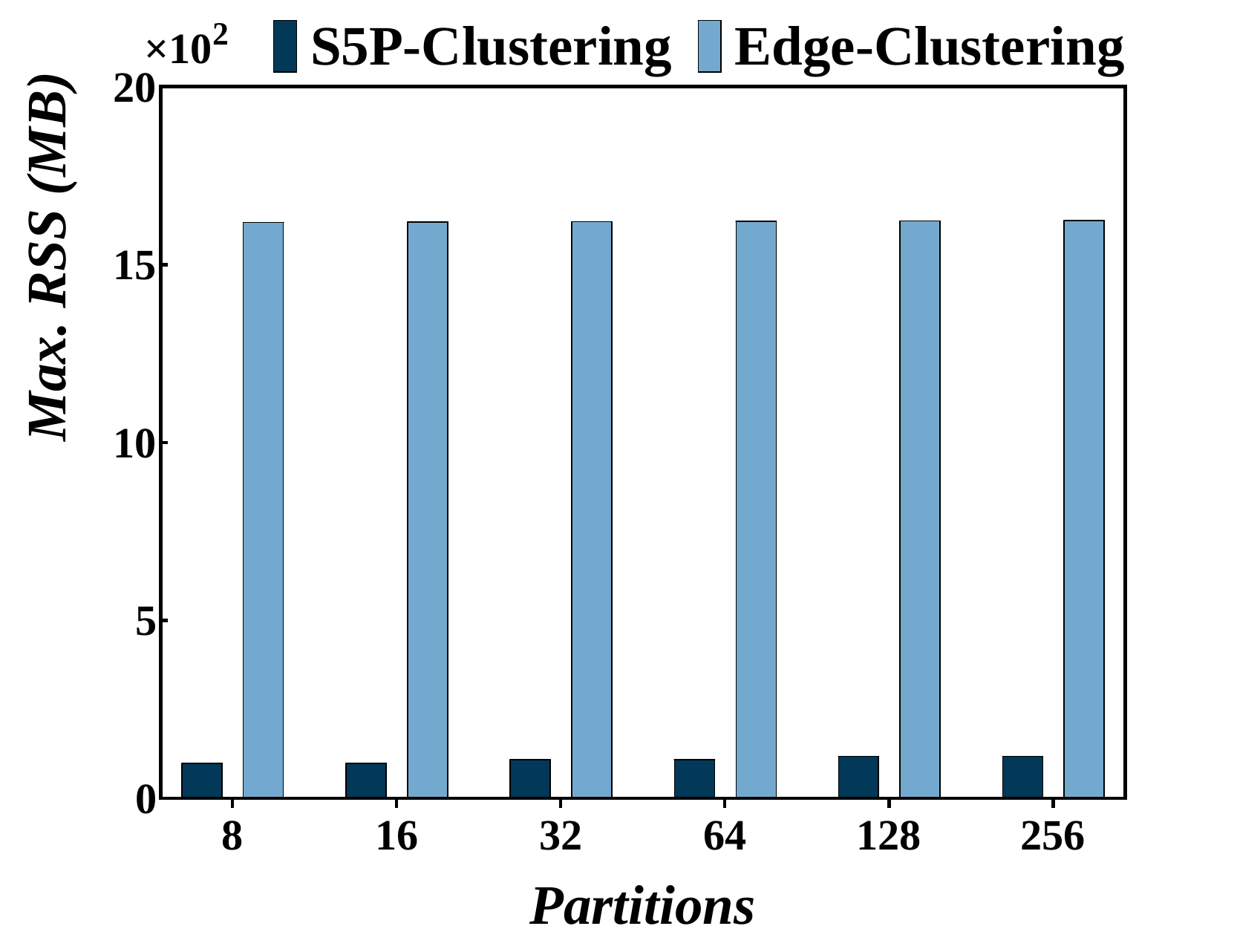}}
		\subfigure[RF w/ vs. w/o Clustering]
		{\includegraphics[width= 0.24\columnwidth]{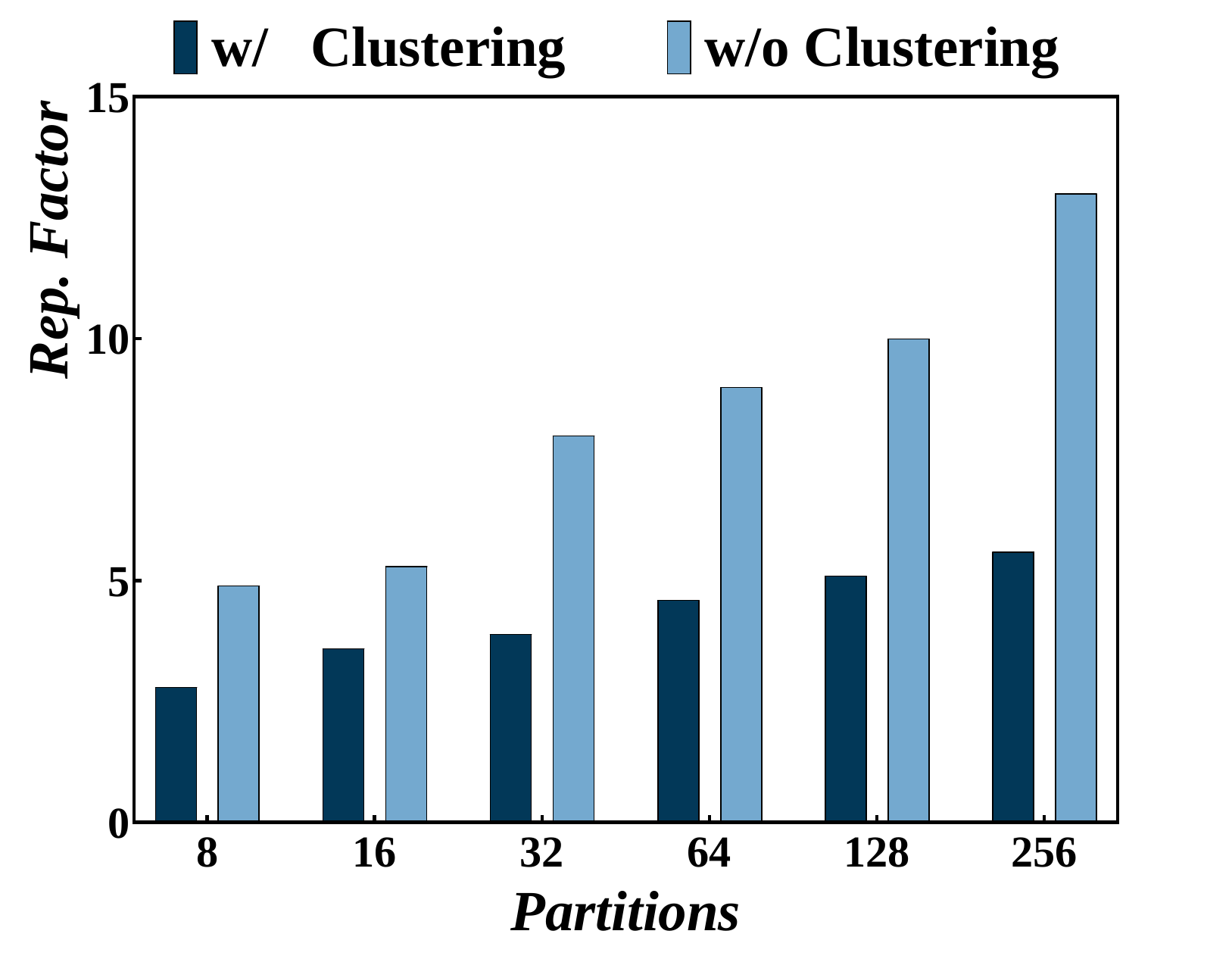}}
		\subfigure[w/ vs. w/o Stackelberg Game]
		{\includegraphics[width= 0.24\columnwidth]{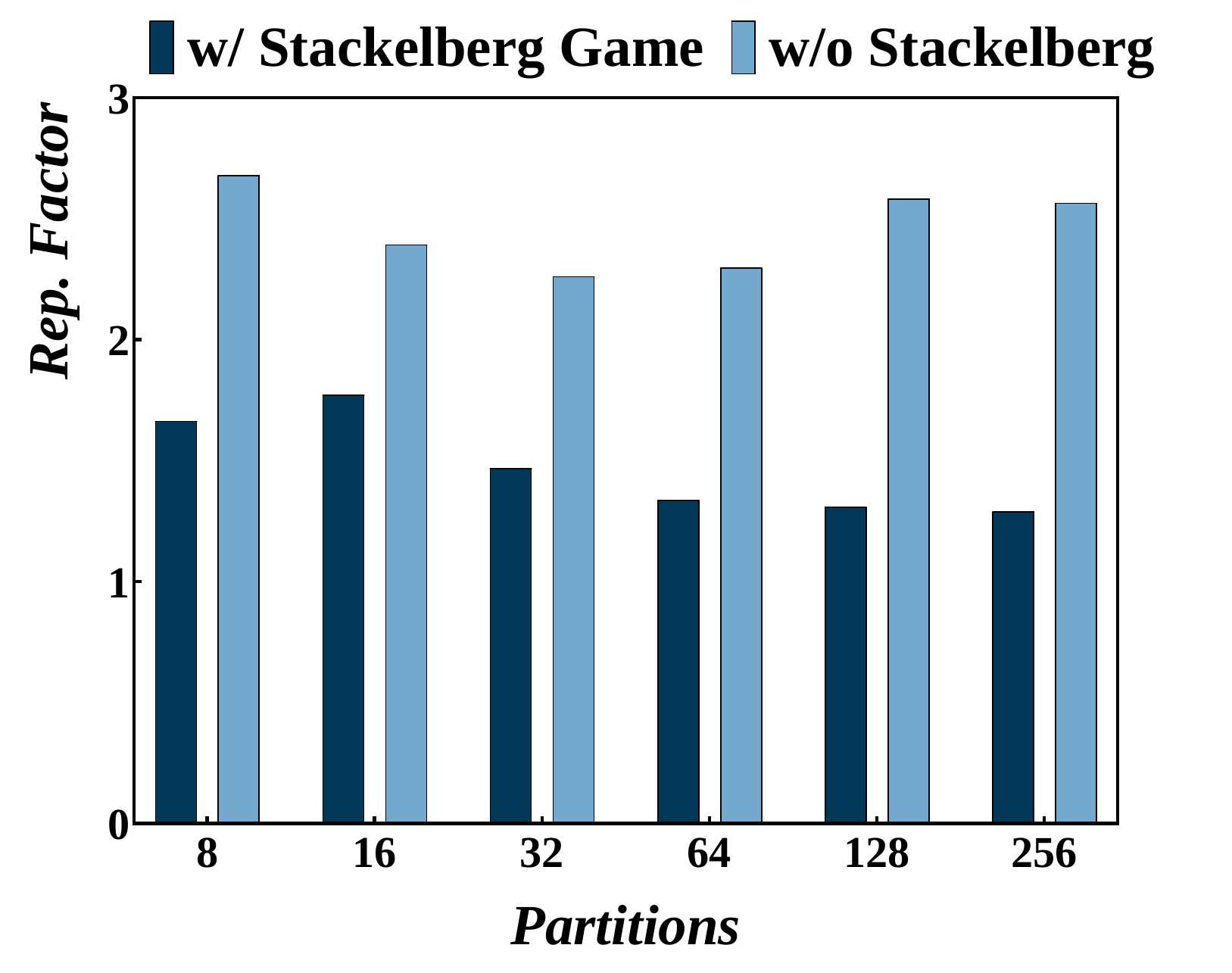}}
		\caption{Component Analysis}
		\label{fig:ablation}
	\end{minipage}%
	\begin{minipage}{0.2\textwidth}
			{\includegraphics[width= 1\columnwidth]{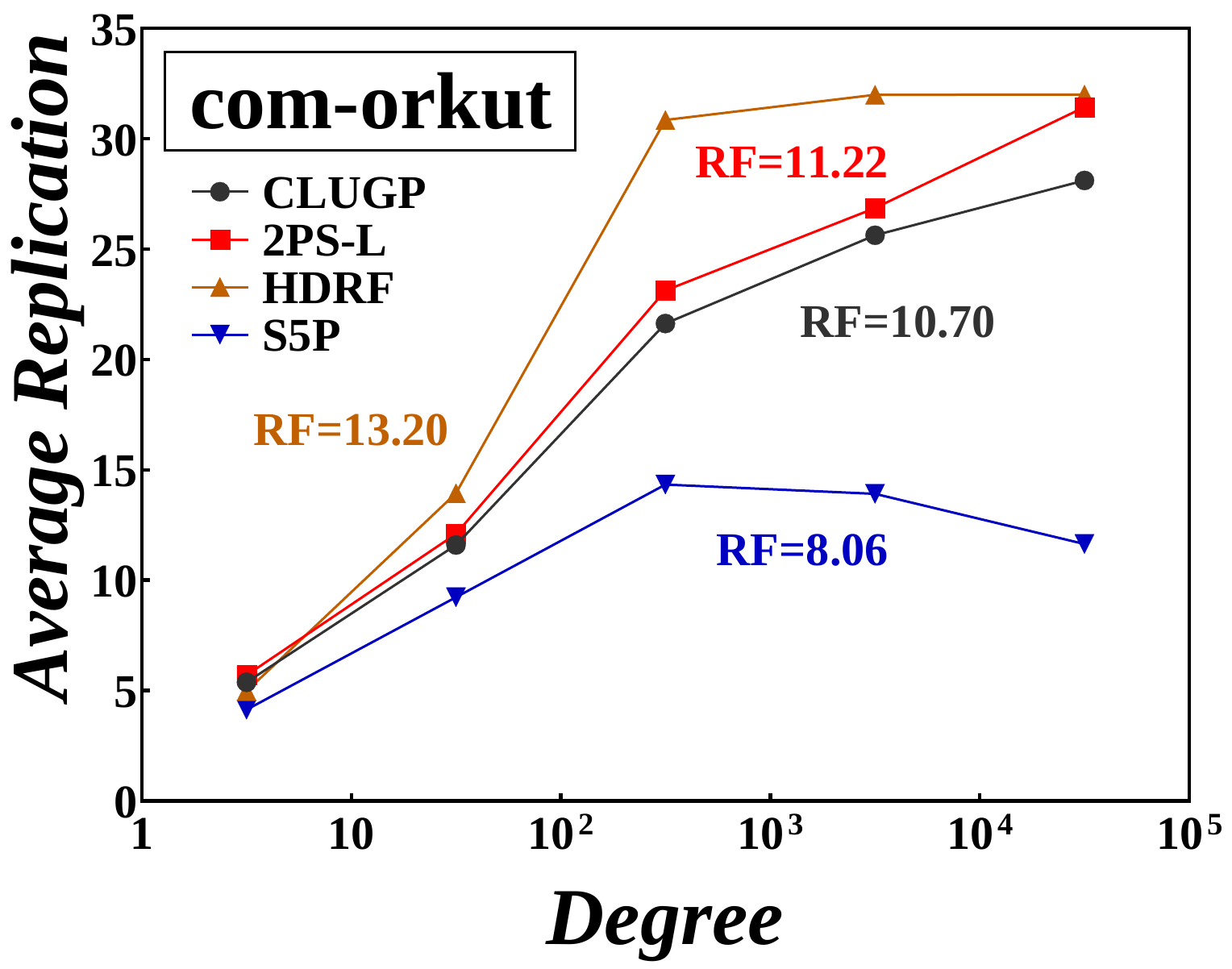}}
			\caption{Degree Analysis ($k$=32)}
			\label{fig:degree}
	\end{minipage}
		\begin{minipage}{1\textwidth}
		\centering
			\subfigure[Sketch (Time)]
			{\includegraphics[width= 0.24\columnwidth]{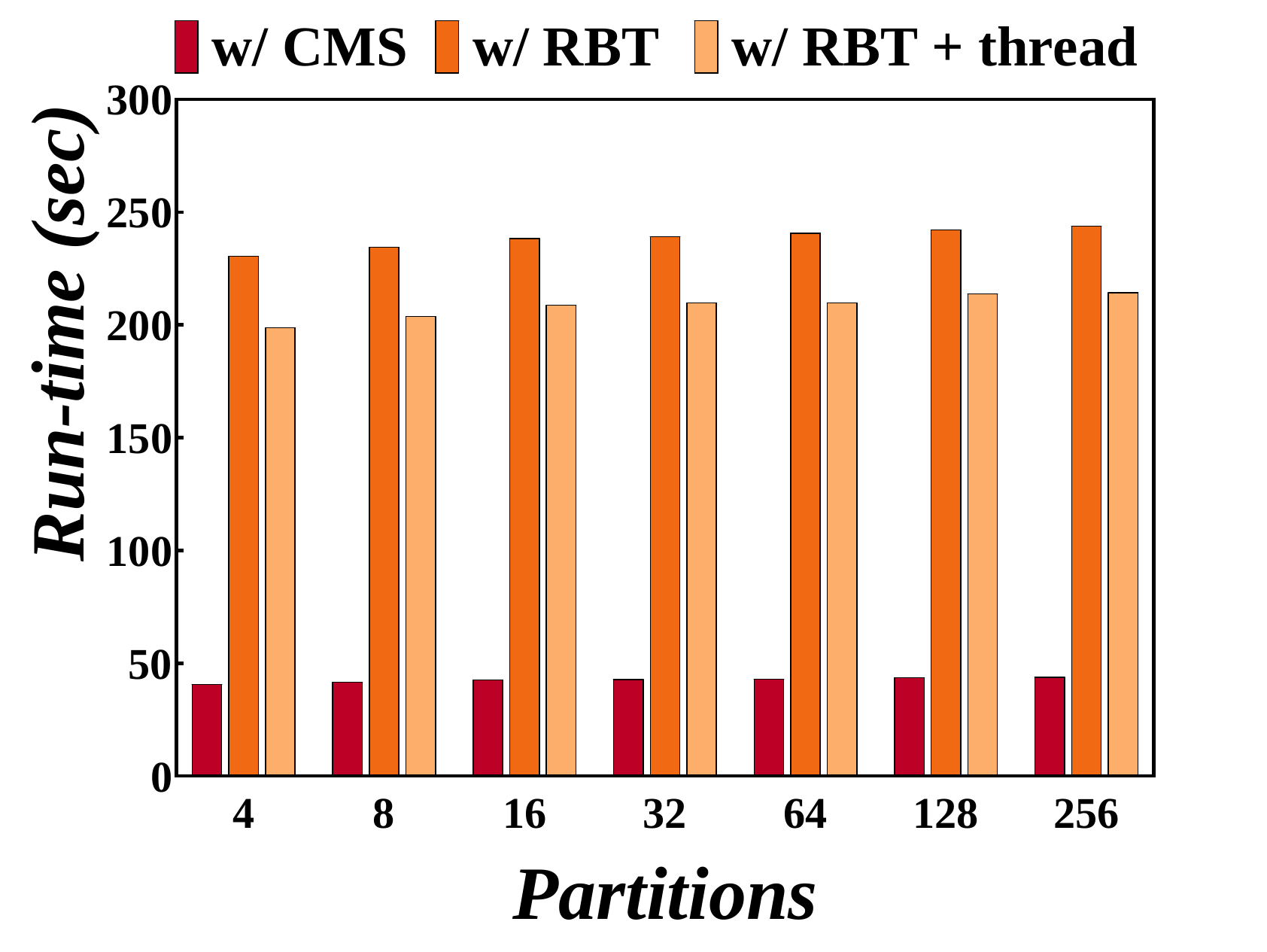}}
			\subfigure[Sketch (Mem.)]
			{\includegraphics[width= 0.24\columnwidth]{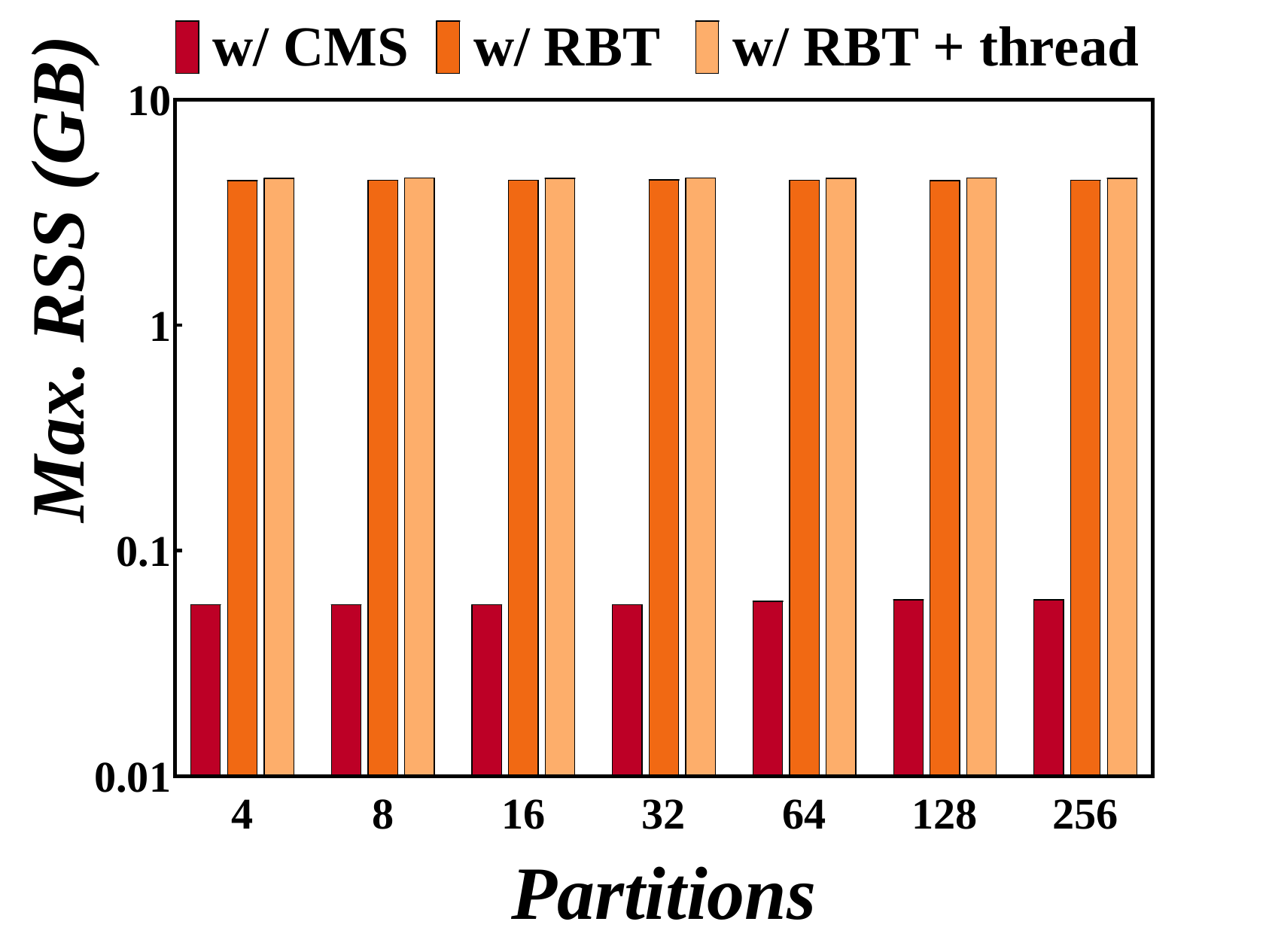}}
			\subfigure[Sketch (Rep. Factor)]
			{\includegraphics[width=0.24\columnwidth]{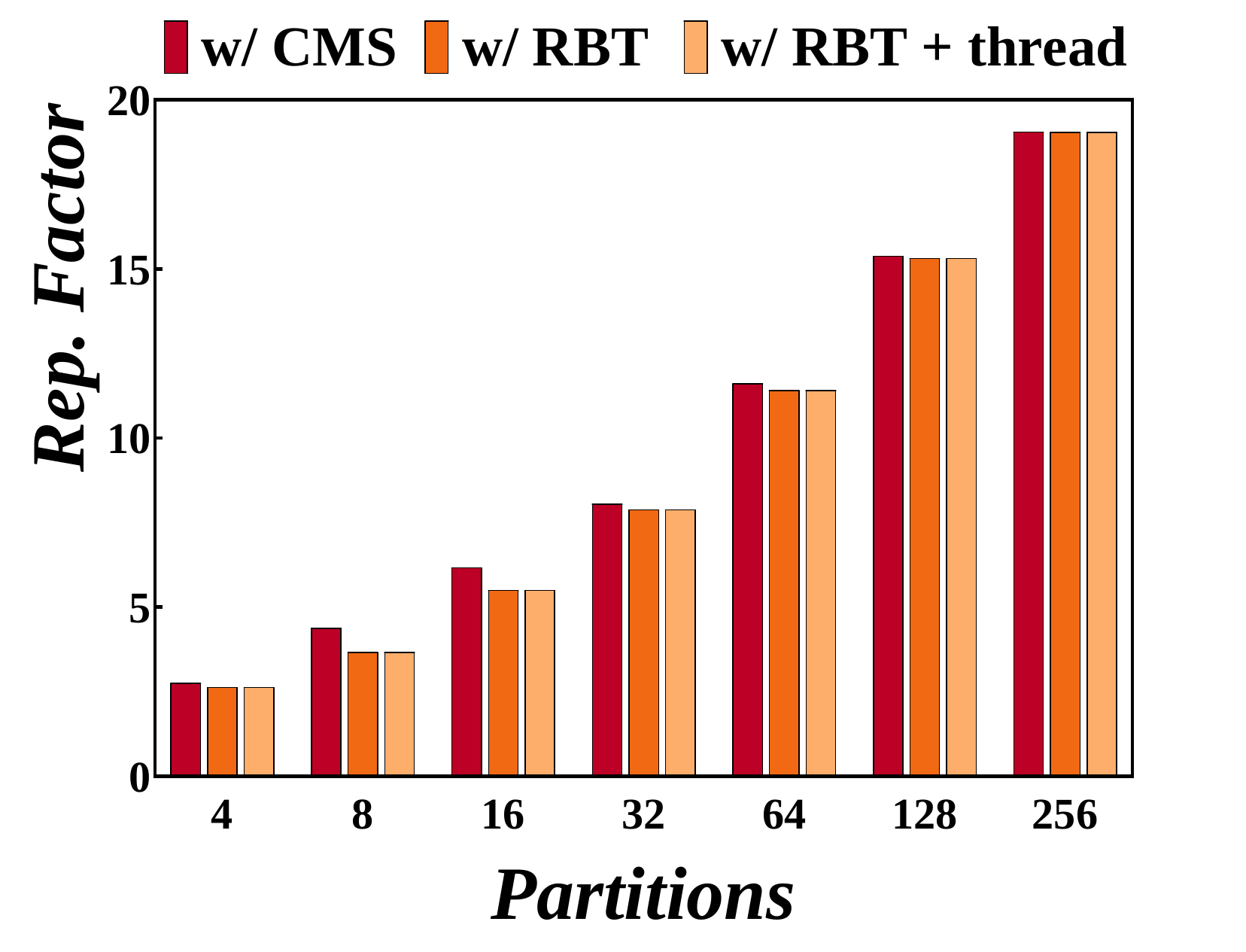}}
			\subfigure[ {Sketch ($\epsilon$ and $\nu$)}]
			{\includegraphics[width=0.26\columnwidth]{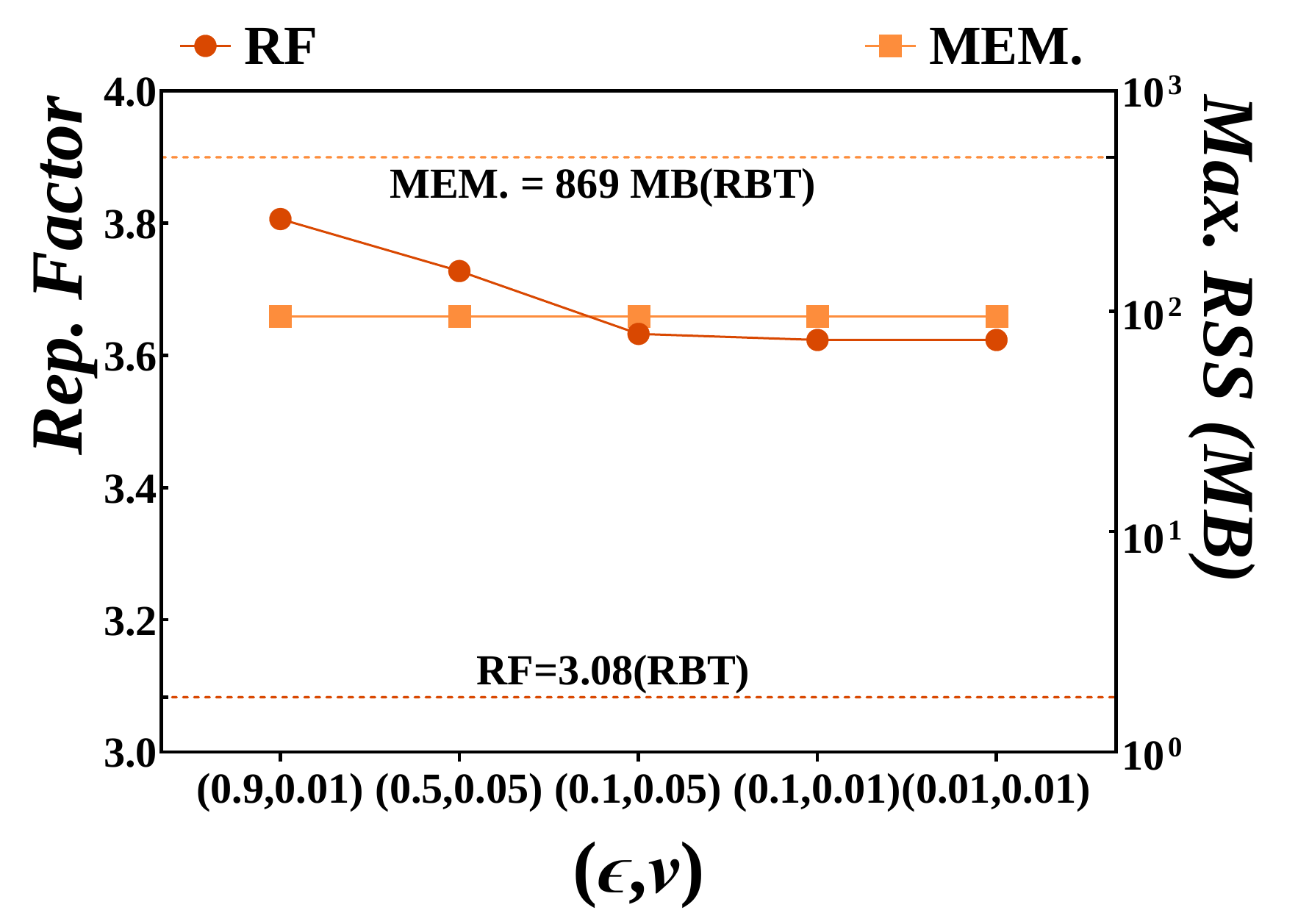}}		
\caption{ {Sketch-based Optimization} }
			\label{fig:expsketch}
	\end{minipage}

		\begin{minipage}{0.29\textwidth}
			{\includegraphics[width= 0.98\columnwidth]{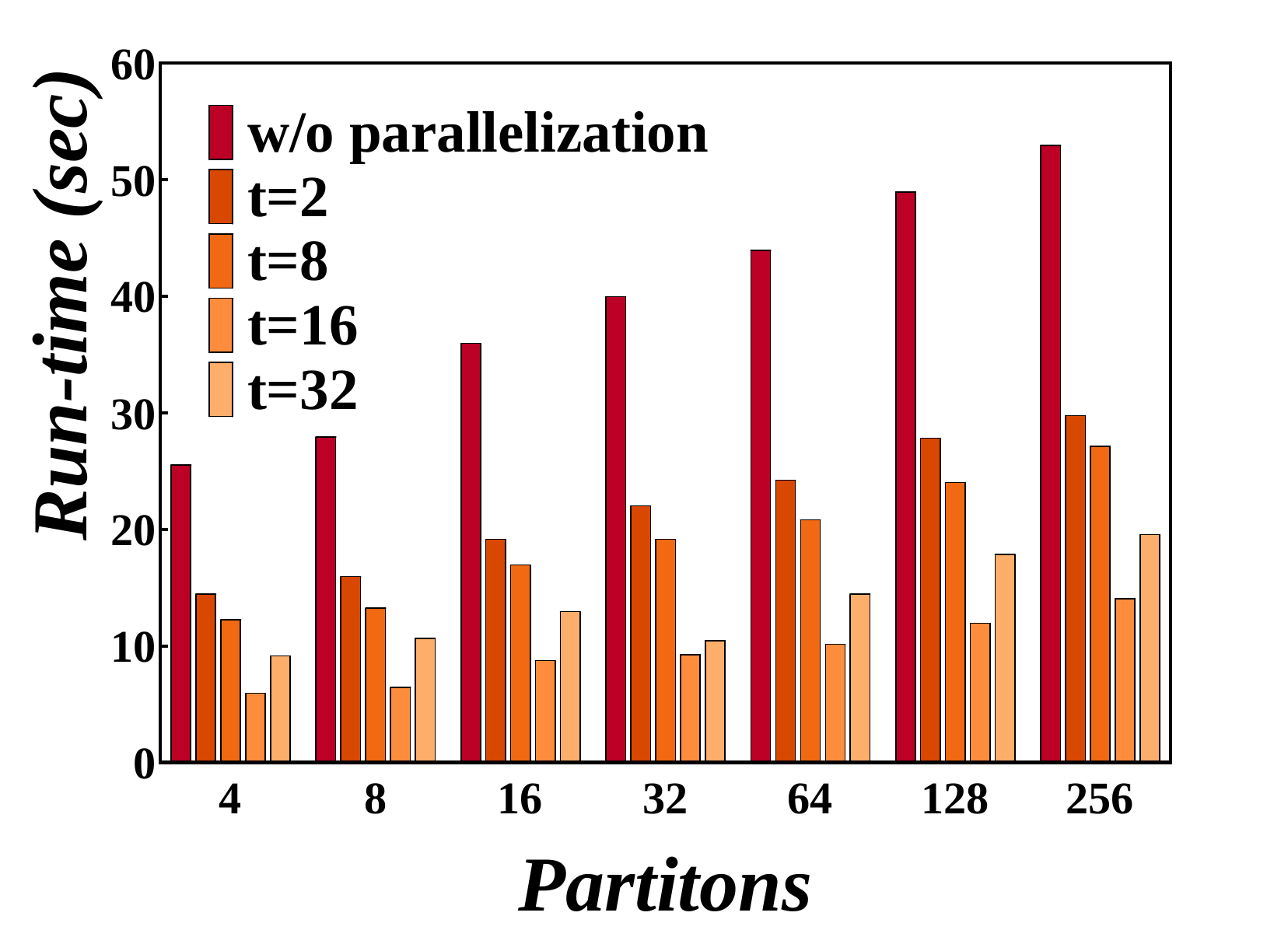}}
			\caption{Parallelization}
			\label{fig:expparalleli}
		\end{minipage}%
	\begin{minipage}{0.71\textwidth}
		\subfigure[Run-time]
		 {\includegraphics[width= 0.32\textwidth]{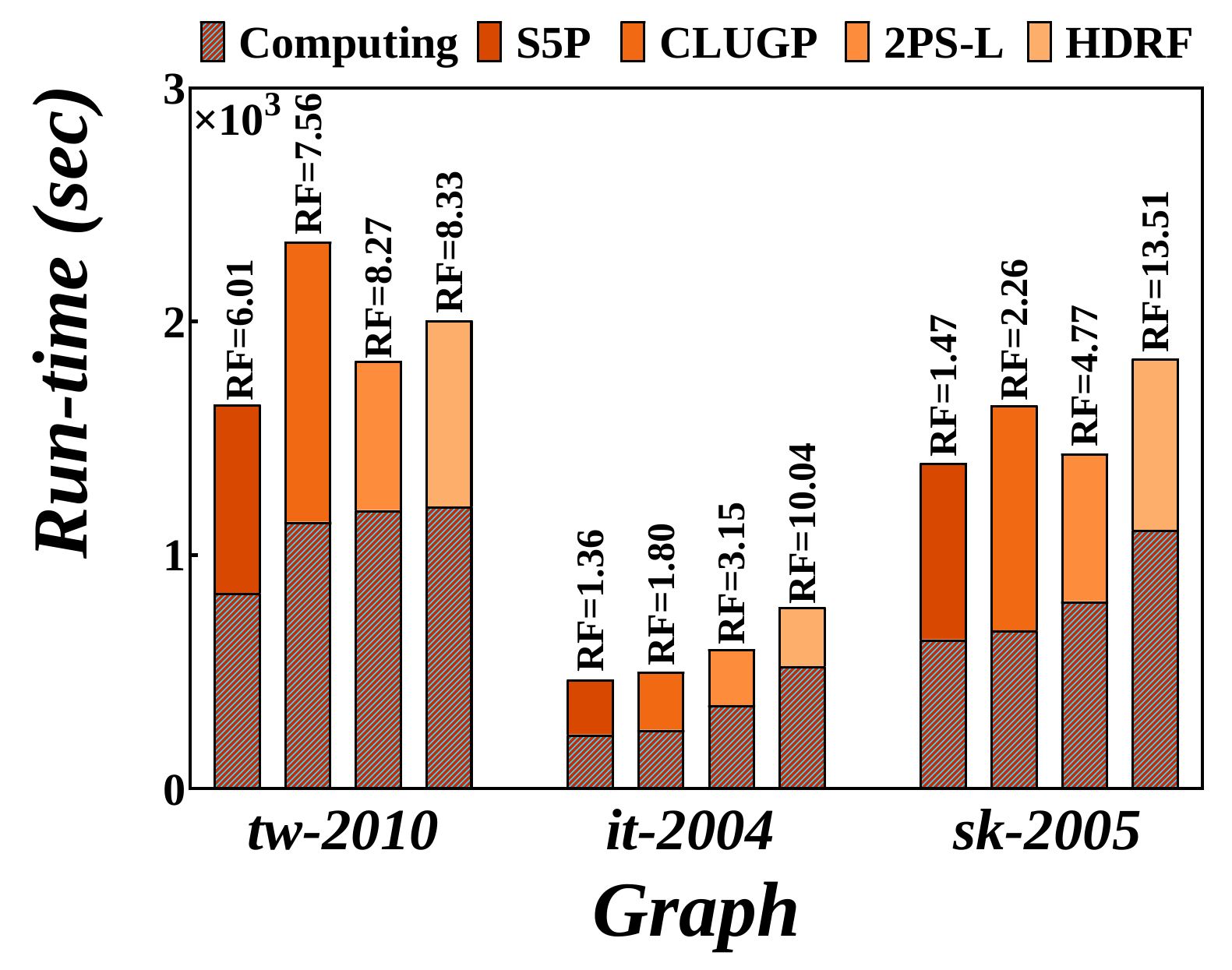}}
		\subfigure[Communication Cost]
		 {\includegraphics[width= 0.33\textwidth]{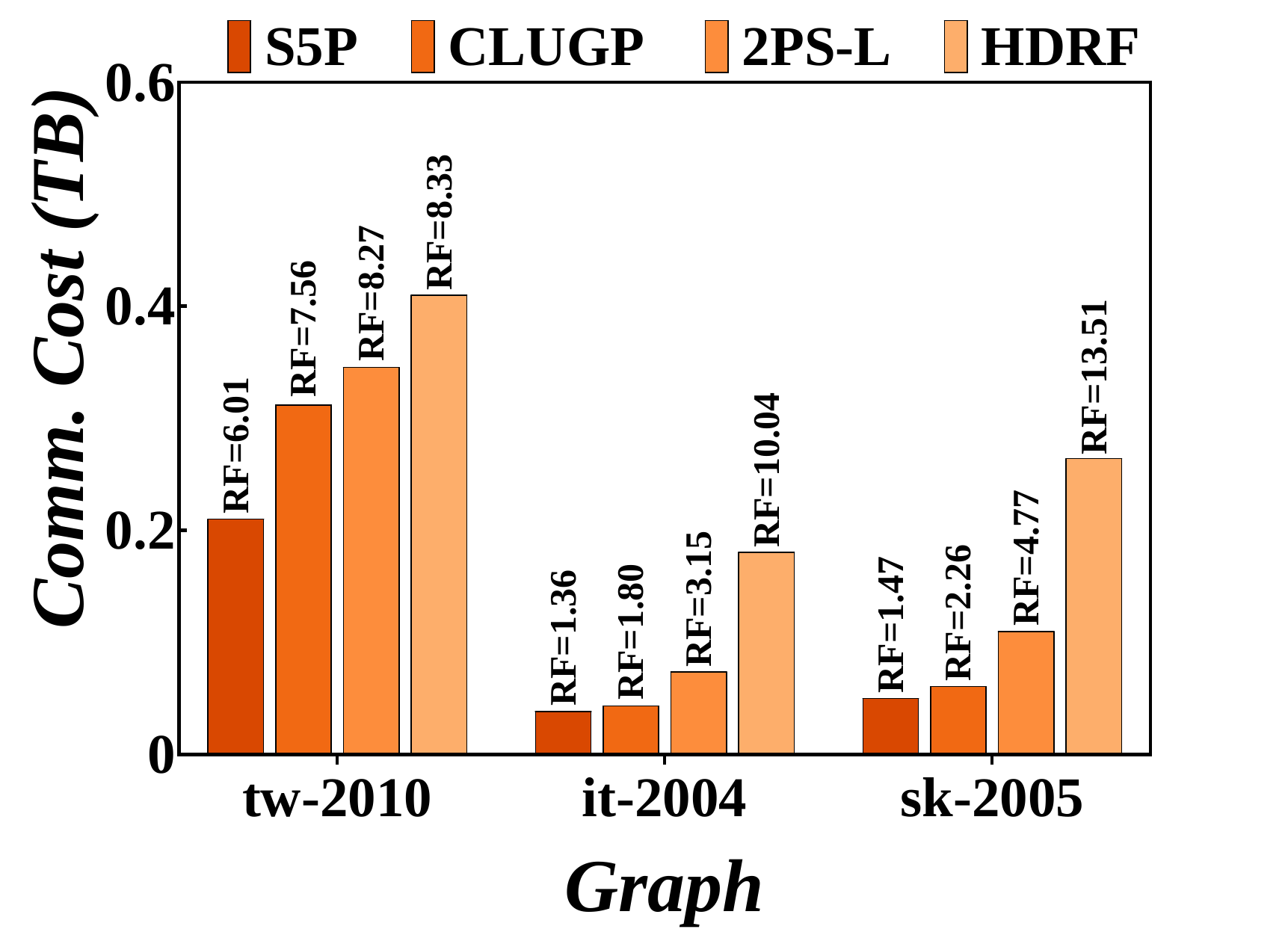}}
		\subfigure[Run-time w/ Latency (IT)]
		{\includegraphics[width= 0.33\textwidth]{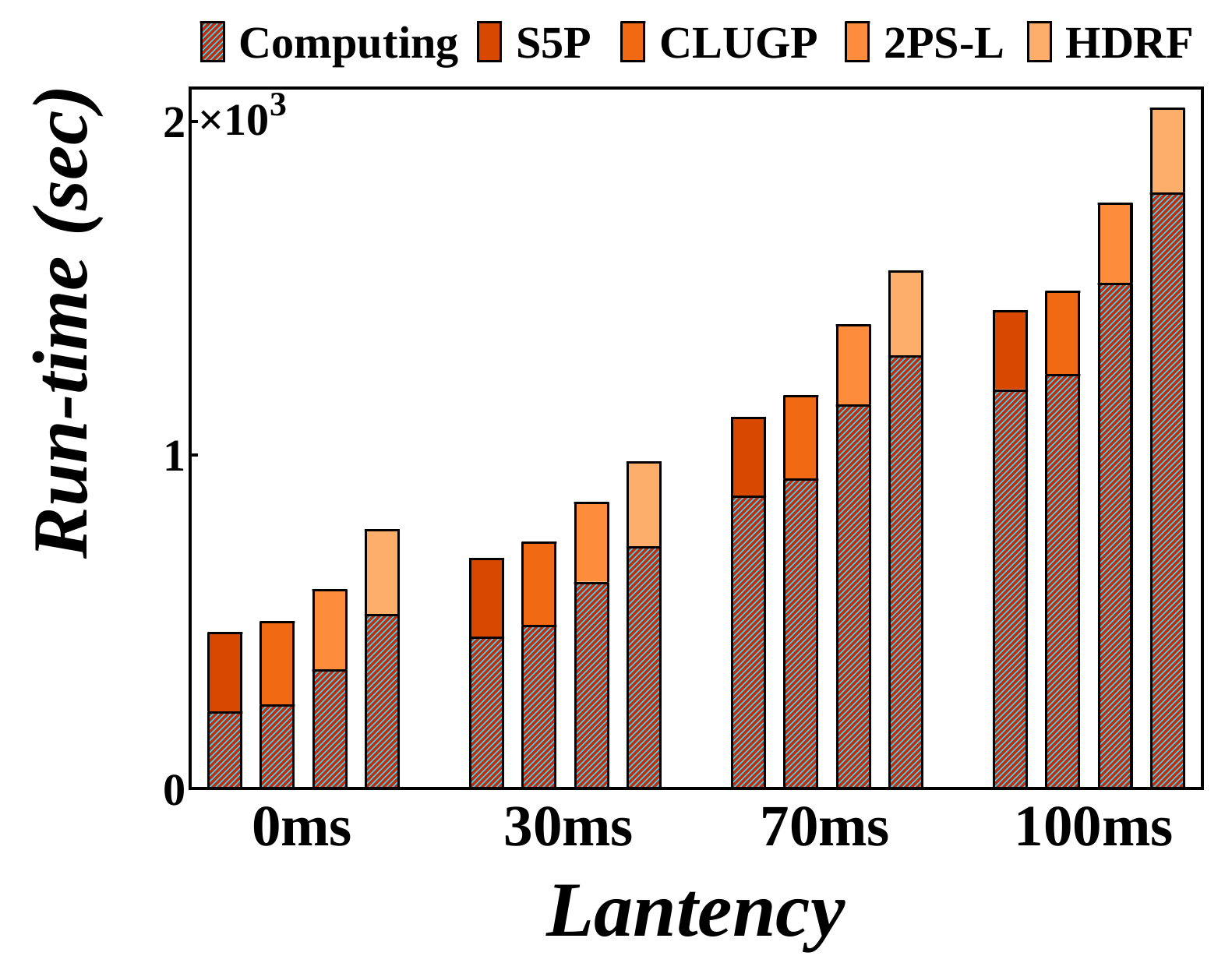}}
		\caption{ {Deployment on Distributed Graph System ($k$=32)} }
		\label{fig:powergraph}
	\end{minipage}

	\begin{minipage}{0.4\textwidth}
		\subfigure[Effect of $\beta$ (OK)] {\includegraphics[width= 0.44\columnwidth]{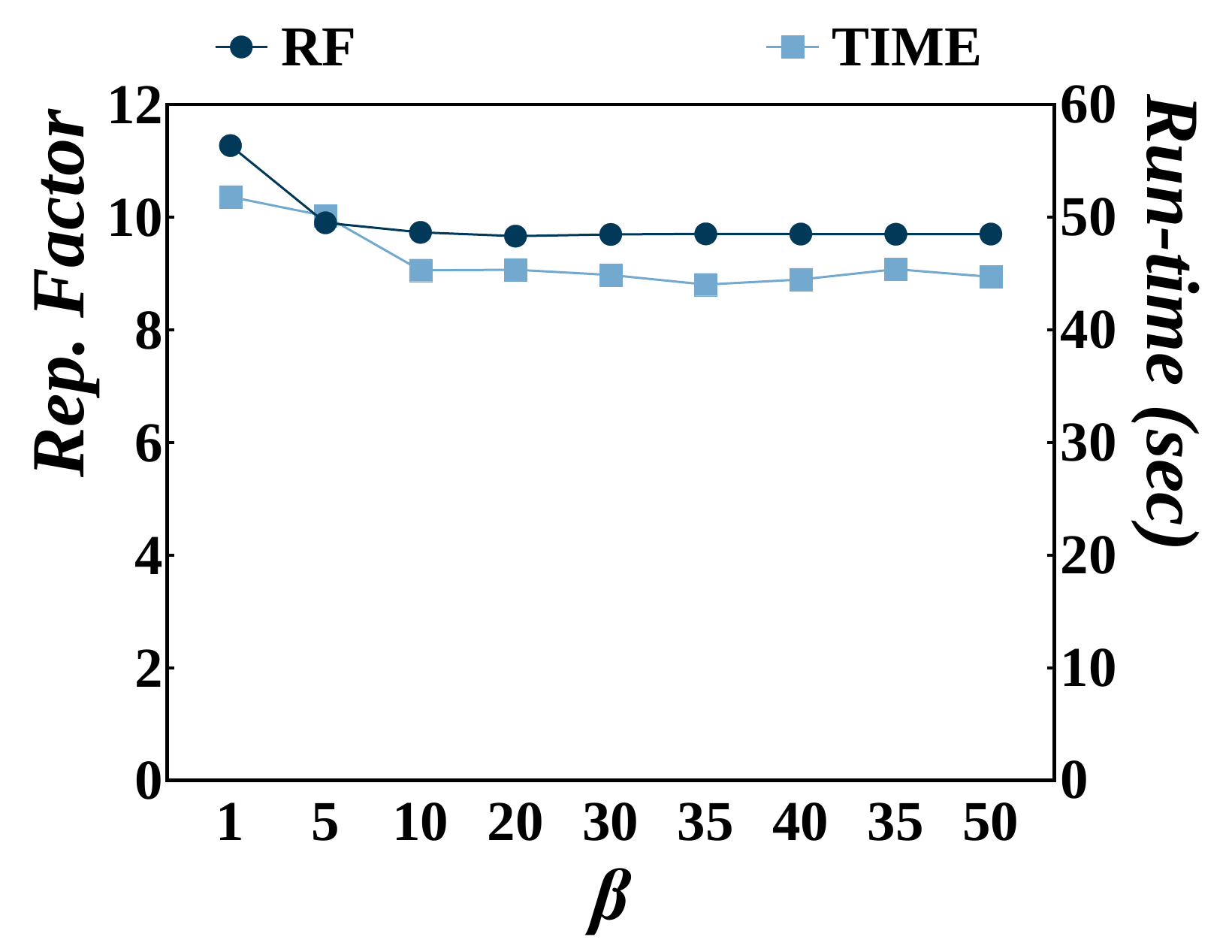}}
		\subfigure[Effect of batchsize (OK)]
		{\includegraphics[width= 0.45\columnwidth]{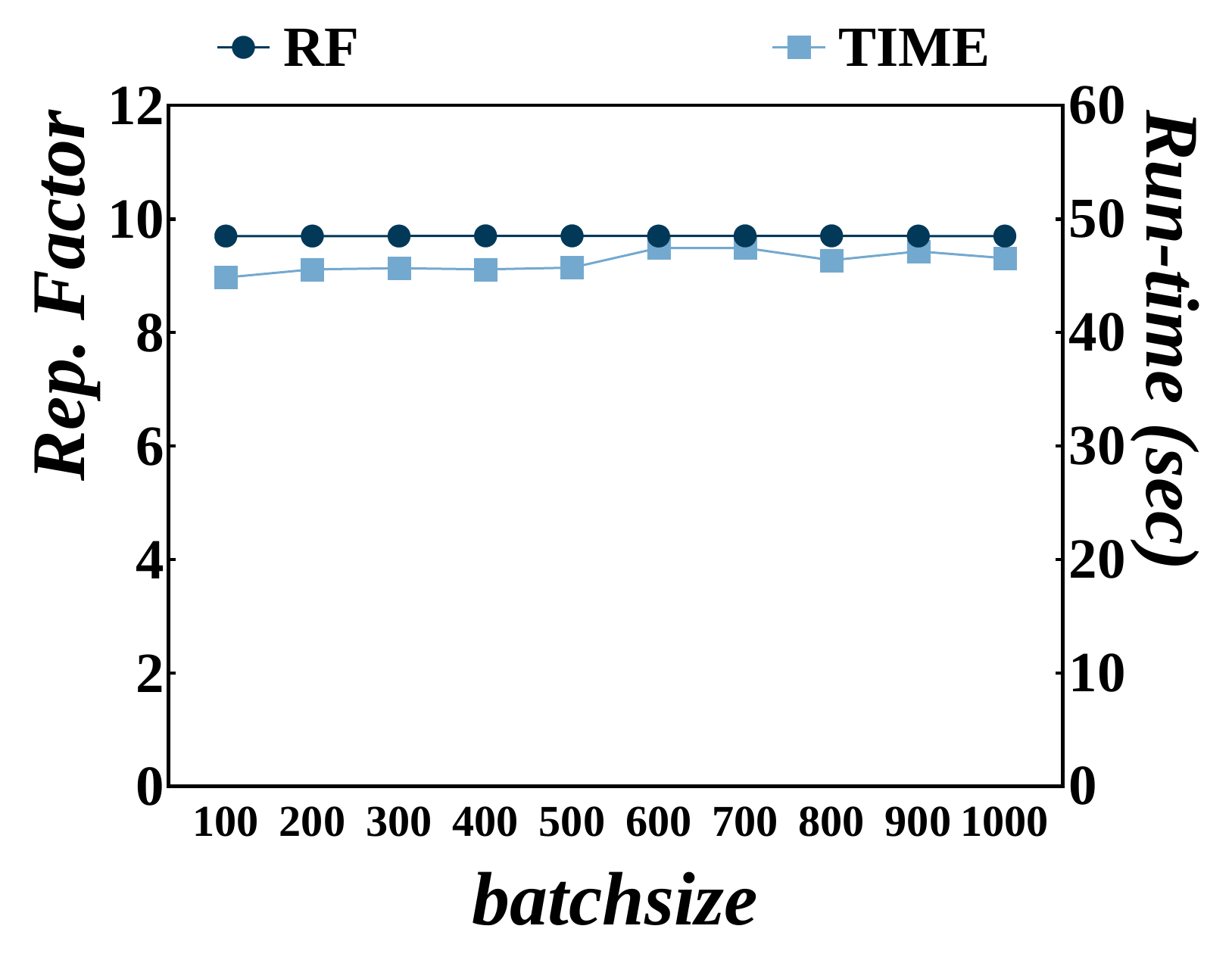}}
		\caption{Parameter Sensitivity}
		\label{fig:para}
	\end{minipage}%
	\begin{minipage}{0.6\textwidth}
		\subfigure[Replication Factor] {\includegraphics[width= 0.31\columnwidth]{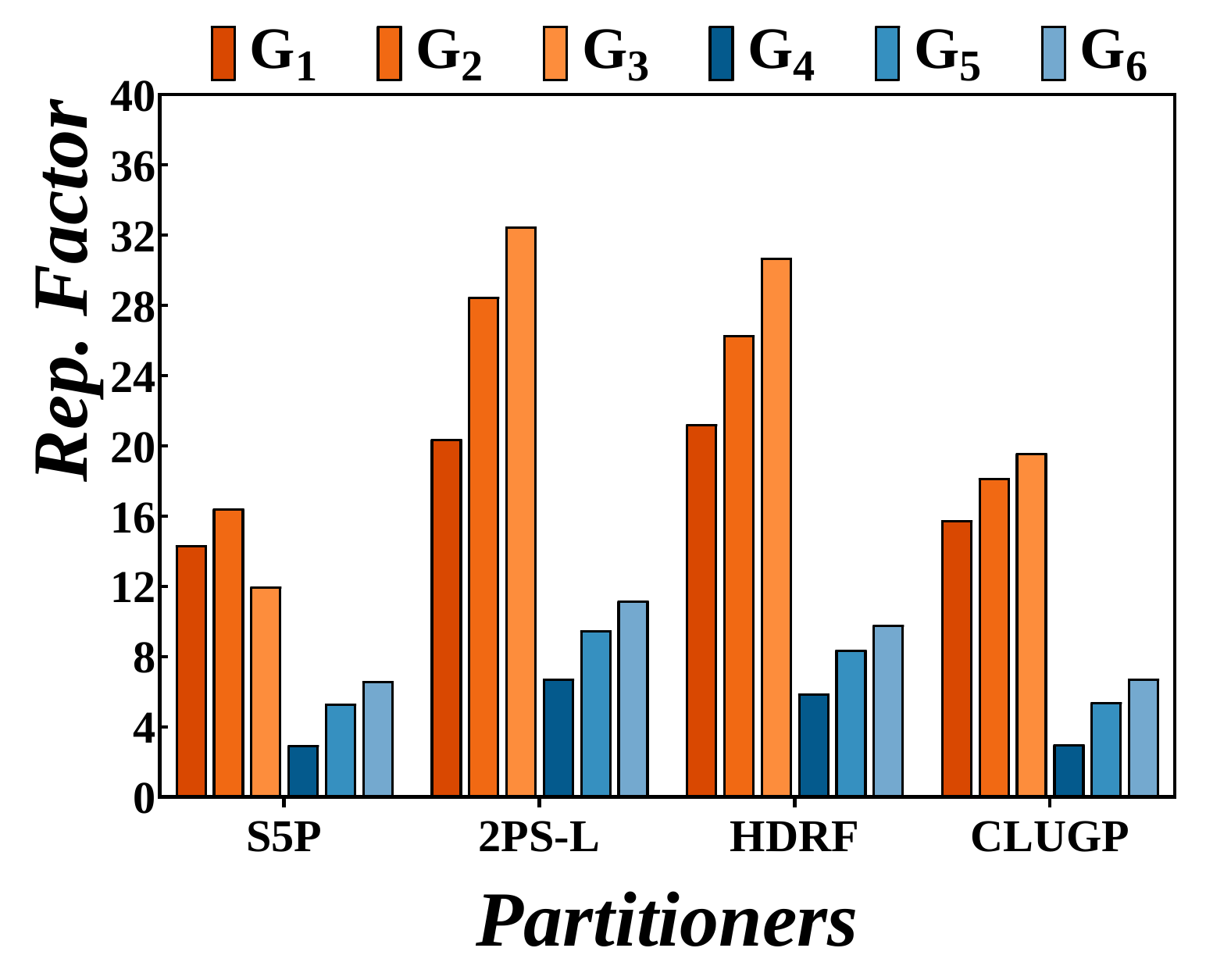}}
		\subfigure[Run-time] {\includegraphics[width= 0.32\columnwidth]{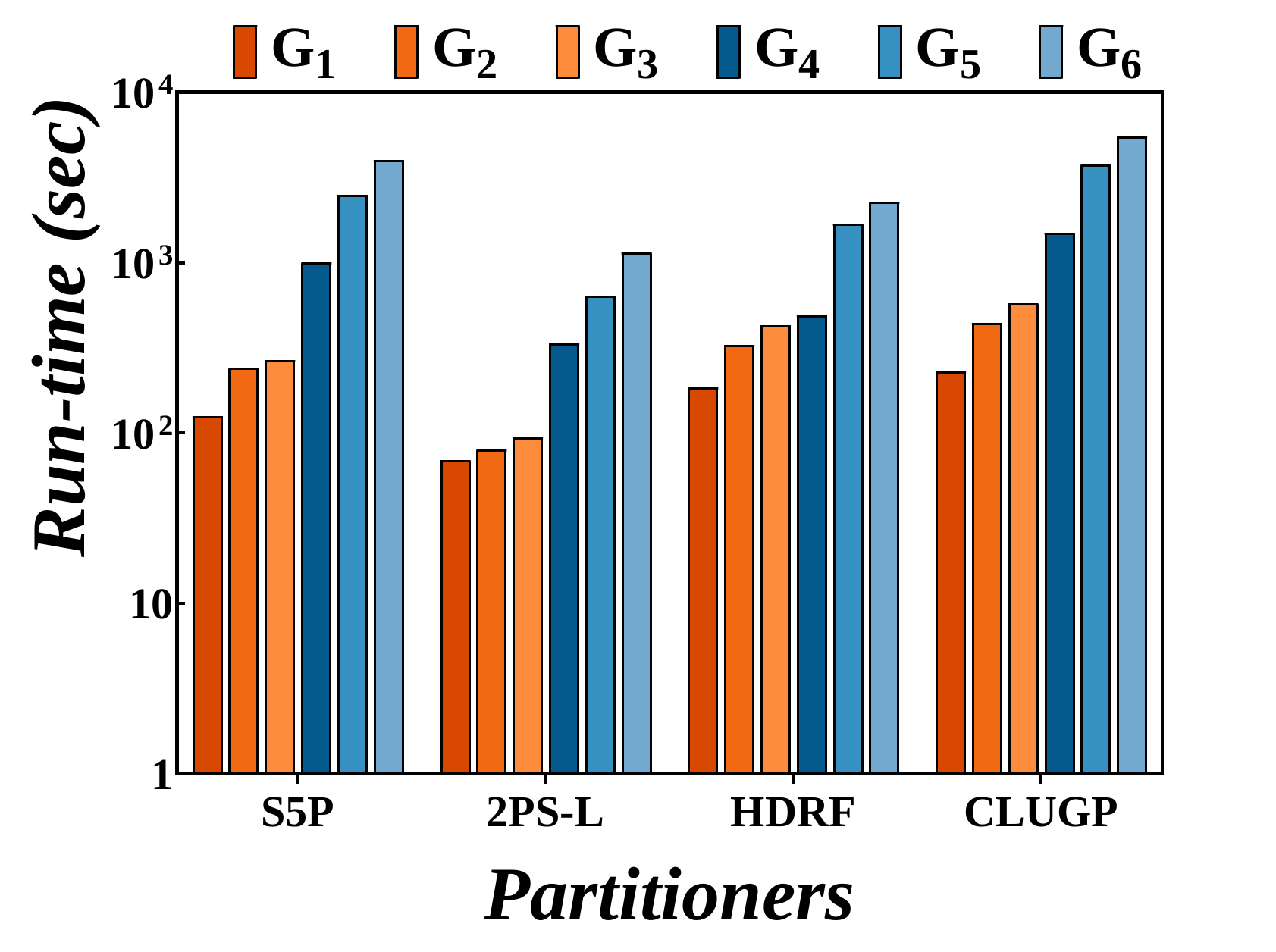}}
		\subfigure[Max. RSS] {\includegraphics[width= 0.32\columnwidth]{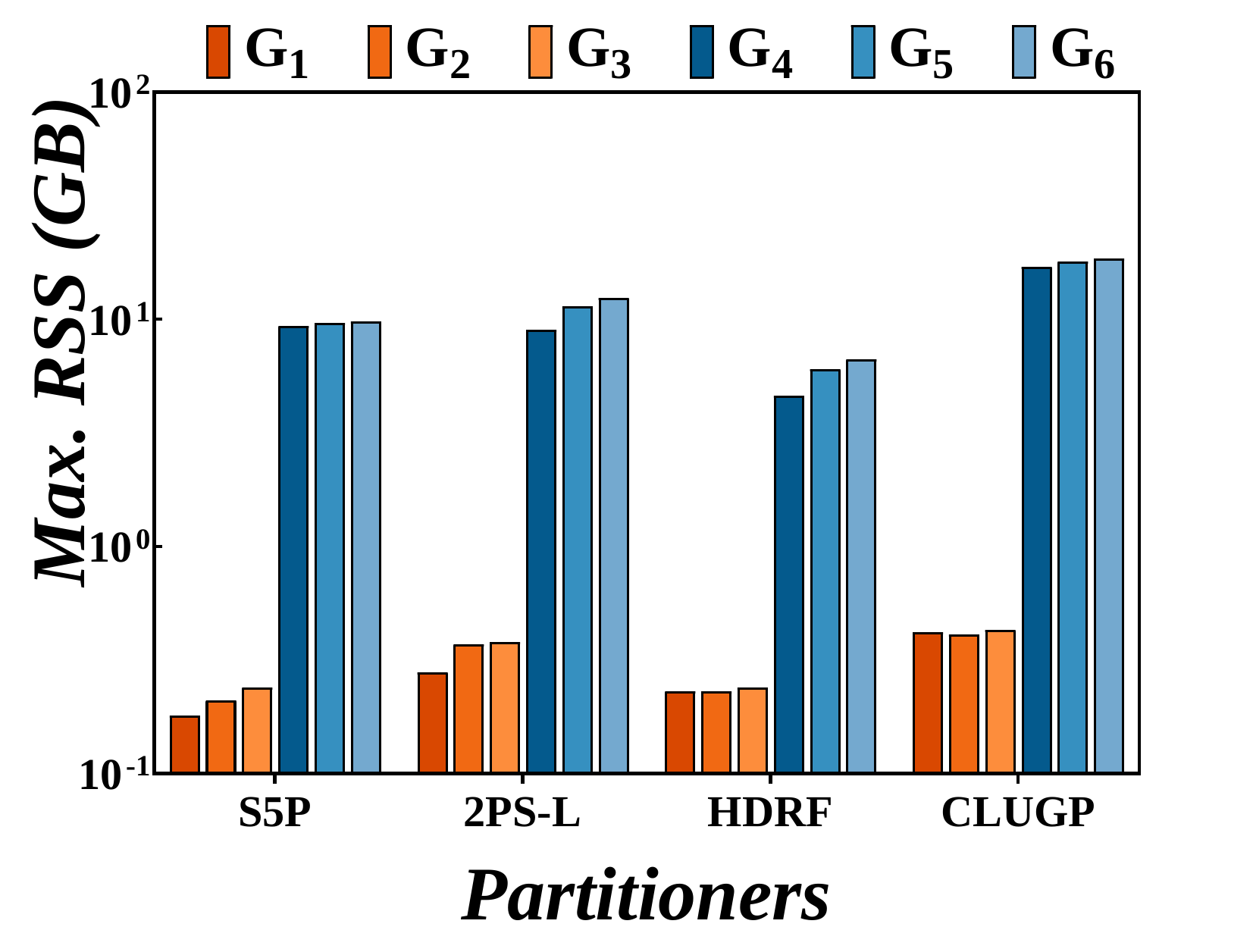}}
		\caption{ {Skewness Analysis ($k$=$64$)}}
		\label{fig:skewnessexp}
	\end{minipage}
\end{figure*}

From now on, we focus exclusively on assessing the RF of the top 4 streaming partitioning methods, including 2PS-L, CLUGP, HDRF, and S5P, which is the category to which our method belongs.
In Table~\ref{tab:rfresult}, it shows that the partitioning quality of S5P dominates those of other baselines, in all test cases.
For social graphs, S5P yields up to a 12\% improvement compared to the second-ranked partitioner, and up to 30\% improvement compared to the fourth-ranked method. For web graphs, S5P can achieve a maximum improvement of 51\% compared to the second-ranked partitioner and up to 91\% improvement compared to the fourth-ranked method.

{What is more, we calculate the approximation ration using small-scale datasets RMAT-$G_{\alpha/\beta/\gamma}$ generated by R-MAT. We obtained the optimality for each graph using the enumeration (i.e., the algorithmic time complexity is $k^{|E|}$). As shown in Table~\ref{ret:approximate}, We have verified S5P reduces the approximation $\alpha$ in comparison to HDRF, CLUGP, and 2PS-L under $k=4$.}

{
We compare game-based methods in Table~\ref{tab:gameresult}.
It shows that S5P has significant advantages in terms of RF, time, and memory.}
 {
Compared with others, S5P introduces a differentiation between head and tail clusters considering the skewness, where actions of head clusters influence the moves of tail clusters, aligning with sequential game theories. 
Compared to other game-based methods, S5P achieves up to a 6-fold improvement in RF, and orders of magnitude acceleration and compression.
}

In summary, S5P, as a streaming vertex-cut partitioner with skewness awareness, outperforms existing streaming edge partitioning methods. It can effectively reduce replication factors and run efficiently in low-memory environments. {The novelty of S5P is underscored by its innovative approach, while adapting clusters as players in a Stackelberg game poses unique challenges: (1) The clustering process is intricate, especially in streaming scenarios. Current streaming clustering algorithms lack skewness awareness, as highlighted in Table~\ref{tab:clusteringNEW}. Therefore, the partition quality they ultimately produce is not good (cf. Table~\ref{tab:rfresult}). (2) Transferring skewness information from clustering to game tasks is a non-trivial task. S5P adeptly utilizes skewness information from clustering for both leaders and followers. In contrast, other games often treat different cluster types unequally, limiting their ability to fully leverage this information. Experimental results further substantiate the significant advantages of our game (cf. Table~\ref{tab:gameresult}).} Hence, S5P represents an appealing new option for streaming graph partitioning when partitioning quality, execution time, and memory usage are crucial factors to consider.

\subsection{Component Analysis ($Q_2$)}
\label{sec:ca}
Hereby, we study the effectiveness of the two major components of S5P, skewness-aware clustering and gaming.

\subsubsection*{Skewness-aware Clustering}

{We investigate the effectiveness of our proposed skewness-aware clustering (\emph{S5P-Clustering}) in LJ}. As our approach employs vertex clusters to represent edge clusters, two questions need addressing: 1) what are the advantages of this representation? 2) whether it compromises the partitioning quality?
Then, we investigate the effect of S5P-Clustering by comparing it with \emph{Edge-Clustering}.
In Figure \ref{fig:ablation} (a), we compare the runtime cost of S5P-Clustering and Edge-Clustering. Compared to Edge-Clustering, using S5P-Clustering achieves a speedup of $8\times$.
In Figure \ref{fig:ablation} (b), we compare the memory consumption of S5P-Clustering and Edge-Clustering.
The memory consumption of S5P-Clustering is only 6\% of Edge-Clustering's.

\subsubsection*{Stackelberg Graph Game}

 {The experiment in Figure \ref{fig:ablation} (c) compares the result of Stackelberg gaming on cluster levels (w/ clustering) and ordinary gaming on edge levels (w/o clustering) in LJ}.
The results show that cluster-level Stackelberg gaming adeptly captures the skewness of distributions, resulting in a reduced replication factor. Particularly noteworthy is the increasing advantage observed as the number of partitions grows. In particular, with the presence of 256 partitions, the replication factor of w/ clustering is merely 40\% of that of w/o clustering, which underscores the effectiveness of Stackelberg game.
 {We further test the effectiveness of gaming on cluster levels in SK}.
Specifically, the experiment presented in Figure \ref{fig:ablation} (d) evaluates the efficacy of the one-stage game, where head and tail clusters are not differentiated, against the two-stage Stackelberg game introduced in Section~\ref{sec:method}. The results emphasize the substantial advantage of the two-stage Stackelberg game in terms of replication factor reduction. In the majority of tests, the two-stage Stackelberg game achieves a RF that is approximately 50\% lower than that of the one-stage game.

\subsection{Optimization ($Q_3$)}

\subsubsection*{Parallelization}
We investigated the impact of different amounts of parallelism (thread counts) on the runtime cost during the game process (cf. Figure \ref{fig:expparalleli}).
It shows that the runtime cost significantly reduces as the increase of the number of threads.
On the other hand, the acceleration becomes less significant, if the benefits of parallelism are counteracted by the cost of CPU switching and memory accessing when given a large number of threads.
Based on comprehensive experimental results, we set the thread count $t$ to 16 in our implementation.

\subsubsection*{Sketch}

We consider three strategies for updating and estimating the intersection size between different clusters: red-black trees (RBT), red-black trees with thread parallelism, and CMS. Experimental results revealed that using CMS not only achieves significant optimizations in terms of time and space but also demonstrates comparable performance with other strategies in RF outputs.  {As shown in Figure~\ref{fig:expsketch} (a), using CMS in the OK dataset achieves an acceleration of over $4$ times.
In Figures~\ref{fig:expsketch} (b) and (c), we compare the memory usage and RF of CMS and the vanilla data structure (i.e., RBT).
When $k = 256$, RBT-based implementation requires a total of $4.5$GB of memory to achieve an RF of $19.046$. In contrast, CMS-based implementation only takes $60$MB of memory yet reaches an RF of $19.055$. 
In other words, we trade less than $0.01$\% RF loss for more than a $70$-fold memory cost reduction. Figure~\ref{fig:expsketch} (d) examines the effect of sketch parameter settings, i.e., $\epsilon$ and $\nu$.
It shows that, as $\epsilon$ and $\nu$ decrease, the estimation error of CMS is also decreased, so that the corresponding RF shows a slight downward trend.
On the other hand, the memory usage is quite stable, because the space cost of CMS is dominated by other parts of Algorithm~\ref{alg:kagc} (i.e., $O(|V|)$).}

\subsection{{  Parameter Sensitivity ($Q_4$)} }
\label{sec:ps}
The variation of $\beta$ affects the distribution changes in the head and tail parts of the graph components. { As shown in Figure~\ref{fig:para} (a), we observed that as $\beta$ increases, the RF decreases rapidly at first, then gradually rises, showing an elbow-shaped trend.} The runtime decreases initially and then increases with the increase of $\beta$. This phenomenon occurs because when $\beta$ is too small or too large, either the head part or the tail part dominates the algorithm execution process, weakening the skewness-aware property of the algorithm.  {It can be observed from the graph that S5P is not sensitive to $\beta$. For simplicity, we recommend setting $\beta$ to 1.} Consequently, this leads to a decline in both partition quality and efficiency. Therefore, it is evident that skewness-aware algorithms not only enhance partition quality but also optimize efficiency.

{ With the increase in batch size, the runtime gradually increases and then stabilizes within a relatively small range, which is shown in Figure~\ref{fig:para} (b).} It is because even a larger batch size results in more time spent on each game iteration, the total number of game iterations decreases due to the reduction in the number of batches. When the batch size increases, gradually RF declines. This occurs because as the batch size increases, the algorithm incorporates more game information, leading to an overall downward trend.

\subsection{ {Deployment on Real Systems ($Q_5$)}  }
\label{sec:drs}
We evaluate the performance of partitioning algorithms on real distributed graph systems, specifically PowerGraph. Figure~\ref{fig:powergraph}~(a) and (b) present the graph processing time and communication cost for the PageRank algorithm. Across all tests, S5P consistently exhibits the shortest computation and communication times. Compared to the second fastest algorithm, S5P achieves a maximum speedup of 1.5 $\times$ and a 49\% reduction in communication cost. Compared to the fourth fastest algorithm, S5P achieves a maximum speedup of 2.3 $\times$ and an 81\% reduction in communication cost. This outstanding performance can be attributed to S5P's high partitioning quality, which includes load balancing and a low replication factor.
  {There exists a trade-off between RF and partitioning runtime. In addition, the RF value has a significant impact on the efficiency of downstream distributed graph tasks.
In Figure~\ref{fig:powergraph}~(a), we can find that S5P excels in the total processing time, consisting of time cost of graph partitioning and downstream distributed tasks.
{ 
Moreover, the cost incurred in graph partitioning serves as an upfront investment that, when reused across multiple distributed graph tasks, yields escalating benefits by consistently reducing communication cost and enhancing overall computational efficiency.}}
In general, hashing-based methods perform the worst, and the performance gap widens as data volumes increase.  {To assess latency, we utilize \emph{PUMBA}~\cite{PUMBA} to manipulate the \emph{Round-Trip Time} (\emph{RTT}) from 10ms to 100ms in IT. Figure~\ref{fig:powergraph}~(c) illustrates the running time of PageRank under various network latency conditions. As the latency increases, the proportion of time spent on graph computing becomes larger. Once again, S5P proves to be the most efficient and stable method in all scenarios.}

\subsection{ {Skewness Analysis ($Q_6$)} }
 {
We evaluate the performance of the top four streaming partitioners on a series of synthetic graphs ($G_1$ to $G_6$), representing different levels of skewness, generated by R-MAT. These graphs can be divided into two groups: $\{G_1, G_2, G_3\}$ and $\{G_4, G_5, G_6\}$, with bigger subscripts indicating more skewed graphs within each group.
In Figure~\ref{fig:skewnessexp} (a), it shows that partitioners other than S5P exhibit a substantial increase in RF as the graphs are more skewed. In contrast, S5P excels in coping with skewed graphs, showing the smallest RF increments as the graphs are more skewed. Remarkably, from $G_2$ to $G_3$, where the graph skewness changes from (0.87, 0.17, 0.48, 626M) to (0.84, 0.19, 0.52, 1B), the resulting RF of S5P even decreases from $16.460$ to $12.011$. Figures~\ref{fig:skewnessexp} (b) and (c) show the performance of S5P on runtime and memory overhead under different skewness settings, which is consistent with the results on real graphs.
}

\section{conclusion}
\label{sec:conclusion}

In this paper, we propose a new skewness-aware vertex-cut streaming partitioner, S5P, for large graphs. S5P follows a biphasic clustering-refinement framework, wherein we propose a skewness-aware clustering algorithm for retrieving find-grained head and tail clusters, and further propose a Stackelberg gaming approach to allocate the clusters into predefined partitions. Alongside this, we investigate a series of techniques, including sketching and parallelization, for optimizing the performance of the partitioner. Theoretically, we provide a detailed analysis of S5P. Experimentally, we conduct extensive experiments on real-world and synthetic graphs to showcase the effectiveness and efficiency of our approach. In the future, we plan to extend the skewness-aware partitioning paradigm to traditional graph computing systems and graph learning systems.
\begin{acks}
	This work is supported by NSFC (No.61772492, 62072428). Xike Xie is the corresponding author.
\end{acks}

\bibliographystyle{ACM-Reference-Format}
\bibliography{sigmod}


\begin{thebibliography}{61}


\ifx \showCODEN    \undefined \def \showCODEN     #1{\unskip}     \fi
\ifx \showDOI      \undefined \def \showDOI       #1{#1}\fi
\ifx \showISBNx    \undefined \def \showISBNx     #1{\unskip}     \fi
\ifx \showISBNxiii \undefined \def \showISBNxiii  #1{\unskip}     \fi
\ifx \showISSN     \undefined \def \showISSN      #1{\unskip}     \fi
\ifx \showLCCN     \undefined \def \showLCCN      #1{\unskip}     \fi
\ifx \shownote     \undefined \def \shownote      #1{#1}          \fi
\ifx \showarticletitle \undefined \def \showarticletitle #1{#1}   \fi
\ifx \showURL      \undefined \def \showURL       {\relax}        \fi
\providecommand\bibfield[2]{#2}
\providecommand\bibinfo[2]{#2}
\providecommand\natexlab[1]{#1}
\providecommand\showeprint[2][]{arXiv:#2}

\bibitem[Albert et~al\mbox{.}(2000)]%
        {albert2000error}
\bibfield{author}{\bibinfo{person}{R{\'e}ka Albert}, \bibinfo{person}{Hawoong
  Jeong}, {and} \bibinfo{person}{Albert-L{\'a}szl{\'o} Barab{\'a}si}.}
  \bibinfo{year}{2000}\natexlab{}.
\newblock \showarticletitle{Error and attack tolerance of complex networks}.
\newblock \bibinfo{journal}{\emph{nature}} \bibinfo{volume}{406},
  \bibinfo{number}{6794} (\bibinfo{year}{2000}), \bibinfo{pages}{378--382}.
\newblock
\urldef\tempurl%
\url{https://doi.org/10.1038/35019019}
\showDOI{\tempurl}


\bibitem[Armenatzoglou et~al\mbox{.}(2015)]%
        {rmgp}
\bibfield{author}{\bibinfo{person}{Nikos Armenatzoglou}, \bibinfo{person}{Huy
  Pham}, \bibinfo{person}{Vasilis Ntranos}, \bibinfo{person}{Dimitris
  Papadias}, {and} \bibinfo{person}{Cyrus Shahabi}.}
  \bibinfo{year}{2015}\natexlab{}.
\newblock \showarticletitle{Real-Time Multi-Criteria Social Graph Partitioning:
  A Game Theoretic Approach}. In \bibinfo{booktitle}{\emph{SIGMOD}}.
  \bibinfo{publisher}{ACM}, \bibinfo{address}{New York, NY, USA},
  \bibinfo{pages}{1617–1628}.
\newblock
\showISBNx{9781450327589}
\urldef\tempurl%
\url{https://doi.org/10.1145/2723372.2749450}
\showDOI{\tempurl}


\bibitem[Boldi et~al\mbox{.}(2018)]%
        {craw}
\bibfield{author}{\bibinfo{person}{Paolo Boldi}, \bibinfo{person}{Andrea
  Marino}, \bibinfo{person}{Massimo Santini}, {and} \bibinfo{person}{Sebastiano
  Vigna}.} \bibinfo{year}{2018}\natexlab{}.
\newblock \showarticletitle{BUbiNG: Massive crawling for the masses}.
\newblock \bibinfo{journal}{\emph{TWEB}} (\bibinfo{year}{2018}),
  \bibinfo{pages}{1--26}.
\newblock
\urldef\tempurl%
\url{https://doi.org/10.1145/3160017}
\showDOI{\tempurl}


\bibitem[Boldi et~al\mbox{.}(2011)]%
        {llp}
\bibfield{author}{\bibinfo{person}{Paolo Boldi}, \bibinfo{person}{Marco Rosa},
  \bibinfo{person}{Massimo Santini}, {and} \bibinfo{person}{Sebastiano Vigna}.}
  \bibinfo{year}{2011}\natexlab{}.
\newblock \showarticletitle{Layered label propagation: A multiresolution
  coordinate-free ordering for compressing social networks}. In
  \bibinfo{booktitle}{\emph{WWW}}. \bibinfo{pages}{587--596}.
\newblock
\urldef\tempurl%
\url{https://doi.org/10.1145/1963405.1963488}
\showDOI{\tempurl}


\bibitem[Boldi and Vigna(2004)]%
        {wg}
\bibfield{author}{\bibinfo{person}{Paolo Boldi} {and}
  \bibinfo{person}{Sebastiano Vigna}.} \bibinfo{year}{2004}\natexlab{}.
\newblock \showarticletitle{The webgraph framework I: compression techniques}.
  In \bibinfo{booktitle}{\emph{WWW}}. \bibinfo{pages}{595--602}.
\newblock
\urldef\tempurl%
\url{https://doi.org/10.1145/988672.988752}
\showDOI{\tempurl}


\bibitem[Bourse et~al\mbox{.}(2014)]%
        {bgep}
\bibfield{author}{\bibinfo{person}{Florian Bourse}, \bibinfo{person}{Marc
  Lelarge}, {and} \bibinfo{person}{Milan Vojnovic}.}
  \bibinfo{year}{2014}\natexlab{}.
\newblock \showarticletitle{Balanced Graph Edge Partition}. In
  \bibinfo{booktitle}{\emph{SIGKDD}} (New York, New York, USA).
  \bibinfo{publisher}{Association for Computing Machinery},
  \bibinfo{address}{New York, NY, USA}, \bibinfo{pages}{1456–1465}.
\newblock
\showISBNx{9781450329569}
\urldef\tempurl%
\url{https://doi.org/10.1145/2623330.2623660}
\showDOI{\tempurl}


\bibitem[Chakrabarti et~al\mbox{.}(2004)]%
        {chakrabarti2004r}
\bibfield{author}{\bibinfo{person}{Deepayan Chakrabarti},
  \bibinfo{person}{Yiping Zhan}, {and} \bibinfo{person}{Christos Faloutsos}.}
  \bibinfo{year}{2004}\natexlab{}.
\newblock \showarticletitle{R-MAT: A recursive model for graph mining}. In
  \bibinfo{booktitle}{\emph{Proceedings of the 2004 SIAM International
  Conference on Data Mining}}. SIAM, \bibinfo{pages}{442--446}.
\newblock
\urldef\tempurl%
\url{https://epubs.siam.org/doi/abs/10.1137/1.9781611972740.43}
\showURL{%
\tempurl}


\bibitem[Chen et~al\mbox{.}(2015)]%
        {powerlyra}
\bibfield{author}{\bibinfo{person}{Rong Chen}, \bibinfo{person}{Jiaxin Shi},
  \bibinfo{person}{Yanzhe Chen}, {and} \bibinfo{person}{Haibo Chen}.}
  \bibinfo{year}{2015}\natexlab{}.
\newblock \showarticletitle{PowerLyra: differentiated graph computation and
  partitioning on skewed graphs}. In \bibinfo{booktitle}{\emph{EuroSys}}.
  \bibinfo{publisher}{ACM}, \bibinfo{pages}{1:1--1:15}.
\newblock
\urldef\tempurl%
\url{https://doi.org/10.1145/3298989}
\showDOI{\tempurl}


\bibitem[Choromanski et~al\mbox{.}(2013)]%
        {choromanski2013scale}
\bibfield{author}{\bibinfo{person}{Krzysztof Choromanski},
  \bibinfo{person}{Michal Matuszak}, {and} \bibinfo{person}{Jacek Miekisz}.}
  \bibinfo{year}{2013}\natexlab{}.
\newblock \showarticletitle{Scale-free graph with preferential attachment and
  evolving internal vertex structure}.
\newblock \bibinfo{journal}{\emph{Journal of Statistical Physics}}
  \bibinfo{volume}{151} (\bibinfo{year}{2013}), \bibinfo{pages}{1175--1183}.
\newblock
\urldef\tempurl%
\url{https://doi.org/10.1007/s10955-013-0749-1}
\showDOI{\tempurl}


\bibitem[Cimikowski(1992)]%
        {cimikowski1992graph}
\bibfield{author}{\bibinfo{person}{Robert~J Cimikowski}.}
  \bibinfo{year}{1992}\natexlab{}.
\newblock \showarticletitle{Graph planarization and skewness}.
\newblock \bibinfo{journal}{\emph{Congressus Numerantium}}
  (\bibinfo{year}{1992}), \bibinfo{pages}{21--21}.
\newblock
\urldef\tempurl%
\url{https://citeseerx.ist.psu.edu/document?repid=rep1&type=pdf&doi=98060fa1da8eb732b8095e6f731ae387671d9ebb}
\showURL{%
\tempurl}


\bibitem[Cohen et~al\mbox{.}(2001)]%
        {cohen}
\bibfield{author}{\bibinfo{person}{Reuven Cohen}, \bibinfo{person}{Keren Erez},
  \bibinfo{person}{Daniel Ben-Avraham}, {and} \bibinfo{person}{Shlomo Havlin}.}
  \bibinfo{year}{2001}\natexlab{}.
\newblock \showarticletitle{Breakdown of the internet under intentional
  attack}.
\newblock \bibinfo{journal}{\emph{Physical review letters}}
  \bibinfo{volume}{86}, \bibinfo{number}{16} (\bibinfo{year}{2001}),
  \bibinfo{pages}{3682}.
\newblock
\urldef\tempurl%
\url{https://doi.org/10.1103/PhysRevLett.86.3682}
\showDOI{\tempurl}


\bibitem[Cormode and Muthukrishnan(2005)]%
        {cmsketch}
\bibfield{author}{\bibinfo{person}{Graham Cormode} {and} \bibinfo{person}{S.
  Muthukrishnan}.} \bibinfo{year}{2005}\natexlab{}.
\newblock \showarticletitle{An Improved Data Stream Summary: The Count-Min
  Sketch and Its Applications}.
\newblock \bibinfo{journal}{\emph{J. Algorithms}} \bibinfo{volume}{55},
  \bibinfo{number}{1} (\bibinfo{year}{2005}), \bibinfo{pages}{58–75}.
\newblock
\showISSN{0196-6774}
\urldef\tempurl%
\url{https://doi.org/10.1016/j.jalgor.2003.12.001}
\showDOI{\tempurl}


\bibitem[Doane and Seward(2011)]%
        {pearson}
\bibfield{author}{\bibinfo{person}{David~P. Doane} {and}
  \bibinfo{person}{Lori~E. Seward}.} \bibinfo{year}{2011}\natexlab{}.
\newblock \showarticletitle{Measuring Skewness: A Forgotten Statistic?}
\newblock \bibinfo{journal}{\emph{Journal of Statistics Education}}
  \bibinfo{volume}{19}, \bibinfo{number}{2} (\bibinfo{year}{2011}).
\newblock
\urldef\tempurl%
\url{https://doi.org/10.1080/10691898.2011.11889611}
\showDOI{\tempurl}


\bibitem[Durrett(2007)]%
        {durrett2007random}
\bibfield{author}{\bibinfo{person}{Richard Durrett}.}
  \bibinfo{year}{2007}\natexlab{}.
\newblock \bibinfo{booktitle}{\emph{Random graph dynamics}}.
  Vol.~\bibinfo{volume}{200}.
\newblock \bibinfo{publisher}{Cambridge university press Cambridge}.
\newblock
\urldef\tempurl%
\url{https://services.math.duke.edu/~rtd/RGD/RGD.pdf}
\showURL{%
\tempurl}


\bibitem[Etro(2008)]%
        {stackelberggameeq}
\bibfield{author}{\bibinfo{person}{Federico Etro}.}
  \bibinfo{year}{2008}\natexlab{}.
\newblock \showarticletitle{{Stackelberg Competition with Endogenous Entry}}.
\newblock \bibinfo{journal}{\emph{The Economic Journal}} \bibinfo{volume}{118},
  \bibinfo{number}{532} (\bibinfo{date}{09} \bibinfo{year}{2008}),
  \bibinfo{pages}{1670--1697}.
\newblock
\showISSN{0013-0133}
\urldef\tempurl%
\url{https://doi.org/10.1111/j.1468-0297.2008.02185.x}
\showDOI{\tempurl}


\bibitem[Fabio et~al\mbox{.}(2021)]%
        {c+cv}
\bibfield{author}{\bibinfo{person}{Furini Fabio}, \bibinfo{person}{Ljubi{\'c}
  Ivana}, \bibinfo{person}{Malaguti Enrico}, {and} \bibinfo{person}{Paronuzzi
  Paolo}.} \bibinfo{year}{2021}\natexlab{}.
\newblock \showarticletitle{Casting Light on the Hidden Bilevel Combinatorial
  Structure of the Capacitated Vertex Separator Problem}.
\newblock \bibinfo{journal}{\emph{Operations Research}} (\bibinfo{year}{2021}).
\newblock
\urldef\tempurl%
\url{https://api.semanticscholar.org/CorpusID:219719610}
\showURL{%
\tempurl}


\bibitem[Fan et~al\mbox{.}(2023)]%
        {FanXYYZ23}
\bibfield{author}{\bibinfo{person}{Wenfei Fan}, \bibinfo{person}{Ruiqi Xu},
  \bibinfo{person}{Qiang Yin}, \bibinfo{person}{Wenyuan Yu}, {and}
  \bibinfo{person}{Jingren Zhou}.} \bibinfo{year}{2023}\natexlab{}.
\newblock \showarticletitle{Application-driven graph partitioning}.
\newblock \bibinfo{journal}{\emph{VLDB J.}} \bibinfo{volume}{32},
  \bibinfo{number}{1} (\bibinfo{year}{2023}), \bibinfo{pages}{149--172}.
\newblock
\urldef\tempurl%
\url{https://doi.org/10.1007/s00778-022-00736-2}
\showDOI{\tempurl}


\bibitem[Feige et~al\mbox{.}(2008a)]%
        {IAA}
\bibfield{author}{\bibinfo{person}{Uriel Feige}, \bibinfo{person}{MohammadTaghi
  Hajiaghayi}, {and} \bibinfo{person}{James~R. Lee}.}
  \bibinfo{year}{2008}\natexlab{a}.
\newblock \showarticletitle{Improved Approximation Algorithms for Minimum
  Weight Vertex Separators}.
\newblock \bibinfo{journal}{\emph{SIAM J. Comput.}} \bibinfo{volume}{38},
  \bibinfo{number}{2} (\bibinfo{year}{2008}), \bibinfo{pages}{629--657}.
\newblock
\urldef\tempurl%
\url{https://doi.org/10.1137/05064299X}
\showDOI{\tempurl}


\bibitem[Feige et~al\mbox{.}(2008b)]%
        {FeigeHL08}
\bibfield{author}{\bibinfo{person}{Uriel Feige}, \bibinfo{person}{MohammadTaghi
  Hajiaghayi}, {and} \bibinfo{person}{James~R. Lee}.}
  \bibinfo{year}{2008}\natexlab{b}.
\newblock \showarticletitle{Improved Approximation Algorithms for Minimum
  Weight Vertex Separators}.
\newblock \bibinfo{journal}{\emph{SIAM J. Comput.}} \bibinfo{volume}{38},
  \bibinfo{number}{2} (\bibinfo{year}{2008}), \bibinfo{pages}{629--657}.
\newblock
\urldef\tempurl%
\url{https://doi.org/10.1137/05064299X}
\showDOI{\tempurl}


\bibitem[Gairing et~al\mbox{.}(2005)]%
        {nash_2005}
\bibfield{author}{\bibinfo{person}{Martin Gairing}, \bibinfo{person}{Thomas
  Lücking}, \bibinfo{person}{Burkhard Monien}, {and} \bibinfo{person}{Karsten
  Tiemann}.} \bibinfo{year}{2005}\natexlab{}.
\newblock \showarticletitle{Nash {Equilibria}, the {Price} of {Anarchy} and the
  {Fully} {Mixed} {Nash} {Equilibrium} {Conjecture}}. In
  \bibinfo{booktitle}{\emph{Automata, {Languages} and {Programming}}},
  \bibfield{editor}{\bibinfo{person}{Luís Caires},
  \bibinfo{person}{Giuseppe~F. Italiano}, \bibinfo{person}{Luís Monteiro},
  \bibinfo{person}{Catuscia Palamidessi}, {and} \bibinfo{person}{Moti Yung}}
  (Eds.). \bibinfo{publisher}{Springer Berlin Heidelberg},
  \bibinfo{address}{Berlin, Heidelberg}, \bibinfo{pages}{51--65}.
\newblock
\showISBNx{978-3-540-31691-6}
\urldef\tempurl%
\url{https://link.springer.com/chapter/10.1007/11523468_5}
\showURL{%
\tempurl}


\bibitem[Gonzalez et~al\mbox{.}(2012)]%
        {powergraph}
\bibfield{author}{\bibinfo{person}{Joseph~E. Gonzalez},
  \bibinfo{person}{Yucheng Low}, \bibinfo{person}{Haijie Gu},
  \bibinfo{person}{Danny Bickson}, {and} \bibinfo{person}{Carlos Guestrin}.}
  \bibinfo{year}{2012}\natexlab{}.
\newblock \showarticletitle{PowerGraph: Distributed Graph-Parallel Computation
  on Natural Graphs}. In \bibinfo{booktitle}{\emph{OSDI}}.
  \bibinfo{publisher}{USENIX}, \bibinfo{pages}{17--30}.
\newblock
\urldef\tempurl%
\url{https://www.usenix.org/conference/osdi12/technical-sessions/presentation/gonzalez}
\showURL{%
\tempurl}


\bibitem[Gonzalez et~al\mbox{.}(2014)]%
        {graphx}
\bibfield{author}{\bibinfo{person}{Joseph~E Gonzalez},
  \bibinfo{person}{Reynold~S Xin}, \bibinfo{person}{Ankur Dave},
  \bibinfo{person}{Daniel Crankshaw}, \bibinfo{person}{Michael~J Franklin},
  {and} \bibinfo{person}{Ion Stoica}.} \bibinfo{year}{2014}\natexlab{}.
\newblock \showarticletitle{$\{$GraphX$\}$: Graph processing in a distributed
  dataflow framework}. In \bibinfo{booktitle}{\emph{OSDI}}.
  \bibinfo{pages}{599--613}.
\newblock
\urldef\tempurl%
\url{https://www.usenix.org/conference/osdi14/technical-sessions/presentation/gonzalez}
\showURL{%
\tempurl}


\bibitem[Gou et~al\mbox{.}(2019)]%
        {GSS}
\bibfield{author}{\bibinfo{person}{Xiangyang Gou}, \bibinfo{person}{Lei Zou},
  \bibinfo{person}{Chenxingyu Zhao}, {and} \bibinfo{person}{Tong Yang}.}
  \bibinfo{year}{2019}\natexlab{}.
\newblock \showarticletitle{Fast and Accurate Graph Stream Summarization}. In
  \bibinfo{booktitle}{\emph{ICDE}}. \bibinfo{publisher}{IEEE},
  \bibinfo{pages}{1118--1129}.
\newblock
\urldef\tempurl%
\url{https://doi.org/10.1109/ICDE.2019.00103}
\showDOI{\tempurl}


\bibitem[Hanai et~al\mbox{.}(2019)]%
        {dne}
\bibfield{author}{\bibinfo{person}{Masatoshi Hanai}, \bibinfo{person}{Toyotaro
  Suzumura}, \bibinfo{person}{Wen~Jun Tan}, \bibinfo{person}{Elvis~S. Liu},
  \bibinfo{person}{Georgios Theodoropoulos}, {and} \bibinfo{person}{Wentong
  Cai}.} \bibinfo{year}{2019}\natexlab{}.
\newblock \showarticletitle{Distributed Edge Partitioning for Trillion-edge
  Graphs}.
\newblock \bibinfo{journal}{\emph{VLDB}} \bibinfo{volume}{12},
  \bibinfo{number}{13} (\bibinfo{year}{2019}), \bibinfo{pages}{2379--2392}.
\newblock
\urldef\tempurl%
\url{https://doi.org/10.48550/arXiv.1908.05855}
\showURL{%
\tempurl}


\bibitem[Hollocou et~al\mbox{.}(2017)]%
        {holl}
\bibfield{author}{\bibinfo{person}{Alexandre Hollocou}, \bibinfo{person}{Julien
  Maudet}, \bibinfo{person}{Thomas Bonald}, {and} \bibinfo{person}{Marc
  Lelarge}.} \bibinfo{year}{2017}\natexlab{}.
\newblock \showarticletitle{A Streaming Algorithm for Graph Clustering}.
\newblock \bibinfo{journal}{\emph{CoRR}}  \bibinfo{volume}{abs/1712.04337}
  (\bibinfo{year}{2017}).
\newblock
\urldef\tempurl%
\url{https://doi.org/10.48550/arXiv.1712.04337}
\showURL{%
\tempurl}


\bibitem[Holyer(1981)]%
        {nphard}
\bibfield{author}{\bibinfo{person}{Ian Holyer}.}
  \bibinfo{year}{1981}\natexlab{}.
\newblock \showarticletitle{The NP-Completeness of Some Edge-Partition
  Problems}.
\newblock \bibinfo{journal}{\emph{SIAM J. Comput.}} \bibinfo{volume}{10},
  \bibinfo{number}{4} (\bibinfo{year}{1981}), \bibinfo{pages}{713--717}.
\newblock
\urldef\tempurl%
\url{https://doi.org/10.1137/0210054}
\showDOI{\tempurl}


\bibitem[Hua et~al\mbox{.}(2019)]%
        {cusp}
\bibfield{author}{\bibinfo{person}{Qiang-Sheng Hua}, \bibinfo{person}{Yangyang
  Li}, \bibinfo{person}{Dongxiao Yu}, {and} \bibinfo{person}{Hai Jin}.}
  \bibinfo{year}{2019}\natexlab{}.
\newblock \showarticletitle{Quasi-streaming graph partitioning: A game
  theoretical approach}.
\newblock \bibinfo{journal}{\emph{TPDS}} \bibinfo{volume}{30},
  \bibinfo{number}{7} (\bibinfo{year}{2019}), \bibinfo{pages}{1643--1656}.
\newblock
\urldef\tempurl%
\url{https://doi.org/10.1109/TPDS.2018.2890515}
\showDOI{\tempurl}


\bibitem[Jain et~al\mbox{.}(2013)]%
        {grid}
\bibfield{author}{\bibinfo{person}{Nilesh Jain}, \bibinfo{person}{Guangdeng
  Liao}, {and} \bibinfo{person}{Theodore~L. Willke}.}
  \bibinfo{year}{2013}\natexlab{}.
\newblock \showarticletitle{GraphBuilder: Scalable Graph ETL Framework}.
  \bibinfo{publisher}{ACM}.
\newblock
\showISBNx{9781450321884}
\urldef\tempurl%
\url{https://doi.org/10.1145/2484425.2484429}
\showDOI{\tempurl}


\bibitem[Jiang et~al\mbox{.}(2023)]%
        {graphsummary}
\bibfield{author}{\bibinfo{person}{Zhiguo Jiang}, \bibinfo{person}{Hanhua
  Chen}, {and} \bibinfo{person}{Hai Jin}.} \bibinfo{year}{2023}\natexlab{}.
\newblock \showarticletitle{Auxo: {A} Scalable and Efficient Graph Stream
  Summarization Structure}.
\newblock \bibinfo{journal}{\emph{VLDB}} \bibinfo{volume}{16},
  \bibinfo{number}{6} (\bibinfo{year}{2023}), \bibinfo{pages}{1386--1398}.
\newblock
\urldef\tempurl%
\url{https://doi.org/10.14778/3583140.3583154}
\showDOI{\tempurl}


\bibitem[Karypis and Kumar(1998)]%
        {metis}
\bibfield{author}{\bibinfo{person}{George Karypis} {and} \bibinfo{person}{Vipin
  Kumar}.} \bibinfo{year}{1998}\natexlab{}.
\newblock \showarticletitle{A fast and high quality multilevel scheme for
  partitioning irregular graphs}.
\newblock \bibinfo{journal}{\emph{SIAM Journal on scientific Computing}}
  \bibinfo{volume}{20}, \bibinfo{number}{1} (\bibinfo{year}{1998}),
  \bibinfo{pages}{359--392}.
\newblock
\urldef\tempurl%
\url{https://doi.org/10.1137/S1064827595287997}
\showDOI{\tempurl}


\bibitem[Kong et~al\mbox{.}(2022)]%
        {clugp}
\bibfield{author}{\bibinfo{person}{Deyu Kong}, \bibinfo{person}{Xike Xie},
  {and} \bibinfo{person}{Zhuoxu Zhang}.} \bibinfo{year}{2022}\natexlab{}.
\newblock \showarticletitle{Clustering-based Partitioning for Large Web
  Graphs}. In \bibinfo{booktitle}{\emph{ICDE}}. \bibinfo{publisher}{IEEE},
  \bibinfo{pages}{593--606}.
\newblock
\urldef\tempurl%
\url{https://doi.org/10.1109/ICDE53745.2022.00049}
\showDOI{\tempurl}


\bibitem[Kwak et~al\mbox{.}(2010)]%
        {tw}
\bibfield{author}{\bibinfo{person}{Haewoon Kwak}, \bibinfo{person}{Changhyun
  Lee}, \bibinfo{person}{Hosung Park}, {and} \bibinfo{person}{Sue Moon}.}
  \bibinfo{year}{2010}\natexlab{}.
\newblock \showarticletitle{What is Twitter, a social network or a news
  media?}. In \bibinfo{booktitle}{\emph{WWW}}. \bibinfo{pages}{591--600}.
\newblock
\urldef\tempurl%
\url{https://doi.org/10.1145/1772690.1772751}
\showDOI{\tempurl}


\bibitem[Ledenev(2023)]%
        {PUMBA}
\bibfield{author}{\bibinfo{person}{Alexei Ledenev}.}
  \bibinfo{year}{2023}\natexlab{}.
\newblock \bibinfo{title}{Pumba: chaos testing tool for Docker}.
\newblock
\newblock
\urldef\tempurl%
\url{https://github.com/alexei-led/pumba}
\showURL{%
\tempurl}


\bibitem[Leskovec and Krevl(2014)]%
        {snap}
\bibfield{author}{\bibinfo{person}{Jure Leskovec} {and} \bibinfo{person}{Andrej
  Krevl}.} \bibinfo{year}{2014}\natexlab{}.
\newblock \bibinfo{title}{SNAP Datasets: Stanford Large Network Dataset
  Collection}.
\newblock \bibinfo{howpublished}{\url{http://snap.stanford.edu/data}}.
\newblock


\bibitem[Liu et~al\mbox{.}(2021)]%
        {tailgnn}
\bibfield{author}{\bibinfo{person}{Zemin Liu}, \bibinfo{person}{Trung{-}Kien
  Nguyen}, {and} \bibinfo{person}{Yuan Fang}.} \bibinfo{year}{2021}\natexlab{}.
\newblock \showarticletitle{Tail-GNN: Tail-Node Graph Neural Networks}. In
  \bibinfo{booktitle}{\emph{SIGKDD}}. \bibinfo{publisher}{ACM},
  \bibinfo{pages}{1109--1119}.
\newblock
\urldef\tempurl%
\url{https://doi.org/10.1145/3447548.3467276}
\showDOI{\tempurl}


\bibitem[Low et~al\mbox{.}(2012)]%
        {graphlab}
\bibfield{author}{\bibinfo{person}{Yucheng Low}, \bibinfo{person}{Joseph
  Gonzalez}, \bibinfo{person}{Aapo Kyrola}, \bibinfo{person}{Danny Bickson},
  \bibinfo{person}{Carlos Guestrin}, {and} \bibinfo{person}{Joseph~M.
  Hellerstein}.} \bibinfo{year}{2012}\natexlab{}.
\newblock \showarticletitle{Distributed GraphLab: {A} Framework for Machine
  Learning in the Cloud}.
\newblock \bibinfo{journal}{\emph{VLDB}} \bibinfo{volume}{5},
  \bibinfo{number}{8} (\bibinfo{year}{2012}), \bibinfo{pages}{716--727}.
\newblock
\urldef\tempurl%
\url{https://doi.org/10.14778/2212351.2212354}
\showDOI{\tempurl}


\bibitem[Malewicz et~al\mbox{.}(2010)]%
        {pregel}
\bibfield{author}{\bibinfo{person}{Grzegorz Malewicz},
  \bibinfo{person}{Matthew~H. Austern}, \bibinfo{person}{Aart J.~C. Bik},
  \bibinfo{person}{James~C. Dehnert}, \bibinfo{person}{Ilan Horn},
  \bibinfo{person}{Naty Leiser}, {and} \bibinfo{person}{Grzegorz Czajkowski}.}
  \bibinfo{year}{2010}\natexlab{}.
\newblock \showarticletitle{Pregel: a system for large-scale graph processing}.
  In \bibinfo{booktitle}{\emph{SIGMOD}}. \bibinfo{publisher}{ACM},
  \bibinfo{pages}{135--146}.
\newblock
\urldef\tempurl%
\url{https://doi.org/10.1145/1807167.1807184}
\showDOI{\tempurl}


\bibitem[Mayer et~al\mbox{.}(2018)]%
        {adwise}
\bibfield{author}{\bibinfo{person}{Christian Mayer}, \bibinfo{person}{Ruben
  Mayer}, \bibinfo{person}{Muhammad~Adnan Tariq}, \bibinfo{person}{Heiko
  Geppert}, \bibinfo{person}{Larissa Laich}, \bibinfo{person}{Lukas Rieger},
  {and} \bibinfo{person}{Kurt Rothermel}.} \bibinfo{year}{2018}\natexlab{}.
\newblock \showarticletitle{{ADWISE:} Adaptive Window-Based Streaming Edge
  Partitioning for High-Speed Graph Processing}. In
  \bibinfo{booktitle}{\emph{ICDCS}}. \bibinfo{publisher}{IEEE},
  \bibinfo{pages}{685--695}.
\newblock
\urldef\tempurl%
\url{https://doi.org/10.1109/ICDCS.2018.00072}
\showDOI{\tempurl}


\bibitem[Mayer and Jacobsen(2021)]%
        {hep}
\bibfield{author}{\bibinfo{person}{Ruben Mayer} {and}
  \bibinfo{person}{Hans{-}Arno Jacobsen}.} \bibinfo{year}{2021}\natexlab{}.
\newblock \showarticletitle{Hybrid Edge Partitioner: Partitioning Large
  Power-Law Graphs under Memory Constraints}. In
  \bibinfo{booktitle}{\emph{SIGMOD}}. \bibinfo{publisher}{ACM},
  \bibinfo{pages}{1289--1302}.
\newblock
\urldef\tempurl%
\url{https://doi.org/10.1145/3448016.3457300}
\showDOI{\tempurl}


\bibitem[Mayer et~al\mbox{.}(2022)]%
        {2ps-l}
\bibfield{author}{\bibinfo{person}{Ruben Mayer}, \bibinfo{person}{Kamil
  Orujzade}, {and} \bibinfo{person}{Hans{-}Arno Jacobsen}.}
  \bibinfo{year}{2022}\natexlab{}.
\newblock \showarticletitle{Out-of-Core Edge Partitioning at Linear Run-Time}.
  In \bibinfo{booktitle}{\emph{ICDE}}. \bibinfo{publisher}{IEEE},
  \bibinfo{pages}{2629--2642}.
\newblock
\urldef\tempurl%
\url{https://doi.org/10.1109/ICDE53745.2022.00242}
\showDOI{\tempurl}


\bibitem[Md et~al\mbox{.}(2021)]%
        {distgcn}
\bibfield{author}{\bibinfo{person}{Vasimuddin Md}, \bibinfo{person}{Sanchit
  Misra}, \bibinfo{person}{Guixiang Ma}, \bibinfo{person}{Ramanarayan Mohanty},
  \bibinfo{person}{Evangelos Georganas}, \bibinfo{person}{Alexander Heinecke},
  \bibinfo{person}{Dhiraj Kalamkar}, \bibinfo{person}{Nesreen~K. Ahmed}, {and}
  \bibinfo{person}{Sasikanth Avancha}.} \bibinfo{year}{2021}\natexlab{}.
\newblock \showarticletitle{DistGNN: Scalable Distributed Training for
  Large-Scale Graph Neural Networks}. In \bibinfo{booktitle}{\emph{SC}}.
  \bibinfo{publisher}{ACM}, \bibinfo{address}{New York, NY, USA}, Article
  \bibinfo{articleno}{76}, \bibinfo{numpages}{14}~pages.
\newblock
\showISBNx{9781450384421}
\urldef\tempurl%
\url{https://doi.org/10.1145/3458817.3480856}
\showDOI{\tempurl}


\bibitem[Merkel et~al\mbox{.}(2023)]%
        {MerkelMFJ23}
\bibfield{author}{\bibinfo{person}{Nikolai Merkel}, \bibinfo{person}{Ruben
  Mayer}, \bibinfo{person}{Tawkir~Ahmed Fakir}, {and}
  \bibinfo{person}{Hans{-}Arno Jacobsen}.} \bibinfo{year}{2023}\natexlab{}.
\newblock \showarticletitle{Partitioner Selection with {EASE} to Optimize
  Distributed Graph Processing}. \bibinfo{publisher}{IEEE},
  \bibinfo{pages}{2400--2414}.
\newblock
\urldef\tempurl%
\url{https://doi.org/10.1109/ICDE55515.2023.00185}
\showDOI{\tempurl}


\bibitem[Newman(2005)]%
        {newman2005power}
\bibfield{author}{\bibinfo{person}{Mark~EJ Newman}.}
  \bibinfo{year}{2005}\natexlab{}.
\newblock \showarticletitle{Power laws, Pareto distributions and Zipf's law}.
\newblock \bibinfo{journal}{\emph{Contemporary physics}} \bibinfo{volume}{46},
  \bibinfo{number}{5} (\bibinfo{year}{2005}), \bibinfo{pages}{323--351}.
\newblock
\urldef\tempurl%
\url{https://www.tandfonline.com/doi/abs/10.1080/00107510500052444}
\showURL{%
\tempurl}


\bibitem[Park and Kim(2017)]%
        {rmat}
\bibfield{author}{\bibinfo{person}{Himchan Park} {and} \bibinfo{person}{Min-Soo
  Kim}.} \bibinfo{year}{2017}\natexlab{}.
\newblock \showarticletitle{TrillionG: A Trillion-Scale Synthetic Graph
  Generator Using a Recursive Vector Model}. In
  \bibinfo{booktitle}{\emph{SIGMOD}}. \bibinfo{publisher}{ACM},
  \bibinfo{address}{New York, NY, USA}, \bibinfo{pages}{913–928}.
\newblock
\showISBNx{9781450341974}
\urldef\tempurl%
\url{https://doi.org/10.1145/3035918.3064014}
\showDOI{\tempurl}


\bibitem[Petroni et~al\mbox{.}(2015)]%
        {hdrf}
\bibfield{author}{\bibinfo{person}{Fabio Petroni}, \bibinfo{person}{Leonardo
  Querzoni}, \bibinfo{person}{Khuzaima Daudjee}, \bibinfo{person}{Shahin
  Kamali}, {and} \bibinfo{person}{Giorgio Iacoboni}.}
  \bibinfo{year}{2015}\natexlab{}.
\newblock \showarticletitle{{HDRF:} Stream-Based Partitioning for Power-Law
  Graphs}. In \bibinfo{booktitle}{\emph{CIKM}}. \bibinfo{publisher}{ACM},
  \bibinfo{pages}{243--252}.
\newblock
\urldef\tempurl%
\url{https://doi.org/10.1145/2806416.2806424}
\showDOI{\tempurl}


\bibitem[Qu et~al\mbox{.}(2023)]%
        {hcpd}
\bibfield{author}{\bibinfo{person}{Wenwen Qu}, \bibinfo{person}{Weixi Zhang},
  \bibinfo{person}{Ji Cheng}, \bibinfo{person}{Chaorui Zhang},
  \bibinfo{person}{Wei Han}, \bibinfo{person}{Bo Bai},
  \bibinfo{person}{Chen~Jason Zhang}, \bibinfo{person}{Liang He}, {and}
  \bibinfo{person}{Xiaoling Wang}.} \bibinfo{year}{2023}\natexlab{}.
\newblock \showarticletitle{Optimizing Graph Partition by Optimal Vertex-Cut:
  {A} Holistic Approach}. In \bibinfo{booktitle}{\emph{ICDE}}.
  \bibinfo{publisher}{IEEE}, \bibinfo{pages}{1019--1031}.
\newblock
\urldef\tempurl%
\url{https://doi.org/10.1109/ICDE55515.2023.00083}
\showDOI{\tempurl}


\bibitem[Sun et~al\mbox{.}(2022)]%
        {mdsgp}
\bibfield{author}{\bibinfo{person}{Zhipeng Sun}, \bibinfo{person}{Guosun Zeng},
  \bibinfo{person}{Chunling Ding}, {and} \bibinfo{person}{Tengteng Cheng}.}
  \bibinfo{year}{2022}\natexlab{}.
\newblock \showarticletitle{A Streaming Graph Partitioning Method to Achieve
  High Cohesion and Equilibrium via Multiplayer Repeated Game}.
\newblock \bibinfo{journal}{\emph{TCSS}} (\bibinfo{year}{2022}),
  \bibinfo{pages}{1--12}.
\newblock
\urldef\tempurl%
\url{https://doi.org/10.1109/TCSS.2022.3226230}
\showDOI{\tempurl}


\bibitem[Taimouri and Saadatfar(2019)]%
        {TaimouriS19}
\bibfield{author}{\bibinfo{person}{Monireh Taimouri} {and}
  \bibinfo{person}{Hamid Saadatfar}.} \bibinfo{year}{2019}\natexlab{}.
\newblock \showarticletitle{{RBSEP:} a reassignment and buffer based streaming
  edge partitioning approach}.
\newblock \bibinfo{journal}{\emph{J. Big Data}}  \bibinfo{volume}{6}
  (\bibinfo{year}{2019}), \bibinfo{pages}{92}.
\newblock
\urldef\tempurl%
\url{https://doi.org/10.1186/s40537-019-0257-5}
\showDOI{\tempurl}


\bibitem[Tang et~al\mbox{.}(2016)]%
        {tcm}
\bibfield{author}{\bibinfo{person}{Nan Tang}, \bibinfo{person}{Qing Chen},
  {and} \bibinfo{person}{Prasenjit Mitra}.} \bibinfo{year}{2016}\natexlab{}.
\newblock \showarticletitle{Graph Stream Summarization: From Big Bang to Big
  Crunch}. In \bibinfo{booktitle}{\emph{SIGMOD}},
  \bibfield{editor}{\bibinfo{person}{Fatma {\"{O}}zcan},
  \bibinfo{person}{Georgia Koutrika}, {and} \bibinfo{person}{Sam Madden}}
  (Eds.). \bibinfo{publisher}{ACM}, \bibinfo{pages}{1481--1496}.
\newblock
\urldef\tempurl%
\url{https://doi.org/10.1145/2882903.2915223}
\showDOI{\tempurl}


\bibitem[Wan et~al\mbox{.}(2022)]%
        {WanLLKL22}
\bibfield{author}{\bibinfo{person}{Cheng Wan}, \bibinfo{person}{Youjie Li},
  \bibinfo{person}{Ang Li}, \bibinfo{person}{Nam~Sung Kim}, {and}
  \bibinfo{person}{Yingyan Lin}.} \bibinfo{year}{2022}\natexlab{}.
\newblock \showarticletitle{{BNS-GCN:} Efficient Full-Graph Training of Graph
  Convolutional Networks with Partition-Parallelism and Random Boundary Node
  Sampling}. In \bibinfo{booktitle}{\emph{MLSys}}.
  \bibinfo{publisher}{mlsys.org}.
\newblock
\urldef\tempurl%
\url{https://arxiv.org/abs/2203.10983}
\showURL{%
\tempurl}


\bibitem[Wang et~al\mbox{.}(2023)]%
        {ScaleG}
\bibfield{author}{\bibinfo{person}{Xubo Wang}, \bibinfo{person}{Dong Wen},
  \bibinfo{person}{Lu Qin}, \bibinfo{person}{Lijun Chang},
  \bibinfo{person}{Ying Zhang}, {and} \bibinfo{person}{Wenjie Zhang}.}
  \bibinfo{year}{2023}\natexlab{}.
\newblock \showarticletitle{ScaleG: {A} Distributed Disk-Based System for
  Vertex-Centric Graph Processing}.
\newblock \bibinfo{journal}{\emph{IEEE TKDE}} \bibinfo{volume}{35},
  \bibinfo{number}{2} (\bibinfo{year}{2023}), \bibinfo{pages}{2019--2033}.
\newblock
\urldef\tempurl%
\url{https://doi.org/10.1109/TKDE.2021.3101057}
\showDOI{\tempurl}


\bibitem[Xie et~al\mbox{.}(2014)]%
        {dbh}
\bibfield{author}{\bibinfo{person}{Cong Xie}, \bibinfo{person}{Ling Yan},
  \bibinfo{person}{Wu{-}Jun Li}, {and} \bibinfo{person}{Zhihua Zhang}.}
  \bibinfo{year}{2014}\natexlab{}.
\newblock \showarticletitle{Distributed Power-law Graph Computing: Theoretical
  and Empirical Analysis}. In \bibinfo{booktitle}{\emph{NeurIPS}}.
  \bibinfo{pages}{1673--1681}.
\newblock
\urldef\tempurl%
\url{https://dl.acm.org/doi/10.5555/2968826.2969013}
\showURL{%
\tempurl}


\bibitem[Yan et~al\mbox{.}(2014)]%
        {blogel}
\bibfield{author}{\bibinfo{person}{Da Yan}, \bibinfo{person}{James Cheng},
  \bibinfo{person}{Yi Lu}, {and} \bibinfo{person}{Wilfred Ng}.}
  \bibinfo{year}{2014}\natexlab{}.
\newblock \showarticletitle{Blogel: {A} Block-Centric Framework for Distributed
  Computation on Real-World Graphs}.
\newblock \bibinfo{journal}{\emph{VLDB}} \bibinfo{volume}{7},
  \bibinfo{number}{14} (\bibinfo{year}{2014}), \bibinfo{pages}{1981--1992}.
\newblock
\urldef\tempurl%
\url{https://doi.org/10.14778/2733085.2733103}
\showDOI{\tempurl}


\bibitem[Yan et~al\mbox{.}(2020)]%
        {Gthinker}
\bibfield{author}{\bibinfo{person}{Da Yan}, \bibinfo{person}{Guimu Guo},
  \bibinfo{person}{Md~Mashiur~Rahman Chowdhury}, \bibinfo{person}{M.~Tamer
  {\"{O}}zsu}, \bibinfo{person}{Wei{-}Shinn Ku}, {and} \bibinfo{person}{John
  C.~S. Lui}.} \bibinfo{year}{2020}\natexlab{}.
\newblock \showarticletitle{G-thinker: {A} Distributed Framework for Mining
  Subgraphs in a Big Graph}. In \bibinfo{booktitle}{\emph{ICDE}}.
  \bibinfo{publisher}{IEEE}, \bibinfo{pages}{1369--1380}.
\newblock
\urldef\tempurl%
\url{https://doi.org/10.1109/ICDE48307.2020.00122}
\showDOI{\tempurl}


\bibitem[Yang and Leskovec(2012)]%
        {fr}
\bibfield{author}{\bibinfo{person}{Jaewon Yang} {and} \bibinfo{person}{Jure
  Leskovec}.} \bibinfo{year}{2012}\natexlab{}.
\newblock \showarticletitle{Defining and evaluating network communities based
  on ground-truth}. In \bibinfo{booktitle}{\emph{ACM SIGKDD Workshop}}.
  \bibinfo{pages}{1--8}.
\newblock
\urldef\tempurl%
\url{https://doi.org/10.1145/2350190.2350193}
\showDOI{\tempurl}


\bibitem[Yang et~al\mbox{.}(2023)]%
        {yang2023betty}
\bibfield{author}{\bibinfo{person}{Shuangyan Yang}, \bibinfo{person}{Minjia
  Zhang}, \bibinfo{person}{Wenqian Dong}, {and} \bibinfo{person}{Dong Li}.}
  \bibinfo{year}{2023}\natexlab{}.
\newblock \showarticletitle{Betty: Enabling Large-Scale GNN Training with
  Batch-Level Graph Partitioning}. In \bibinfo{booktitle}{\emph{Proceedings of
  the 28th ACM International Conference on Architectural Support for
  Programming Languages and Operating Systems, Volume 2}}.
  \bibinfo{pages}{103--117}.
\newblock
\urldef\tempurl%
\url{https://dl.acm.org/doi/abs/10.1145/3575693.3575725}
\showURL{%
\tempurl}


\bibitem[Zhang et~al\mbox{.}(2017)]%
        {ne}
\bibfield{author}{\bibinfo{person}{Chenzi Zhang}, \bibinfo{person}{Fan Wei},
  \bibinfo{person}{Qin Liu}, \bibinfo{person}{Zhihao~Gavin Tang}, {and}
  \bibinfo{person}{Zhenguo Li}.} \bibinfo{year}{2017}\natexlab{}.
\newblock \showarticletitle{Graph Edge Partitioning via Neighborhood
  Heuristic}. In \bibinfo{booktitle}{\emph{SIGKDD}}. \bibinfo{publisher}{ACM},
  \bibinfo{pages}{605--614}.
\newblock
\urldef\tempurl%
\url{https://doi.org/10.1145/3097983.3098033}
\showDOI{\tempurl}


\bibitem[Zhang et~al\mbox{.}(2020)]%
        {TopoX}
\bibfield{author}{\bibinfo{person}{Yiming Zhang}, \bibinfo{person}{Haonan
  Wang}, \bibinfo{person}{Menghan Jia}, \bibinfo{person}{Jinyan Wang},
  \bibinfo{person}{Dong~Sheng Li}, \bibinfo{person}{Guangtao Xue}, {and}
  \bibinfo{person}{Kian{-}Lee Tan}.} \bibinfo{year}{2020}\natexlab{}.
\newblock \showarticletitle{TopoX: Topology Refactorization for Minimizing
  Network Communication in Graph Computations}.
\newblock \bibinfo{journal}{\emph{IEEE ToN}} \bibinfo{volume}{28},
  \bibinfo{number}{6} (\bibinfo{year}{2020}), \bibinfo{pages}{2768--2782}.
\newblock
\urldef\tempurl%
\url{https://doi.org/10.1109/TNET.2020.3020813}
\showDOI{\tempurl}


\bibitem[Zhu et~al\mbox{.}(2019)]%
        {ZhuZYLZALZ19}
\bibfield{author}{\bibinfo{person}{Rong Zhu}, \bibinfo{person}{Kun Zhao},
  \bibinfo{person}{Hongxia Yang}, \bibinfo{person}{Wei Lin},
  \bibinfo{person}{Chang Zhou}, \bibinfo{person}{Baole Ai},
  \bibinfo{person}{Yong Li}, {and} \bibinfo{person}{Jingren Zhou}.}
  \bibinfo{year}{2019}\natexlab{}.
\newblock \showarticletitle{AliGraph: {A} Comprehensive Graph Neural Network
  Platform}.
\newblock \bibinfo{journal}{\emph{VLDB}} \bibinfo{volume}{12},
  \bibinfo{number}{12} (\bibinfo{year}{2019}), \bibinfo{pages}{2094--2105}.
\newblock
\urldef\tempurl%
\url{https://doi.org/10.14778/3352063.3352127}
\showDOI{\tempurl}


\bibitem[Zhu et~al\mbox{.}(2015)]%
        {ZhuHC15}
\bibfield{author}{\bibinfo{person}{Xiaowei Zhu}, \bibinfo{person}{Wentao Han},
  {and} \bibinfo{person}{Wenguang Chen}.} \bibinfo{year}{2015}\natexlab{}.
\newblock \showarticletitle{GridGraph: Large-Scale Graph Processing on a Single
  Machine Using 2-Level Hierarchical Partitioning}. In
  \bibinfo{booktitle}{\emph{ATC}}, \bibfield{editor}{\bibinfo{person}{Shan Lu}
  {and} \bibinfo{person}{Erik Riedel}} (Eds.). \bibinfo{publisher}{USENIX},
  \bibinfo{pages}{375--386}.
\newblock
\urldef\tempurl%
\url{https://dl.acm.org/doi/10.5555/2813767.2813795}
\showURL{%
\tempurl}


\bibitem[Zwolak et~al\mbox{.}(2022)]%
        {gcnsplit}
\bibfield{author}{\bibinfo{person}{Micha\l{} Zwolak}, \bibinfo{person}{Zainab
  Abbas}, \bibinfo{person}{Sonia Horchidan}, \bibinfo{person}{Paris Carbone},
  {and} \bibinfo{person}{Vasiliki Kalavri}.} \bibinfo{year}{2022}\natexlab{}.
\newblock \showarticletitle{GCNSplit: Bounding the State of Streaming Graph
  Partitioning}. \bibinfo{publisher}{ACM}, \bibinfo{address}{New York, NY,
  USA}, Article \bibinfo{articleno}{3}, \bibinfo{numpages}{12}~pages.
\newblock
\showISBNx{9781450393775}
\urldef\tempurl%
\url{https://doi.org/10.1145/3533702.3534920}
\showDOI{\tempurl}


\end{thebibliography}

\end{document}